%% file: main.tex
\documentclass[a4paper,10pt]{article}
\usepackage{fullpage}
\usepackage[parfill]{parskip} 

\input{preamble}

\title{A Trichotomy in the Data Complexity of Certain Query Answering for Conjunctive Queries}
\date{}

\begin{document}

\author[1]{Paraschos Koutris}
\author[2]{Jef Wijsen}
\affil[1]{University of Washington, Seattle, USA}
\affil[2]{University of Mons, Belgium}


\maketitle


\begin{abstract}
A relational database is said to be uncertain if primary key constraints can possibly be violated. 
A repair (or possible world) of an uncertain database is obtained by selecting a maximal number of tuples without ever selecting two distinct tuples with the same primary key value.  
For any Boolean query $q$, $\cqa{q}$ is the problem that takes an uncertain database $\db$ on input, and asks whether $q$ is true in every repair of $\db$. 
The complexity of this problem has been particularly studied for $q$ ranging over the class of self-join-free Boolean conjunctive queries. 
A research challenge is to determine, given $q$, whether $\cqa{q}$ belongs to complexity classes $\FO$, $\P$, or $\coNP$-complete.  
In this paper, we combine existing techniques for studying the above complexity classification task. 
We show that for any self-join-free Boolean conjunctive query $q$, it can be decided whether or not $\cqa{q}$ is in $\FO$.
Further, for any self-join-free Boolean conjunctive query $q$, $\cqa{q}$ is either in $\P$ or $\coNP$-complete, and the complexity dichotomy is effective.
This settles a research question that has been open for ten years, since~\cite{FUXMAN2005}.
\end{abstract}

\input{introduction}

\input{preliminaries}

\input{attackgraph}

\input{firstorder}

\input{intractability}

\input{tractability}

\input{conclusion}

\bibliographystyle{plain}
\bibliography{foiks}
\normalsize


\appendix

\input{proofs}

\input{tractability_extended}

\end{document}

%% file: preamble.tex
\usepackage{authblk}
\usepackage{latexsym}
\usepackage{amsmath}
\usepackage{amssymb}
\usepackage{times}
\usepackage{xcolor}
\usepackage{graphicx}
\usepackage[bold]{complexity}
\usepackage{stmaryrd}
\usepackage{arydshln}
\usepackage{bigstrut}

\newcommand{\calC}{{\mathcal{C}}}
\newcommand{\calE}{{\mathcal{E}}}
\newcommand{\card}[1]{\left|{#1}\right|}
\newcommand{\formula}[1]{\left({#1}\right)}
\newcommand{\tuple}[1]{\langle{#1}\rangle}
\newcommand{\defeq}{\mathrel{\mathop:}=}
\newcommand{\signature}[2]{[{#1},{#2}]}
\newcommand{\db}{{\mathbf{db}}}
\newcommand{\rep}{{\mathbf{r}}}
\newcommand{\sep}{{\mathbf{s}}}
\newcommand{\block}{{\mathbf{b}}}
\newcommand{\theblock}[2]{{\mathsf{block}}({#1},{#2})}
\newcommand{\gblock}{{\mathbf{g}}}   
\newcommand{\repairs}[1]{\mathsf{rset}({#1})}
\newcommand{\adom}[1]{{\mathbf{adom}}({#1})}
\newcommand{\atomvars}[1]{{\mathsf{vars}}({#1})}
\newcommand{\keyvars}[1]{{\mathsf{key}}({#1})}
\newcommand{\queryvars}[1]{\mathsf{vars}({#1})}
\newcommand{\sequencevars}[1]{\mathsf{vars}({#1})}
\newcommand{\substitute}[3]{{#1}_{[{{#2}\mapsto{#3}}]}}
\newcommand{\vmodels}[1]{\Vvdash_{#1}}
\newcommand{\fd}[2]{{#1}\rightarrow{#2}}
\newcommand{\FD}[1]{{\mathcal{K}}({#1})}
\newcommand{\keycl}[2]{{#1}^{+,{#2}}}
\newcommand{\keyclsup}[2]{{#1}^{\boxplus,{#2}}}
\newcommand{\cqa}[1]{{\mathsf{CERTAINTY}}({#1})}
\newcommand{\attacksymbol}[1]{\stackrel{#1}{\rightsquigarrow}}
\newcommand{\nattacksymbol}[1]{\stackrel{#1}{\not\rightsquigarrow}}
\newcommand{\attacks}[1]{\attacksymbol{#1}}
\newcommand{\nattacks}[1]{\nattacksymbol{#1}}

\newtheorem{lemma}{Lemma}
\newtheorem{theorem}{Theorem}

\newtheorem{proposition}{Proposition}

\newtheorem{sublemma}{Sublemma}
\newenvironment{proof}[1][]{{\bf Proof #1}\hspace{0.5ex}}{\hfill$\Box$\newline}
\newenvironment{subproof}[1][]{{\bf Proof #1}}{\hfill$\dashv$\newline}
\newtheorem{defi}{Definition}
\newenvironment{definition}{\begin{defi}\rm}{\hfill$\lhd$\end{defi}}
\newtheorem{exa}{Example}
\newenvironment{example}{\begin{exa}\rm}{\hfill$\lhd$\end{exa}}
\newtheorem{rem}{Remark}

\newcommand{\fuxman}{0}
\newcommand{\pema}{1}
\newcommand{\depth}[1]{d({#1})}
\newcommand{\bival}[2]{\Theta^{#1}_{#2}}
\newcommand{\trival}[3]{\Theta^{#1}_{#2,#3}}

\newcommand{\step}[1]{\stackrel{#1}{\smallfrown}}
\newcommand{\consistent}[1]{[\![{#1}]\!]}
\newcommand{\icard}[1]{{\mathsf{incnt}}({#1})}
\newcommand{\markov}{\stackrel{\mathsf{{}_{M}}}{\longrightarrow}}
\newcommand{\markovpath}{\stackrel{\mathsf{{}_{M}}*}{\longrightarrow}}
\newcommand{\markovpatharg}[1]{\stackrel{{#1},\mathsf{{}_{M}}*}{\longrightarrow}}
\newcommand{\markovarg}[1]{\stackrel{{#1},\mathsf{{}_{M}}}{\longrightarrow}}
\newcommand{\type}[1]{{\mathsf{type}}({#1})}
\newcommand{\fpro}[2]{\preceq_{#1}^{#2}}
\newcommand{\resolve}[2]{{\mathsf{dissolve}}({#1},{#2})}
\newcommand{\clutch}[2]{{\mathsf{C}}_{#2}({#1})}
\newcommand\MyLBrace[2]{%
  \left.\rule{0pt}{#1}\right\}{#2}}
\newcommand{\ucq}{\mathsf{UCQ}}
\newcommand{\dbgraph}[1]{{\mathcal{G}}({#1})}

%% file: introduction.tex
\section{Introduction}
\label{sec:intro}

Primary key violations provide an elementary means for capturing uncertainty in the relational data model.
A {\em block\/} is a maximal set of tuples of the same relation that agree on the primary key of the relation.
Tuples in the same block are mutually exclusive: exactly one tuple is true, but we are uncertain about which one.
We will refer to databases as ``uncertain databases" to stress that they can violate primary key constraints.

A {\em repair\/} (or possible world) of an uncertain database is obtained by selecting exactly one tuple from each block.
In general, the number of repairs of an uncertain database can be exponential in its size.
For instance, if an uncertain database contains $n$ blocks with two tuples each,
then it contains $2n$ tuples and has $2^{n}$ repairs.

There are two natural semantics for answering Boolean queries $q$ on an uncertain database. 
Under the {\em possibility semantics\/}, the question is whether the query evaluates to true on some repair.
Under the {\em certainty semantics\/}, which is adopted in this paper, the question is whether the query evaluates to true on every repair. 
The certainty semantics adheres to the paradigm of {\em consistent query answering\/}~\cite{DBLP:conf/pods/ArenasBC99,DBLP:series/synthesis/2011Bertossi}, 
which introduces the notion of database repairs with respect to general integrity constraints.
In this work, repairing is exclusively with respect to primary key constraints, one per relation.

For any Boolean query $q$, the decision problem $\cqa{q}$ is the following.
\begin{center}
\fbox{
\begin{tabular}{ll}
{\sc Problem:\/} & $\cqa{q}$\\
{\sc Input:\/} & uncertain database $\db$\\
{\sc Question:\/} & Does every repair of $\db$ satisfy $q$?  
\end{tabular}}
\end{center}
Three comments are in place.
First, the Boolean query $q$ is not part of the input.
Every Boolean query $q$ gives thus rise to a new problem.
Since the input to $\cqa{q}$ is an uncertain database, we consider the {\em data complexity\/} of the problem. 
Second, we will assume that every relation name in $q$ or $\db$ has a fixed known arity and primary key.
The primary key constraints are thus implicitly present in all problems.
Third, all the complexity results obtained in this paper can be carried over to non-Boolean queries;
the restriction to Boolean queries eases the technical treatment, but is not fundamental.

The complexity of $\cqa{q}$ has gained considerable research attention in recent years, especially for $q$ ranging over the set of self-join-free conjunctive queries. 
A challenging question is to distinguish queries $q$ for which the problem $\cqa{q}$ is tractable from queries for which the problem is intractable. 
Further, if $\cqa{q}$ is tractable, one may ask whether it is first-order expressible.
We will refer to these questions as the {\em complexity classification task of $\cqa{q}$\/}.

In the past decade, a variety of tools and techniques have been used in the complexity classification task of $\cqa{q}$ for self-join-free conjunctive queries $q$.
In their pioneering work, Fuxman and Miller~\cite{FUXMAN2005} introduced the notion of {\em join graph\/} (not to be confused with the classical notion of join tree).
Later on, Wijsen~\cite{DBLP:conf/pods/Wijsen10} introduced the notion of {\em attack graph\/}.
Kolaitis and Pema~\cite{DBLP:journals/ipl/KolaitisP12} applied Minty's algorithm~\cite{DBLP:journals/jct/Minty80} to the task.
Koutris and Suciu~\cite{DBLP:conf/icdt/KoutrisS14} introduced the notion of {\em query graph\/} and the distinction between consistent and possibly inconsistent relations.  
All these techniques have limited applicability:
join graphs seem too rudimentary to obtain general complexity dichotomies;
attack graphs enable to characterize first-order expressibility of $\cqa{q}$, but only for acyclic (in the sense of~\cite{DBLP:journals/jacm/BeeriFMY83}) queries $q$;
Minty's algorithm has been used to establish a $\P$-$\coNP$-complete dichotomy in the complexity of $\cqa{q}$, but only for queries $q$ with exactly two atoms;
the framework of Koutris and Suciu has also resulted in a $\P$-$\coNP$-complete dichotomy, but only when all primary keys consist of a single attribute.
On top of the limited applicability of each individual technique,
there is the difficulty that complexity classifications expressed in terms of different techniques cannot be easily compared.  

In this paper, we make significant progress in the complexity classification task of $\cqa{q}$ for $q$ ranging over the set of self-join-free conjunctive queries,
by establishing the following effective complexity trichotomy:
\begin{itemize}
\item
Given a self-join-free Boolean conjunctive query $q$,
it is decidable whether $\cqa{q}$ is in $\FO$.
In~\cite{DBLP:conf/pods/Wijsen10}, this was only shown  under the assumption that queries are acyclic (in the sense of~\cite{DBLP:journals/jacm/BeeriFMY83}).
\item
Given a self-join-free Boolean conjunctive query $q$,
if $\cqa{q}$ is not in $\FO$, then it is $\L$-hard.
In previous works~\cite{DBLP:conf/pods/Wijsen10,DBLP:journals/tods/Wijsen12}, Hanf locality was used to show first-order inexpressibility, resulting in involved proofs.
The current paper takes a complexity-theoretic approach to first-order inexpressibility, which results in an easier proof of a stronger result.
\item
For every self-join-free Boolean conjunctive query,
$\cqa{q}$ is either in $\P$ or $\coNP$-complete, and the dichotomy is effective.
In~\cite{DBLP:conf/icdt/KoutrisS14}, this was only shown under the assumption that all primary keys are simple (i.e., consist of a single attribute).
\end{itemize}
The established complexity trichotomy solves a problem that has been open since 2005~\cite{FUXMAN2005}.

\paragraph{Organization}
This paper is organized as follows.
Section~\ref{sec:relatedwork} discusses related work.
Section~\ref{sec:preliminaries} introduces our data and query model.
Section~\ref{sec:attackgraph} defines attack graphs for Boolean conjunctive queries, extending an older notion of attack graph~\cite{DBLP:journals/tods/Wijsen12} that was defined exclusively for acyclic Boolean conjunctive queries. 
The section also states the main result of the paper, Theorem~\ref{thm:main}.
Section~\ref{sec:firstorder} establishes an effective procedure that takes in a self-join-free Boolean conjunctive query $q$,
and decides whether $\cqa{q}$ is in $\FO$.
Section~\ref{sec:intractability} provides a sufficient condition for $\coNP$-hardness of $\cqa{q}$, for any self-join-free Boolean conjunctive query $q$.
Section~\ref{sec:road} shows that if the condition is not satisfied, then $\cqa{q}$ is in $\P$.
The appendix contains the proofs of some non-trivial results.

\section{Related Work}
\label{sec:relatedwork}

Consistent query answering (CQA) goes back to the seminal work by Arenas, Bertossi, and Chomicki~\cite{DBLP:conf/pods/ArenasBC99}.
Fuxman and Miller~\cite{FUXMAN2005} were the first ones to focus on CQA under the restrictions that
consistency is only with respect to primary keys and that queries are self-join-free conjunctive.
The term $\cqa{q}$ was coined in~\cite{DBLP:conf/pods/Wijsen10}. 
A recent and comprehensive survey on $\cqa{q}$ is~\cite{DBLP:conf/foiks/Wijsen14}.

Little is known about $\cqa{q}$ beyond self-join-free conjunctive queries.
An interesting recent result by Fontaine~\cite{DBLP:conf/lics/Fontaine13} goes as follows.
Let $\ucq$ be the class of Boolean first-order queries that can be expressed as disjunctions of Boolean conjunctive queries (possibly with constants and self-joins).
A daring conjecture is that for every query $q$ in $\ucq$, $\cqa{q}$ is either in $\P$ or $\coNP$-complete.
Fontaine showed that this conjecture implies Bulatov's dichotomy theorem for conservative CSP~\cite{DBLP:journals/tocl/Bulatov11},
the proof of which is highly involved (the full paper contains 66~pages).

%% file: preliminaries.tex
\section{Preliminaries}
\label{sec:preliminaries}

We assume disjoint sets of {\em variables\/} and {\em constants\/}.
If $\vec{x}$ is a sequence containing variables and constants, then $\sequencevars{\vec{x}}$ denotes the set of variables that occur in $\vec{x}$.
A {\em valuation\/} over a set $U$ of variables is a total mapping $\theta$ from $U$ to the set of constants.
At several places, it is implicitly understood that such a valuation $\theta$ is extended to be the identity on constants and on variables not in $U$.
If $V \subseteq U$, then $\theta[V]$ denotes the restriction of $\theta$ to $V$.

If $\theta$ is a valuation over a set $U$ of variables, $x$ is a variable, and $a$ is a constant,
then $\substitute{\theta}{x}{a}$ is the valuation over $U\cup\{x\}$ such that
$\substitute{\theta}{x}{a}(x)=a$ and for every variable $y$ such that $y\neq x$, $\substitute{\theta}{x}{a}(y)=\theta(y)$.
Notice that $x\in U$ is allowed.

\paragraph{Atoms and key-equal facts}
Each {\em relation name\/} $R$ of arity $n$, $n\geq 1$, has a unique {\em primary key\/} which is a set $\{1,2,\dots,k\}$ where $1\leq k\leq n$.
We say that $R$ has {\em signature\/} $\signature{n}{k}$ if $R$ has arity $n$ and primary key $\{1,2,\dots,k\}$. 
We say that $R$ is {\em simple-key\/} if $k=1$.
Elements of the primary key are called {\em primary-key positions\/},
while $k+1$, $k+2$, \dots, $n$ are {\em non-primary-key positions\/}. 
For all positive integers $n,k$ such that $1\leq k\leq n$, we assume denumerably many relation names with signature $\signature{n}{k}$.

If $R$ is a relation name with signature $\signature{n}{k}$, then $R(s_{1},\dots,s_{n})$ is called an {\em $R$-atom\/} (or simply atom), where each $s_{i}$ is either a constant or a variable ($1\leq i\leq n$).
Such an atom is commonly written as $R(\underline{\vec{x}},\vec{y})$ where the primary key value $\vec{x}=s_{1},\dots,s_{k}$ is underlined and $\vec{y}=s_{k+1},\dots,s_{n}$.
An {\em $R$-fact\/} (or simply fact) is an $R$-atom in which no variable occurs.
Two facts $R_{1}(\underline{\vec{a}_{1}},\vec{b}_{1}),R_{2}(\underline{\vec{a}_{2}},\vec{b}_{2})$ are {\em key-equal\/} if $R_{1}=R_{2}$ and $\vec{a}_{1}=\vec{a}_{2}$.
An $R$-atom or an $R$-fact is {\em simple-key\/} if $R$ is simple-key.

We will use letters $F,G,H$ for atoms.
For an atom $F=R(\underline{\vec{x}},\vec{y})$, we denote by $\keyvars{F}$ the set of variables that occur in $\vec{x}$,
and by $\atomvars{F}$ the set of variables that occur in $F$, that is, $\keyvars{F}=\sequencevars{\vec{x}}$ and $\atomvars{F}=\sequencevars{\vec{x}}\cup\sequencevars{\vec{y}}$.

\paragraph{Uncertain database, blocks, and repairs}
A {\em database schema\/} is a finite set of relation names.
All constructs that follow are defined relative to a fixed database schema.

An {\em uncertain database\/} is a finite set $\db$ of facts using only the relation names of the schema.
We refer to databases as ``uncertain databases" to stress that such databases can violate primary key constraints.

We write $\adom{\db}$ for the active domain of $\db$ (i.e., the set of constants that occur in $\db$).
A {\em block\/} of $\db$ is a maximal set of key-equal facts of $\db$.
The term $R$-block refers to a block of $R$-facts, i.e., facts with relation name $R$.
If $A$ is a fact of $\db$, then $\theblock{A}{\db}$ denotes the block of $\db$ that contains $A$.
An uncertain database $\db$ is {\em consistent\/} if no two distinct facts are key-equal 
(i.e., if every block of $\db$ is a singleton).
A {\em repair\/} of $\db$ is a maximal (with respect to set containment) consistent subset of $\db$. 
We write $\repairs{\db}$ for the set of repairs of $\db$.

\paragraph{Boolean conjunctive queries}
A {\em Boolean query\/} is a mapping $q$ that associates a Boolean (true or false) to each uncertain database, such that $q$ is closed under isomorphism~\cite{DBLP:books/sp/Libkin04}.
We write $\db\models q$ to denote that $q$ associates true to $\db$, in which case $\db$ is said to {\em satisfy\/} $q$.
A {\em Boolean first-order query\/} is a Boolean query that can be defined in first-order logic.
A {\em Boolean conjunctive query\/} is a finite set 
$q=\{R_{1}(\underline{\vec{x}_{1}},\vec{y}_{1})$, $\dots$, $R_{n}(\underline{\vec{x}_{n}},\vec{y}_{n})\}$ of atoms.
We denote by $\queryvars{q}$ the set of variables that occur in $q$.
The set $q$ represents the first-order sentence 
$$\exists u_{1}\dotsm\exists u_{k}\left(R_{1}(\underline{\vec{x}_{1}},\vec{y}_{1})\land\dotsm\land R_{n}(\underline{\vec{x}_{n}},\vec{y}_{n})\right),$$ where $\{u_{1}, \dots, u_{k}\}=\queryvars{q}$.
This query $q$ is satisfied by uncertain database $\db$ if there exists a valuation $\theta$ over $\queryvars{q}$ such that for each $i\in\{1,\dots,n\}$, 
$R_{i}(\underline{\vec{a}},\vec{b})\in\db$ with $\vec{a}=\theta(\vec{x}_{i})$ and $\vec{b}=\theta(\vec{y}_{i})$.

We say that a Boolean conjunctive query $q$ has a {\em self-join\/} if some relation name occurs more than once in $q$.
If $q$ has no self-join, then it is called {\em self-join-free\/}.
By a little abuse of notation, we may confuse atoms with their relation names in a self-join-free Boolean conjunctive query $q$.
That is, if we use a relation name $R$ at places where an atom is expected,
then we mean the (unique) $R$-atom of $q$.

If $q$ is a Boolean conjunctive query, $\vec{x}=\tuple{x_{1},\dots,x_{\ell}}$ is a sequence of distinct variables that occur in $q$, and $\vec{a}=\tuple{a_{1},\dots,a_{\ell}}$ is a sequence of constants,
then $\substitute{q}{\vec{x}}{\vec{a}}$ denotes the query obtained from $q$ by replacing all occurrences of $x_{i}$ with $a_{i}$, for all $1\leq i\leq\ell$.

\paragraph{Typed uncertain databases}
For every variable $x$, we assume an infinite set of constants, denoted $\type{x}$, such that $x\neq y$ implies $\type{x}\cap\type{y}=\emptyset$. 
Let $q$ be 	a self-join-free Boolean conjunctive query, and let $\db$ be an uncertain database.
We say that $\db$ is {\em typed relative to $q$\/} if for every atom $R(x_{1},\dots,x_{n})$ in $q$,
for every $i\in\{1,\dots,n\}$, if $x_{i}$ is a variable,
then for every fact $R(a_{1},\dots,a_{n})$ in $\db$, $a_{i}\in\type{x_{i}}$ and the constant $a_{i}$ does not occur in $q$.
Significantly, since $q$ is self-join-free, the assumption that uncertain databases are typed is without loss of generality.

\paragraph{Purified uncertain databases}
Let $q$ be a Boolean conjunctive query, and let $\db$ be an uncertain database.
We say that a fact $A\in\db$ is {\em relevant for $q$ in $\db$\/} if for some valuation $\theta$ over $\queryvars{q}$,
$A\in\theta(q)\subseteq\db$.
We say that $\db$ is {\em purified relative to $q$\/} if every fact $A\in\db$ is relevant for $q$ in $\db$.

\paragraph{Frugal repairs}
For every uncertain database $\db$, Boolean conjunctive query $q$, and $X\subseteq\queryvars{q}$,  we define a preorder $\fpro{q}{X}$ on $\repairs{\db}$, as follows.
For every two repairs $\rep_{1},\rep_{2}$, we define $\rep_{1}\fpro{q}{X}\rep_{2}$ if for every valuation $\theta$ over $X$,
$\rep_{1}\models\theta(q)$ implies $\rep_{2}\models\theta(q)$.
Here, $\theta(q)$ is the query obtained from $q$ by replacing all occurrences of each $x\in X$ with $\theta(x)$;
variables not in $X$ remain unaffected (i.e., $\theta$ is understood to be the identity on variables not in $X$).
Clearly, $\fpro{q}{X}$ is a preorder (i.e., it is reflexive and transitive), and its minimal elements are 
called {\em $\fpro{q}{X}$-frugal repairs\/}.\footnote{$\rep_{1}$ is minimal if for all $\rep_{2}$, if $\rep_{2}\fpro{q}{X}\rep_{1}$ then $\rep_{1}\fpro{q}{X}\rep_{2}$.}

\paragraph{Functional dependencies}
Let $q$ be a Boolean conjunctive query.
A {\em functional dependency for $q$\/} is an expression $\fd{X}{Y}$ where $X,Y\subseteq\queryvars{q}$.
We say that an uncertain database $\db$ {\em satisfies $\fd{X}{Y}$ for $q$\/}, denoted $\db\vmodels{q}\fd{X}{Y}$,
if for all valuations $\theta,\mu$ over $\queryvars{q}$ such that $\theta(q),\mu(q)\subseteq\db$,
if $\theta[X]=\mu[X]$,
then $\theta[Y]=\mu[Y]$.

\begin{example}
The relation $R$ shown next does not satisfy the standard functional dependency $\fd{2}{3}$,
because its tuples agree on the second position, but disagree on the third position.
Nevertheless, for $q=\exists y\exists z R(a,y,z)$, we have $R\vmodels{q}\fd{y}{z}$.
The second tuple of $R$ is not relevant for the query, because $a$ and $d$ are distinct constants;
the relation $R'$ is purified relative to $q$.
\begin{align*}
\begin{array}{cc}
\begin{array}{c|ccc}
R & \underline{1} & 2 & 3\\\cline{2-4}
  & a & b & c\\
	& d & b & f
\end{array}	
&
\begin{array}{c|ccc}
R' & \underline{1} & 2 & 3\\\cline{2-4}
   & a & b & c
\end{array}	
\end{array}
\end{align*}
\end{example}


\paragraph{Certain query answering}

For every Boolean conjunctive query $q$, the decision problem $\cqa{q}$ takes on input an uncertain database $\db$, and asks whether $q$ is satisfied by every repair of $\db$.

It is easy to show the following upper bound on the complexity of $\cqa{q}$.

\begin{theorem}\label{the:coNP}
For every Boolean first-order query $q$, $\cqa{q}$ is in $\coNP$.
\end{theorem}

The following two lemmas are useful in the study of the complexity of $\cqa{q}$.

\begin{lemma}[\cite{DBLP:conf/pods/Wijsen13}]\label{lem:purified}
Let $q$ be a Boolean conjunctive query.
Let $\db$ be an uncertain database.
It is possible to compute in polynomial time an uncertain database $\db'$ that is purified relative to $q$ such that
every repair of $\db$ satisfies $q$ if and only if every repair of $\db'$ satisfies $q$.
\end{lemma}

\begin{lemma}\label{lem:frugal}
Let $q$ be a self-join-free Boolean conjunctive query, and $X\subseteq\queryvars{q}$.
Let $\db$  be an uncertain database.
Then, every repair of $\db$ satisfies $q$ if and only if every $\fpro{q}{X}$-frugal repair of $\db$ satisfies $q$.
\end{lemma}

%% file: attackgraph.tex
\section{Attack Graphs}\label{sec:attackgraph}

Attack graphs were introduced in~\cite{DBLP:conf/pods/Wijsen10} for studying first-order expressibility of $\cqa{q}$ for acyclic (in the sense of~\cite{DBLP:journals/jacm/BeeriFMY83}) self-join-free conjunctive queries $q$.
Here, we extend the notion of attack graph to all (cyclic or acyclic) self-join-free conjunctive queries. 

Let $q$ be a self-join-free Boolean conjunctive query.
We define $\FD{q}$ as the following set of functional dependencies:
$$
\FD{q}\defeq\{\fd{\keyvars{F}}{\atomvars{F}}\mid F\in q\}
$$
For every atom $F\in q$, we define $\keycl{F}{q}$ and $\keyclsup{F}{q}$ as the following sets of variables.
\begin{eqnarray*}
\keycl{F}{q} & \defeq & \{x\in\queryvars{q}\mid\FD{q\setminus\{F\}}\models\fd{\keyvars{F}}{x}\}\\
\keyclsup{F}{q} & \defeq & \{x\in\queryvars{q}\mid\FD{q}\models\fd{\keyvars{F}}{x}\}
\end{eqnarray*}

The {\em attack graph of $q$\/} is a directed graph whose vertices are the atoms of $q$.
There is a directed edge from $F$ to $G$ ($F\neq G$) if there exists a sequence 
\begin{equation}\label{eq:defwitness}
F_{0}, F_{1}, \dots, F_{n}
\end{equation}
of (not necessarily distinct) atoms of $q$ such that 
\begin{itemize}
\item
$F_{0}=F$ and $F_{n}=G$; and
\item
for all $i\in\{0,\dots,n-1\}$, $\atomvars{F_{i}}\cap\atomvars{F_{i+1}}\nsubseteq\keycl{F}{q}$.
\end{itemize}
A directed edge from $F$ to $G$ in the attack graph of $q$ is also called an {\em attack from $F$ to $G$\/},  denoted by $F\attacks{q}G$. 
The sequence (\ref{eq:defwitness}) is called a {\em witness\/} for the attack $F\attacks{q}G$.
We will often add variables to a witness:
if we write 
$F_{0}\step{z_{1}}F_{1}\step{z_{2}}F_{2}\dots\step{z_{n}}F_{n}$,
then it is understood that for $i\in\{1,\dots,n\}$, 
$z_{i}\in\atomvars{F_{i-1}}\cap\atomvars{F_{i}}$ and $z_{i}\not\in\keycl{F_{0}}{q}$.
If $F\attacks{q}G$, then we also say that $F$ {\em attacks\/} $G$ (or that $G$ is attacked by $F$).

An attack from $F$ to $G$ is called {\em weak\/} if $\FD{q}\models\fd{\keyvars{F}}{\keyvars{G}}$; 
otherwise it is {\em strong\/}.
A directed cycle in the attack graph of of $q$ is called {\em weak\/} if all attacks in the cycle are weak;
otherwise the cycle is called {\em strong\/}.

\begin{figure}\centering
\includegraphics[scale=1.0]{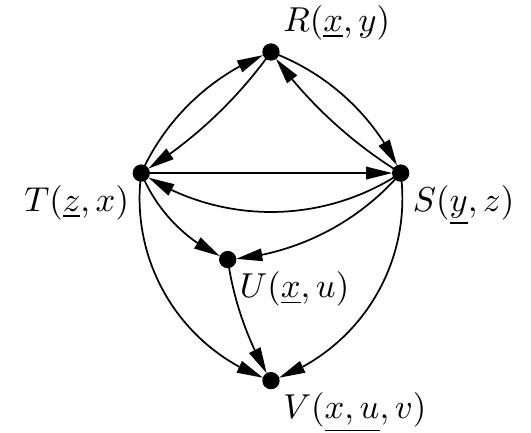}
\caption{Attack graph of the query in Example \ref{ex:ag}.}\label{fig:ag}
\end{figure}

\begin{example}\label{ex:ag}
Let $q=\{R(\underline{x},y)$, $S(\underline{y},z)$, $T(\underline{z},x)$, $U(\underline{x},u)$, $V(\underline{x,u},v)\}$.
By a little abuse of notation, we denote each atom by its relation name (e.g., $R$ is used to denote the atom $R(\underline{x},y)$).
We have $\keycl{R}{q}=\{x,u,v\}$.
A witness for $R\attacks{q}T$ is $R\step{y}S\step{z}T$.
The complete attack graph is shown in Fig.~\ref{fig:ag}.
All attacks are weak.
\end{example}

The above notion of attack graph is purely syntactic.
Semantically, an attack from an $R$-atom to an $S$-atom  in the attack graph of $q$ means
that there exists an uncertain database $\db$ such that every repair of $\db$ satisfies $q$,
and such that two $R$-facts of a same $R$-block join exclusively with two $S$-facts belonging to distinct $S$-blocks.
For the query of Example~\ref{ex:ag}, such a database could be $\db=\{R(\underline{1},a)$, $R(\underline{1},b)$, $S(\underline{a},\alpha)$,  $S(\underline{b},\beta),\dots\}$,
in which the two $R$-facts belong to the same $R$-block, and
$R(\underline{1},a)$ joins exclusively with $S(\underline{a},\alpha)$, and
$R(\underline{1},b)$ joins exclusively with $S(\underline{b},\beta)$,
and the two $S$-facts belong to distinct $S$-blocks.
Therefore, the attack graph of Fig.~\ref{fig:ag} contains a directed edge from the $R$-atom to the $S$-atom.

Equipped with the notion of attack graph, we can now present the effective
complexity trichotomy in the set \linebreak
$\{\cqa{q}\mid q$ is a self-join-free Boolean conjunctive query$\}$.

\begin{theorem}[Trichotomy Theorem]
\label{thm:main} 
Let $q$ be a self-join-free Boolean conjunctive query. 
\begin{enumerate}
\item If the attack graph of $q$ is acyclic, then $\cqa{q}$ is in $\FO$.
\item If the attack graph of $q$ is cyclic but contains no strong cycle, then $\cqa{q}$ is in $\P$ and is $\L$-hard.
\item If the attack graph of $q$ contains a strong cycle, then $\cqa{q}$ is $\coNP$-complete.
\end{enumerate}
\end{theorem}

The rest of the paper presents the proof of Theorem~\ref{thm:main}.
We first present some properties of attack graphs that will be useful in subsequent sections.

\begin{lemma}\label{lem:trans}
Let $q$ be a self-join-free Boolean conjunctive query.
If $F\attacks{q}G$ and $G\attacks{q}H$,
then either $F\attacks{q}H$ or $G\attacks{q}F$ (or both).
\end{lemma}

\begin{lemma}\label{lem:cycletwo}
Let $q$ be a self-join-free Boolean conjunctive query.
\begin{enumerate}
\item
If the attack graph of $q$ contains a cycle,
then it contains a cycle of size two.
\item
If the attack graph of $q$ contains a strong cycle,
then it contains a strong cycle of size two.
\end{enumerate}
\end{lemma}

\begin{lemma}\label{lem:nonewcycles}
Let $q$ be a self-join-free Boolean conjunctive query.
Let $x\in\queryvars{q}$ and let $a$ be an arbitrary constant.
\begin{enumerate}
\item
If the attack graph of $q$ is acyclic, then the attack graph of $\substitute{q}{x}{a}$ is acyclic.
\item
If the attack graph of $q$ contains no strong cycle, then the attack graph of $\substitute{q}{x}{a}$ contains no strong cycle.
\end{enumerate}
\end{lemma}

We conclude this section with three definitions.
The following definition is taken from~\cite{DBLP:journals/ipl/AspvallPT79} and applies to directed graphs in general.

\begin{definition}\label{def:tarjan}
A directed graph is {\em strongly connected\/} if there is a directed path from any vertex to any other.
The maximal strongly connected subgraphs of a graph are vertex-disjoint and are called its {\em strong components\/}.
If $S_{1}$ and $S_{2}$ are strong components such that an edge leads from a vertex in $S_{1}$ to a vertex in $S_{2}$,
then $S_{1}$ is a {\em predecessor\/} of $S_{2}$ and $S_{2}$ is a {\em successor\/} of $S_{1}$.
A strong component is called {\em initial\/} if it has no predecessor.
\end{definition}

Strong components in the attack graph should not be confused with strong attacks.

\begin{example}
In the attack graph of Fig.~\ref{fig:ag}, the atoms $R(\underline{x},y)$, $S(\underline{y},z)$, and $T(\underline{z},x)$
together constitute an initial strong component.
\end{example}

So far we have defined an attack from an atom to another atom. 
The following definition introduces attacks from an atom to a variable.

\begin{definition}\label{def:az}
Let $q$ be a self-join-free Boolean conjunctive query.
Let $R$ be a relation name with signature $\signature{1}{1}$ such that $R$ does not occur in $q$.
For $F\in q$ and $z\in\queryvars{q}$,
we say that {\em $F$ attacks $z$\/}, denoted $F\attacks{q}z$, if $F\attacks{q'}R(\underline{z})$ where $q'=q\cup\{R(\underline{z})\}$.
\end{definition}

\begin{example}
Clearly, if $F_{0}\step{z_{1}}F_{1}\dots\step{z_{n}}F_{n}$ is a witness for $F_{0}\attacks{q}F_{n}$,
then $F_{0}\attacks{q}z_{i}$ for every $i\in\{1,\dots,n\}$.
Notice also that if $q=\{R(\underline{x},y)\}$, then the attack graph of $q$ contains no edge, yet $R\attacks{q}y$.
\end{example}

Finally, we introduce the notion of {\em sequential proof\/}, which mimics an algorithm for testing logical implication for functional dependencies~\cite[Algorithm~8.2.7]{DBLP:books/aw/AbiteboulHV95}.

\begin{definition}\label{def:sqp}
Let $q$ be a self-join free Boolean conjunctive query. 
Let $X\subseteq\queryvars{q}$ and $y\in\queryvars{q}$.
A {\em sequential proof\/} of $\FD{q}\models\fd{X}{y}$ is a sequence $H_{0},H_{1},\dots,H_{\ell}$ of atoms of $q$
such that 
\begin{itemize}
\item
$y\in X\cup\bigcup_{i=1}^{\ell}\atomvars{H_{i}}$; and
\item
for $i\in\{0,\dots,\ell\}$, $\keyvars{H_{i}}\subseteq X\cup\bigcup_{j=0}^{i-1}\atomvars{H_{j}}$.
\end{itemize}
Notice that if $y\in X$, then the empty sequence is a sequential proof of  $\FD{q}\models\fd{X}{y}$.
\end{definition}

%% file: firstorder.tex
\section{First-Order Expressibility}\label{sec:firstorder}

In this section, we prove the first item in the statement of Theorem~\ref{thm:main}, as well as the $\L$-hard lower complexity bound stated in the second item. 

\begin{theorem}\label{the:FO}
Let $q$ be a self-join-free Boolean conjunctive query.
Then the following are equivalent:
\begin{enumerate}
\item\label{it:FOone}
$\cqa{q}$ is in $\FO$;
\item\label{it:FOtwo}
the attack graph of $q$ is acyclic.
\end{enumerate}
\end{theorem}

That is, acyclicity of the attack graph of $q$ is both a necessary and sufficient condition for first-order expressibility of $\cqa{q}$.
In Section~\ref{sec:necessary}, we show the contrapositive of the implication $\ref{it:FOone}\implies\ref{it:FOtwo}$.
In Section~\ref{sec:sufficient}, we show the implication $\ref{it:FOtwo}\implies\ref{it:FOone}$.

\subsection{Necessary Condition}\label{sec:necessary}

Let $q_{\fuxman}=\{R_{\fuxman}(\underline{x},y),S_{\fuxman}(\underline{y},x)\}$.
In~\cite{DBLP:journals/ipl/Wijsen10}, it was shown that $\cqa{q_{\fuxman}}$ is not in $\FO$.
The following lemma shows a stronger result.

\begin{lemma}\label{lem:Lhard}
Let $q_{\fuxman}=\{R_{\fuxman}(\underline{x},y),S_{\fuxman}(\underline{y},x)\}$. Then $\cqa{q_{\fuxman}}$ is $\L$-hard.
\end{lemma}

\begin{lemma}\label{lem:notFO}
Let $q$ be a self-join-free Boolean conjunctive query.
If the attack graph of $q$ is cyclic,
then $\cqa{q}$ is  $\L$-hard (and hence not in $\FO$).
\end{lemma}
\begin{proof}
Assume that the attack graph of $q$ is cyclic.
We show hereinafter that there exists a first-order many-one reduction from $\cqa{q_{\fuxman}}$ to $\cqa{q}$.
The desired result then follows from Lemma~\ref{lem:Lhard}.

By Lemma~\ref{lem:cycletwo}, we can assume two distinct atoms $F,G\in q$ such that $F\attacks{q}G\attacks{q}F$ is an attack cycle of size two.
We will assume hereinafter that the relation names in $F$ and $G$ are $R$ and $S$ respectively.

For all constants $a,b$ we define the valuation $\bival{a}{b}$ over $\queryvars{q}$ as follows.
Let $\bot$ be a fixed constant not occurring elsewhere.
For every variable $u\in\queryvars{q}$,
\begin{enumerate}
\item
if $u\in\keycl{F}{q}\setminus\keycl{G}{q}$, then $\bival{a}{b}(u)=a$;
\item
if $u\in\keycl{G}{q}\setminus\keycl{F}{q}$, then $\bival{a}{b}(u)=b$;
\item
if $u\in\keycl{F}{q}\cap\keycl{G}{q}$, then $\bival{a}{b}(u)=\bot$;
\item
if $u\in\queryvars{q}\setminus\formula{\keycl{F}{q}\cup\keycl{G}{q}}$, then $\bival{a}{b}(u)=\tuple{a,b}$.
\end{enumerate}

\begin{sublemma}\label{lem:notFnotG}
For all constants $a,b,a',b'$,
if $H\in q\setminus\{F,G\}$,
then $\{\bival{a}{b}(H),\bival{a'}{b'}(H)\}$ is consistent.
\end{sublemma}
\begin{subproof}[of Sublemma~\ref{lem:notFnotG}]
Assume that for all $u\in\keyvars{H}$, $\bival{a}{b}(u)=\bival{a'}{b'}(u)$.
We distinguish four cases.
\begin{description}
\item[Case $a=a'$ and $b=b'$.]
Then $\bival{a}{b}(H)=\bival{a'}{b'}(H)$.
\item[Case $a=a'$ and $b\neq b'$.]
Then $\keyvars{H}\subseteq\keycl{F}{q}$, hence $\atomvars{H}\subseteq\keycl{F}{q}$.
Then $\bival{a}{b}(H)=\bival{a'}{b'}(H)$.
\item[Case $a\neq a'$ and $b=b'$.]
Then $\keyvars{H}\subseteq\keycl{G}{q}$, hence $\atomvars{H}\subseteq\keycl{G}{q}$.
Then $\bival{a}{b}(H)=\bival{a'}{b'}(H)$.
\item[Case $a\neq a'$ and $b\neq b'$.]
Then $\keyvars{H}\subseteq\keycl{F}{q}\cap\keycl{G}{q}$, hence $\atomvars{H}\subseteq\keycl{F}{q}\cap\keycl{G}{q}$.
Then $\bival{a}{b}(H)=\bival{a'}{b'}(H)$.
\end{description} 
\end{subproof}

\begin{sublemma}\label{lem:FandG}
For all constants $a,b,a',b'$,
\begin{enumerate}
\item\label{it:abceone}
$\bival{a}{b}(F)$ and $\bival{a'}{b'}(F)$ are key-equal if and only if $a=a'$.
\item\label{it:abcetwo}
$\bival{a}{b}(F)=\bival{a'}{b'}(F)$ if and only if $a=a'$ and $b=b'$.
\item
$\bival{a}{b}(G)$ and $\bival{a'}{b'}(G)$ are key-equal if and only if $b=b'$.
\item
$\bival{a}{b}(G)=\bival{a'}{b'}(G)$ if and only if $a=a'$ and $b=b'$.
\end{enumerate}
\end{sublemma}
\begin{subproof}[of Sublemma~\ref{lem:FandG}]

\framebox{\ref{it:abceone}. $\implies$}
Consequence of $\keyvars{F}\nsubseteq\keycl{G}{q}$ (because $G\attacks{q}F$).
\framebox{\ref{it:abceone}. $\impliedby$}
Consequence of $\keyvars{F}\subseteq\keycl{F}{q}$.

\framebox{\ref{it:abcetwo}. $\implies$}
Consequence of $\atomvars{F}\nsubseteq\keycl{F}{q}$ (because $F\attacks{q}G$).
\framebox{\ref{it:abcetwo}. $\impliedby$}
Trivial.

The proof of the remaining items is analogous.
\end{subproof}

For every uncertain database $\db$ with $R_{\fuxman}$-facts and $S_{\fuxman}$-facts,
we define $f(\db)$ as the following uncertain database:
\begin{enumerate}
\item
for every $R_{\fuxman}(\underline{a},b)$ in $\db$, $f(\db)$ contains $\bival{a}{b}(q\setminus\{G\})$; and
\item
for every $S_{\fuxman}(\underline{b},a)$ in $\db$, $f(\db)$ contains $\bival{a}{b}(q\setminus\{F\})$.
\end{enumerate}
It is easy to see that $f$ is computable in $\FO$.

In what follows, we assume that $\db$ is typed, as explained in Section~\ref{sec:preliminaries}.
It will be understood that $a,a_{1},a_{2},\dots$ belong to $\type{x}$, and that $b,b_{1},b_{2},\dots$ belong to $\type{y}$.

Let us define $g(\db)$ as follows:
$$g(\db)\defeq f(\db)\setminus\left(\{\bival{a}{b}(F)\mid R_{\fuxman}(\underline{a},b)\in\db\}\cup \{\bival{a}{b}(G)\mid S_{\fuxman}(\underline{b},a)\in\db\}\right).$$
That is, $g(\db)$ contains all facts of $f(\db)$ that are neither $R$-facts nor $S$-facts.

By Sublemmas~\ref{lem:notFnotG} and~~\ref{lem:FandG},
\begin{equation}
\repairs{f(\db)}= \{f(\rep)\cup g(\db)\mid\rep\in\repairs{\db}\}.
\end{equation}
Let $\db$ be an arbitrary database with $R_{\fuxman}$-facts and $S_{\fuxman}$-facts.
It suffices to show that the following are equivalent for every repair $\rep$ of $\db$:
\begin{enumerate}
\item\label{it:fgone}
$\rep$ satisfies $q_{\fuxman}$;
\item\label{it:fgtwo}
$f(\rep)\cup g(\db)$ satisfies $q$.
\end{enumerate}
\framebox{\ref{it:fgone}$\implies$\ref{it:fgtwo}}
This is the easier part.

\noindent \framebox{\ref{it:fgtwo}$\implies$\ref{it:fgone}}
Let $\theta$ be a substitution over $\queryvars{q}$ such that $\theta(q)\subseteq f(\rep)\cup g(\db)$.

By our construction,
we can assume $R_{\fuxman}(\underline{a},b)\in\rep$ such that $\theta(F)\in\bival{a}{b}(q\setminus\{G\})$.
Likewise,
we can assume $S_{\fuxman}(\underline{b'},a')\in\rep$ such that $\theta(G)\in\bival{a'}{b'}(q\setminus\{F\})$.

It suffices to show that $a=a'$ and $b=b'$.

Before giving the proof, we provide some intuition.
For every fact $A\in f(\db)$, we can assume an atom in $q$, denoted $H_{A}$, such that $A=\bival{a}{b}(H_{A})$ for some constant $a\in\type{x}$ and some constant $b\in\type{y}$.
Then, for all $z\in\atomvars{H_{A}}$, $\bival{a}{b}(z)\in\{\bot,a,b,\tuple{a,b}\}$.
The constants in the latter set allow to ``trace back" $A$ to some facts $R_{\fuxman}(\underline{a},b)$ or $S_{\fuxman}(\underline{b},a)$ in $\db$. 

With this intuition in mind, it is easy to show $b=b'$ (the proof of $a=a'$ is symmetrical).
Since $F\attacks{q}G$,
there exists a sequence 
$F_{0}, F_{1}, \dots, F_{n}$ of atoms of $q$ such that 
\begin{itemize}
\item
$F_{0}=F$ and $F_{n}=G$; and
\item
for all $i\in\{0,\dots,n-1\}$, we can assume $u_{i}\in\atomvars{F_{i}}\cap\atomvars{F_{i+1}}$ such that $u_{i}\not\in\keycl{F}{q}$.
\end{itemize}

We show by induction on increasing $i$ that for all $i\in\{0,\dots,n-1\}$, there exists constant $a_{i}$ such that 
for all $w_{i}\in\atomvars{F_{i}}$, we have $\theta(w_{i})\in\{\bot$, $a_{i}$, $b$, $\tuple{a_{i},b}\}$.
\begin{description}
\item[Basis $i=0$.]
Since $\theta(F)\in\bival{a}{b}(q\setminus\{G\})$,
for all $w_{0}\in\atomvars{F_{0}}$, we have $\theta(w_{0})\in\{\bot$, $a$, $b$, $\tuple{a,b}\}$.
\item[Step $i\rightarrow i+1$.]
By the induction hypothesis, there exists constant $a_{i}$ such that 
for all $w_{i}\in\atomvars{F_{i}}$, we have $\theta(w_{i})\in\{\bot$, $a_{i}$, $b$, $\tuple{a_{i},b}\}$.

From $u_{i}\not\in\keycl{F}{q}$, it follows that $\theta(u_{i})\in\{b$, $\tuple{a_{i},b}\}$.

Since $u_{i}\in\atomvars{F_{i+1}}$, it follows that there exists constant $a_{i+1}$ such that 
for all $w_{i+1}\in\atomvars{F_{i+1}}$, we have $\theta(w_{i+1})\in\{\bot$, $a_{i+1}$, $b$, $\tuple{a_{i+1},b}\}$.
\end{description}
It follows that for $u_{n-1}\in\atomvars{G}$, there exists constant $a_{n-1}$ such that $\theta(u_{n-1})\in\{b$, $\tuple{a_{n-1},b}\}$.
From $\theta(G)\in\bival{a'}{b'}(q\setminus\{F\})$, it follows $\theta(u_{n-1})\in\{b'$, $\tuple{a',b'}\}$.
Consequently, $b=b'$.
\end{proof}

\subsection{Sufficient Condition}\label{sec:sufficient}

In this section, we show that $\cqa{q}$ is in $\FO$ if the attack graph of $q$ is acyclic.

\begin{lemma}\label{lem:inFO}
Let $q$ be a self-join-free Boolean conjunctive query.
Let $F$ be an atom of $q$ such that in the attack graph of $q$, the indegree of $F$ is zero.
Let $k=\card{\keyvars{F}}$ and let $\vec{x}=(x_{1},\dots,x_{k})$ be a sequence containing (exactly once) each variable of $\keyvars{F}$.
Then the following are equivalent for every uncertain database $\db$:
\begin{enumerate}
\item\label{it:general}
$q$ is true in every repair of $\db$;
\item\label{it:reify}
for some $\vec{a}\in\formula{\adom{\db}}^{k}$,
it is the case that $\substitute{q}{\vec{x}}{\vec{a}}$ is true in every repair of $\db$.
\end{enumerate}
\end{lemma}

Lemma~\ref{lem:inFO} immediately leads to the following result.

\begin{lemma}\label{lem:firstorder}
Let $q$ be a self-join-free Boolean conjunctive query.
If the attack graph of $q$ is acyclic,
then $\cqa{q}$ is in $\FO$.
\end{lemma}
\begin{proof}
Assume that the attack graph of $q$ is acyclic.

The proof runs by induction on $\card{q}$.
If $\card{q}=0$, then $\cqa{q}$ is obviously in $\FO$.

Let $\db$ be an instance of $\cqa{q}$.
Since the attack graph of $q$ is acyclic,
we can assume an atom $R(\underline{\vec{x}},\vec{y})$ that is not attacked in the attack graph of $q$. 
By Lemma~\ref{lem:inFO}, the following are equivalent:
\begin{enumerate}
\item
$q$ is true in every repair of $\db$.
\item\label{it:testfo}
For some fact $R(\underline{\vec{a}},\vec{b})\in\db$,  
there exists of a valuation $\theta$ over $\sequencevars{\vec{x}}$ such that $\theta(\vec{x})=\vec{a}$ and such that for all key-equal facts $R(\underline{\vec{a}},\vec{b}')$ in $\db$,
the valuation $\theta$ can be extended to a valuation $\theta^{+}$ over $\sequencevars{\vec{x}}\cup\sequencevars{\vec{y}}$ such that 
$\theta^{+}(\vec{y})=\vec{b}$ and $\theta^{+}(q')$ is true in every repair of $\db$, where $q'=q\setminus\{R(\underline{\vec{x}},\vec{y})\}$. 
\end{enumerate}
From Lemma~\ref{lem:nonewcycles}, it follows that the attack graph of $\theta^{+}(q')$ is acyclic, and hence $\cqa{\theta^{+}(q')}$ is in $\FO$ by the induction hypothesis.
It is then clear that the latter condition~(\ref{it:testfo}) can be checked in $\FO$.
\end{proof}

For a self-join-free Boolean conjunctive query $q$, the problem $\cqa{q}$ can be equivalently defined as the set containing every uncertain database $\db$ such that every repair of $\db$ satisfies $q$.
If $\cqa{q}$ is in $\FO$, then the set $\cqa{q}$ is definable in first-order logic (by definition of the complexity class $\FO$).
If $\cqa{q}$ is in $\FO$, then its first-order definition is commonly called {\em first-order rewriting\/}. 
Such a first-order rewriting is actually an implementation, in first-order logic, of the algorithm in the proof of Lemma~\ref{lem:firstorder}.
This is illustrated next.

\begin{example}
Let $q=\{R(\underline{x},y), S(\underline{y}, b)\}$, where $b$ is a constant.
The attack graph of $q$ contains a single directed edge, from the $R$-atom to the $S$-atom.
The first-order definition of $\cqa{q}$ is as follows:
$$
\setlength{\arraycolsep}{0pt}
\begin{array}{rl}
\exists x\exists y & \left(R(\underline{x},y)\land\right.\\
        \forall y  & \left.\formula{R(\underline{x},y)\rightarrow\formula{S(\underline{y},b)\land\forall z\formula{S(\underline{y},z)\rightarrow z=b}}}\right).
\end{array}
$$
\end{example}

%% file: intractability.tex
\section{Intractability Result}\label{sec:intractability}

In this section, we prove the $\coNP$-hard lower complexity bound stated in the third item of Theorem~\ref{thm:main}. 


\begin{theorem}\label{the:hard}
Let $q$ be a self-join-free Boolean conjunctive query.
If the attack graph of $q$ contains a strong cycle,
then $\cqa{q}$ is $\coNP$-hard.
\end{theorem}
\begin{proof}
Assume that the attack graph of $q$ contains a strong cycle.
By Lemma~\ref{lem:cycletwo}, we can assume $F,G\in q$ such that $F\attacks{q}G\attacks{q}F$ and the attack $F\attacks{q}G$ is strong.
We will assume hereinafter that the relation names in $F$ and $G$ are $R$ and $S$ respectively.

Let $q_{\pema}=\{R_{\pema}(\underline{x},y),S_{\pema}(\underline{y,z},x)\}$.
We show hereinafter that there exists a polynomial-time (and even first-order) many-one reduction from $\cqa{q_{\pema}}$ to $\cqa{q}$.
Since it is known~\cite{DBLP:journals/ipl/KolaitisP12} that $\cqa{q_{\pema}}$ is $\coNP$-hard,
it follows that  $\cqa{q}$ is $\coNP$-hard.

\begin{figure}\centering
\includegraphics{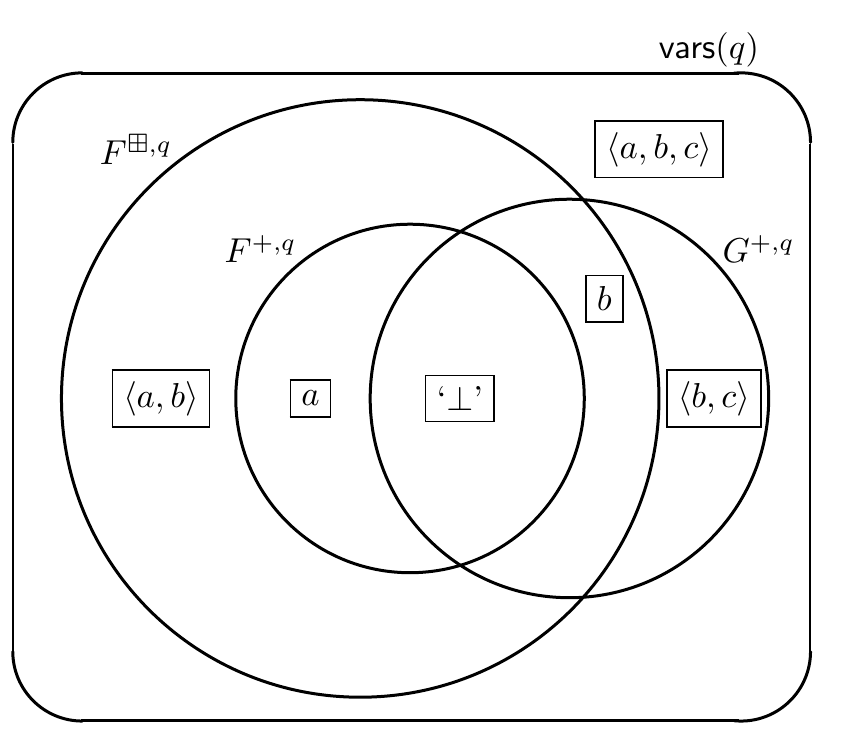}
\caption{Help for the proof of Theorem~\ref{the:hard}.}\label{fig:venn}
\end{figure}

For all constants $a,b,c$,
we define $\trival{a}{b}{c}$ as the following valuation over $\queryvars{q}$ (see Fig.~\ref{fig:venn} for a mnemonic).
Let $\bot$ be some fixed constant.
\begin{enumerate}
\item
If $u\in\keycl{F}{q}\cap\keycl{G}{q}$, then $\trival{a}{b}{c}(u)=\bot$;
\item
if $u\in\keycl{F}{q}\setminus\keycl{G}{q}$, then $\trival{a}{b}{c}(u)=a$;
\item
if $u\in\keycl{G}{q}\setminus\keyclsup{F}{q}$, then $\trival{a}{b}{c}(u)=\tuple{b,c}$;
\item
if $u\in\formula{\keycl{G}{q}\cap\keyclsup{F}{q}}\setminus\keycl{F}{q}$, then $\trival{a}{b}{c}(u)=b$;
\item
if $u\in\keyclsup{F}{q}\setminus\formula{\keycl{F}{q}\cup\keycl{G}{q}}$, then $\trival{a}{b}{c}(u)=\tuple{a,b}$; and
\item
if $u\not\in\keyclsup{F}{q}\cup\keycl{G}{q}$, then $\trival{a}{b}{c}(u)=\tuple{a,b,c}$.
\end{enumerate}

\begin{sublemma}\label{lem:notFnotGstrong}
For all constants $a,b,c,a',b',c'$,
if $H\in q\setminus\{F,G\}$,
then $\{\trival{a}{b}{c}(H),\trival{a'}{b'}{c'}(H)\}$ is consistent.
\end{sublemma}
\begin{subproof}[of Sublemma~\ref{lem:notFnotG}]
Assume that for all $u\in\keyvars{H}$, 
\begin{equation}\label{eq:keyeq}
\trival{a}{b}{c}(u)=\trival{a'}{b'}{c'}(u).
\end{equation}
We distinguish four cases.
\begin{description}
\item[Case $a=a'$ and $b=b'$.]
If $c=c'$, then $\trival{a}{b}{c}(H)=\trival{a'}{b'}{c'}(H)$.
Assume next $c\neq c'$.
From~(\ref{eq:keyeq}), it follows $\keyvars{H}\subseteq\keyclsup{F}{q}$.
Consequently, $\atomvars{H}\subseteq\keyclsup{F}{q}$.
Since $c$ does not occur inside $\keyclsup{F}{q}$ in the Venn diagram of Fig.~\ref{fig:venn},
we have $\trival{a}{b}{c}(H)=\trival{a'}{b'}{c'}(H)$.
\item[Case $a=a'$ and $b\neq b'$.]
From~(\ref{eq:keyeq}), it follows $\keyvars{H}\subseteq\keycl{F}{q}$, hence $\atomvars{H}\subseteq\keycl{F}{q}$.
Since $b$ and $c$ do not occur inside $\keycl{F}{q}$ in the Venn diagram, $\trival{a}{b}{c}(H)=\trival{a'}{b'}{c'}(H)$.
\item[Case $a\neq a'$ and $b=b'$.]
First assume $c=c'$.
From~(\ref{eq:keyeq}), it follows $\keyvars{H}\subseteq\keycl{G}{q}$, hence $\atomvars{H}\subseteq\keycl{G}{q}$.
Since $c$ does not occur inside $\keycl{G}{q}$ in the Venn diagram, $\trival{a}{b}{c}(H)=\trival{a'}{b'}{c'}(H)$.

Next assume $c\neq c'$.
From~(\ref{eq:keyeq}), it follows $\keyvars{H}\subseteq\keyclsup{F}{q}\cap\keycl{G}{q}$, hence $\atomvars{H}\subseteq\keyclsup{F}{q}\cap\keycl{G}{q}$.
Since $a$ and $c$ do not occur inside $\keyclsup{F}{q}\cap\keycl{G}{q}$ in the Venn diagram, $\trival{a}{b}{c}(H)=\trival{a'}{b'}{c'}(H)$.
\item[Case $a\neq a'$ and $b\neq b'$.]
From~(\ref{eq:keyeq}), it follows $\keyvars{H}\subseteq\keycl{F}{q}\cap\keycl{G}{q}$, hence $\atomvars{H}\subseteq\keycl{F}{q}\cap\keycl{G}{q}$.
Since $a,b,c$ do not occur inside $\keycl{F}{q}\cap\keycl{G}{q}$ in the Venn diagram, $\trival{a}{b}{c}(H)=\trival{a'}{b'}{c'}(H)$.
\end{description}
\end{subproof}

\begin{sublemma}\label{lem:FandGstrong}
For all constants $a,b,c,a',b',c'$,
\begin{enumerate}
\item\label{it:abceonestrong}
$\trival{a}{b}{c}(F)$ and $\trival{a'}{b'}{c'}(F)$ are key-equal if{f} $a=a'$.
\item\label{it:abcetwostrong}
$\trival{a}{b}{c}(F)=\trival{a'}{b'}{c'}(F)$ if{f} $a=a'$ and $b=b'$.
\item\label{it:abcethreestrong}
$\trival{a}{b}{c}(G)$ and $\trival{a'}{b'}{c'}(G)$ are key-equal if{f}  $b=b'$ and $c=c'$.
\item\label{it:abcefourstrong}
$\trival{a}{b}{c}(G)=\trival{a'}{b'}{c'}(G)$ if{f} $a=a'$ and $b=b'$ and $c=c'$.
\end{enumerate}
\end{sublemma}
\begin{subproof}[of Sublemma~\ref{lem:FandGstrong}]

\framebox{\ref{it:abceonestrong}. $\implies$}
Consequence of $\keyvars{F}\nsubseteq\keycl{G}{q}$ (because $G\attacks{q}F$).
\framebox{\ref{it:abceonestrong}. $\impliedby$}
Consequence of $\keyvars{F}\subseteq\keycl{F}{q}$.

\framebox{\ref{it:abcetwostrong}. $\implies$}
Consequence of $\atomvars{F}\nsubseteq\keycl{F}{q}$ (because $F\attacks{q}G$).
\framebox{\ref{it:abcetwostrong}. $\impliedby$}
Consequence of $\atomvars{F}\subseteq\keyclsup{F}{q}$.

\framebox{\ref{it:abcethreestrong}. $\implies$}
Consequence of $\keyvars{G}\nsubseteq\keyclsup{F}{q}$ (because $F\attacks{q}G$ is a strong attack).
\framebox{\ref{it:abcethreestrong}. $\impliedby$}
Consequence of $\keyvars{G}\subseteq\keycl{G}{q}$.

\framebox{\ref{it:abcefourstrong}. $\implies$}
Consequence of item~\ref{it:abcethreestrong} and $\atomvars{G}\nsubseteq\keycl{G}{q}$ (because $G\attacks{q}F$).
\framebox{\ref{it:abcefourstrong}. $\impliedby$}
Trivial.
\end{subproof}

Let $\db$ be uncertain database with $R_{\pema}$-facts and $S_{\pema}$-facts.
In what follows, we assume that $\db$ is typed, as explained in Section~\ref{sec:preliminaries}.
It will be understood that $a,a_{1},a_{2},\dots$ belong to $\type{x}$, that $b,b_{1},b_{2},\dots$ belong to $\type{y}$, and that $c,c_{1},c_{2},\dots$ belong to $\type{z}$.

Let $h(\db)$ be the subset of $\db$ such that
\begin{enumerate}
\item
$h(\db)$contains all $S_{\pema}$-facts of $\db$; and
\item
$h(\db)$ contains every $R_{\pema}$-block $\block$ of $\db$ such that for every fact $R_{\pema}(\underline{a},b)$ in $\block$, 
there exists some constant $c$ such that $S_{\pema}(\underline{b,c},a)$ is in $\db$.
\end{enumerate}
Clearly, the computation of $h(\db)$ from $\db$ is in $\FO$, and the following are equivalent:
\begin{enumerate}
\item
every repair of $\db$ satisfies $q_{\pema}$; 
\item
every repair of $h(\db)$ satisfies $q_{\pema}$.
\end{enumerate}
We define $f(\db)$ as the following uncertain database:
\begin{enumerate}
\item
for every pair $\{R_{\pema}(\underline{a},b), S_{\pema}(\underline{b,c},a)\}$ contained in $h(\db)$, $f(\db)$ contains $\trival{a}{b}{c}(q\setminus\{G\})$; and
\item
for every $S_{\pema}(\underline{b,c},a)$ in $h(\db)$, $f(\db)$ contains $\trival{a}{b}{c}(q\setminus\{F\})$.
\end{enumerate}
It is easy to see that $f$ is computable in $\FO$.

Let $g(\db)$ be the subset of $f(\db)$ containing all facts of $f(\db)$ that are neither $R$-facts nor $S$-facts.

By Sublemmas~\ref{lem:notFnotGstrong} and~~\ref{lem:FandGstrong},
\begin{equation}
\repairs{f(\db)}= \{f(\rep)\cup g(\db)\mid\rep\in\repairs{\db}\}.
\end{equation}

Let $\db$ be an arbitrary database with $R_{\pema}$-facts and $S_{\pema}$-facts.
It suffices to show that the following are equivalent for every repair $\rep$ of $\db$:
\begin{enumerate}
\item\label{it:fgonestrong}
$\rep$ satisfies $q_{\pema}$;
\item\label{it:fgtwostrong}
$f(\rep)\cup g(\db)$ satisfies $q$.
\end{enumerate}
\framebox{\ref{it:fgonestrong}$\implies$\ref{it:fgtwostrong}}
This is the easier part.

\framebox{\ref{it:fgtwostrong}$\implies$\ref{it:fgonestrong}}
Let $\theta$ be a substitution over $\queryvars{q}$ such that $\theta(q)\subseteq f(\rep)\cup g(\db)$.
By our construction,
we can assume $R_{\pema}(\underline{a},b)\in\rep$ and some constant $c$ such that $\theta(F)\in\trival{a}{b}{c}(q\setminus\{G\})$.
Likewise,
we can assume $S_{\pema}(\underline{b',c'},a')\in\rep$ such that $\theta(G)\in\trival{a'}{b'}{c'}(q\setminus\{F\})$.
It suffices to show that $a=a'$ and $b=b'$.

\framebox{$b=b'$} 
Since $F\attacks{q}G$,
there exists a sequence 
$F_{0}, F_{1}, \dots, F_{n}$ of distinct atoms of $q$ such that 
\begin{itemize}
\item
$F_{0}=F$ and $F_{n}=G$; and
\item
for all $i\in\{0,\dots,n-1\}$, we can assume $u_{i}\in\atomvars{F_{i}}\cap\atomvars{F_{i+1}}$ such that $u_{i}\not\in\keycl{F}{q}$.
\end{itemize}

We show by induction on increasing $i$ that for all $i\in\{0,\dots,n-1\}$, there exist constants $a_{i}$ and $c_{i}$ such that 
for all $w_{i}\in\atomvars{F_{i}}$, we have $\theta(w_{i})\in\{\bot$, $a_{i}$, $b$, $\tuple{a_{i},b}$, $\tuple{b,c_{i}}$, $\tuple{a_{i},b,c_{i}}\}$.
\begin{description}
\item[Basis $i=0$.]
Since $\theta(F)\in\trival{a}{b}{c}(q\setminus\{G\})$,
for all $w_{0}\in\atomvars{F_{0}}$, we have $\theta(w_{0})\in\{\bot$, $a$, $b$, $\tuple{a,b}$, $\tuple{b,c}$, $\tuple{a,b,c}\}$.
\item[Step $i\rightarrow i+1$.]
By the induction hypothesis, there exist constants $a_{i}$ and $c_{i}$ such that 
for all $w_{i}\in\atomvars{F_{i}}$, we have $\theta(w_{i})\in\{\bot$, $a_{i}$, $b$, $\tuple{a_{i},b}$, $\tuple{b,c_{i}}$, $\tuple{a_{i},b,c_{i}}\}$.

From $u_{i}\not\in\keycl{F}{q}$, it follows that $\theta(u_{i})\in\{b$, $\tuple{a_{i},b}$, $\tuple{b,c_{i}}$, $\tuple{a_{i},b,c_{i}}\}$.

Since $u_{i}\in\atomvars{F_{i+1}}$, it follows that there exist constants $a_{i+1}$ and $c_{i+1}$ such that 
for all $w_{i+1}\in\atomvars{F_{i+1}}$, we have $\theta(w_{i+1})\in\{\bot$, $a_{i+1}$, $b$, $\tuple{a_{i+1},b}$, $\tuple{b,c_{i+1}}$, $\tuple{a_{i+1},b,c_{i+1}}\}$.
\end{description}
It follows that for $u_{n-1}\in\atomvars{G}$, there exist constants $a_{n-1}$ and $c_{n-1}$ such that $\theta(u_{n-1})\in\{b$, $\tuple{a_{n-1},b}$, $\tuple{b,c_{n-1}}$, $\tuple{a_{n-1},b,c_{n-1}}\}$.
From $\theta(G)\in\trival{a'}{b'}{c'}(q\setminus\{F\})$, it follows $\theta(u_{n-1})\in\{b'$, $\tuple{a',b'}$, $\tuple{b',c'}$, $\tuple{a',b',c'}\}$.
Consequently, $b=b'$.

\framebox{$a=a'$} 
Analogous.
\end{proof}

%% file: tractability.tex
\section{Polynomial Tractability}
\label{sec:road}

In this section, we prove the $\P$ upper complexity bound stated in the second item of Theorem~\ref{thm:main}.

\begin{theorem} \label{the:ptime}
Let $q$ be a self-join-free Boolean conjunctive query.
If the attack graph of $q$ contains no strong cycle,
then $\cqa{q}$ is in $\P$.
\end{theorem}

\paragraph{Road map} The proof of Theorem~\ref{the:ptime} is technically involved.
We start by introducing in Section~\ref{subsec:mode} an extension of the data model that allows some syntactic simplifications,
expressed in Section~\ref{subsec:synsim}. 
In Section~\ref{subsec:markov}, we introduce the notion of {\em Markov cycle\/},
and show how the ``dissolution" of Markov cycles is helpful in the proof of Theorem~\ref{the:ptime}, which is given in Section~\ref{subsec:final}.
The dissolution of Markov cycles is explained in detail in Section~\ref{subsec:soluble}.

\subsection{Relations Known to Be Consistent}\label{subsec:mode}

We conservatively extend our data model.
We first distinguish between two kinds of relation names:
those that can be inconsistent, and those that cannot.

\paragraph{Relations known to be consistent}
Every relation name has a unique and fixed {\em mode\/}, which is an element in $\{i,c\}$.
It will come in handy to think of $i$ and $c$ as inconsistent and consistent respectively.
We often write $R^{c}$ to denote that $R$ is a relation name with mode $c$.
If $q$ is a self-join-free Boolean conjunctive query,
then $\consistent{q}$ denotes the subset of $q$ containing each atom whose relation name has mode $c$.
The {\em inconsistency count\/} of $q$, denoted $\icard{q}$, is the number of relation names with mode $i$ in $q$.
Modes carry over to atoms and facts: the mode of an atom $R(\underline{\vec{x}},\vec{y})$ or a fact $R(\underline{\vec{a}},\vec{b})$ is the mode of $R$.

The intended semantics is that if a relation name $R$ has mode $c$,
then the set of $R$-facts of an uncertain database will always be consistent.

\paragraph{Certain query answering with consistent and inconsistent relations} 
The problem $\cqa{q}$ now takes as input an uncertain database $\db$ such that
for every relation name $R$ in $q$,
if $R$ has mode $c$, then the set of $R$-facts of $\db$ is consistent.
The problem is to determine whether every repair of $\db$ satisfies $q$.

All results shown in previous sections carry over to the new setting, by assuming that all relation names used so far had mode $i$.
Furthermore, as stated by Proposition~\ref{pro:mode} (which has an easy proof),  relation names with mode $c$ can be simulated by means exclusively of relation names with mode $i$.
Therefore, having relation names with mode $c$ will be convenient, but is not fundamental.

\begin{proposition}\label{pro:mode}
Let $q$ be a self-join free Boolean conjunctive query. 
Let $R^{c}(\underline{\vec{x}},\vec{y})$ be an atom with mode $c$ in $q$.
Let $R_{1}$ and $R_{2}$ be two relation names, both with mode $i$ and with the same signature as $R$,
such that neither $R_{1}$ nor $R_{2}$ occurs in $q$.
Let $q'=\formula{q\setminus\{R^{c}(\underline{\vec{x}},\vec{y})\}}\cup\{R_{1}(\underline{\vec{x}},\vec{y}), R_{2}(\underline{\vec{x}},\vec{y})\}$. 
Then $\cqa{q}$ and $\cqa{q'}$ are equivalent under first-order reductions.
\end{proposition}

If relation names with mode $c$ are allowed for syntactic convenience, the definition of $\keycl{F}{q}$ needs slight change:
$$
\keycl{F}{q}\defeq\{x\in\queryvars{q}\mid\FD{\formula{q\setminus F}\cup\consistent{q}}\models\fd{\keyvars{F}}{x}\}
$$
Modulo this redefinition, the notion of attack graph remains unchanged.

Proposition~\ref{pro:mode} explains how to replace atoms with mode~$c$.
Conversely, the following lemma states that in pursuing a proof for Theorem~\ref{the:ptime}, 
there are cases where a self-join-free Boolean conjunctive query can be extended with atoms of mode~$c$.

\begin{lemma}\label{lem:reduction}
Let $q$ be a self-join-free Boolean conjunctive query.
Let $x,z\in\queryvars{q}$ such that $\FD{q}\models \fd{x}{z}$ and for every $F\in q$, 
if $\FD{q}\models\fd{x}{\keyvars{F}}$, then $F\nattacks{q}x$ and $F\nattacks{q}z$.
Let $q'=q\cup\{T^c(\underline{x},z)\}$, where $T$ is a fresh relation name with mode $c$.
Then, 
\begin{enumerate}
\item\label{it:reductionone} there exists a polynomial-time many-one reduction from $\cqa{q}$ to  $\cqa{q'}$; and
\item\label{it:reductiontwo} if the attack graph of $q$ contains no strong cycle, then the attack graph of $q'$ contains no strong cycle either.
\end{enumerate} 
\end{lemma}

\paragraph{Saturated queries}
Given a self-join-free Boolean conjunctive query,
the reduction of Lemma~\ref{lem:reduction} can be repeated until it can no longer be applied.
The query so obtained will be called {\em saturated\/}. 

\begin{definition}\label{def:saturated}
Let $q$ be a self-join-free Boolean conjunctive query.
We say that $q$ is {\em saturated\/} if whenever $x,z \in\queryvars{q}$ such that $\FD{q}\models\fd{x}{z}$ and $\FD{\consistent{q}}\not\models\fd{x}{z}$, 
then there exists an atom $F\in q$ with $\FD{q}\models\fd{x}{\keyvars{F}}$ such that $F\attacks{q}x$ or $F\attacks{q}z$.
\end{definition}

\begin{example}\label{ex:saturated}
Consider the query
$q=\{R(\underline{x},y)$, $S_1(\underline{y},z)$, $S_2(\underline{y},z)$, $T^c(\underline{x,z},w)$, $U(\underline{w},x)\}$.
We have $\FD{q}\models\fd{y}{z}$ and $\FD{\consistent{q}}\not \models\fd{y}{z}$.
The set $\{F\in q\mid\FD{q}\models\fd{y}{\keyvars{F}}\}$ equals $\{S_{1}, S_{2}\}$.
We have neither $S_{1}\attacks{q}y$ nor $S_{1}\attacks{q}z$.
Likewise, neither $S_{2}\attacks{q}y$ nor $S_{2}\attacks{q}z$.
Hence, $q$ is not saturated. 
By Lemma~\ref{lem:reduction}, there exists a polynomial-time many-one reduction from $\cqa{q}$ to $\cqa{q'}$ with $q'=q\cup\{S^c(\underline{y},z)\}$, where $S$ is a fresh relation name with mode~$c$. 
It can be verified that the query $q'$ is saturated.
\end{example}

\subsection{Syntactic Simplifications}\label{subsec:synsim}

The following lemma shows that any proof of Theorem~\ref{the:ptime} can assume some syntactic simplifications without loss of generality.

\begin{lemma}\label{lem:mostsimplified}
Let $q$ be a self-join-free Boolean conjunctive query.
There exists a polynomial-time many-one reduction from $\cqa{q}$ to $\cqa{q'}$ for some self-join-free Boolean conjunctive query $q'$ with the following properties:
\begin{itemize}
\item
$\icard{q'}\leq\icard{q}$; 
\item
no atom in $q'$ contains two occurrences of the same variable; 
\item
constants occur in $q'$ exclusively at the primary-key position of simple-key atoms;
\item
every atom with mode $i$ in $q'$ is simple-key;
\item
$q'$ is saturated; and
\item
if the the attack graph of $q$ contains no strong cycle,
then the attack graph of $q'$ contains no strong cycle either.
\end{itemize}
\end{lemma}

\subsection{Dissolving Markov Cycles} \label{subsec:markov}

The following definition introduces Markov graphs.

\begin{definition}\label{def:markov} 
Let $q$ be a self-join-free Boolean conjunctive query such that every atom with mode $i$ in $q$ is simple-key.
For every $x\in\queryvars{q}$, we define
$$\clutch{x}{q}\defeq\{F\in q\mid\mbox{$F$ has mode $i$ and $\keyvars{F}=\{x\}$}\}.$$
Notice that $\clutch{x}{q}$ can be empty.

The {\em Markov graph\/} of $q$ is a directed graph whose vertex set is $\queryvars{q}$.
There is a directed edge from $x$ to $y$, denoted $x\markovarg{q}y$, if $x\neq y$ and $\FD{\clutch{x}{q}\cup\consistent{q}}\models\fd{x}{y}$. 
If the query $q$ is clear from the context, then $x\markovarg{q}y$ can be shortened into $x\markov y$.
We write $x\markovpatharg{q}y$ (or $x\markovpath y$ if $q$ is clear from the context) if the Markov graph of $q$ contains a directed path from $x$ to $y$.\footnote{The term Markov refers to the intuition that in a Markov path, each variable functionally determines the next variable in the path, independently of preceding variables.}
Notice that for every $x\in\queryvars{q}$, $x\markovpatharg{q}x$.

An elementary directed cycle $\calC$ in the Markov graph of $q$ is said to be {\em premier\/} if there exists a variable $x\in\queryvars{q}$ such that
\begin{enumerate}
\item
$\{x\}=\keyvars{F_{0}}$ for some atom $F_{0}$  with mode $i$ that belongs to an initial strong component of the attack graph of $q$; and
\item
for some $y$ in $\calC$, we have $x\markovpatharg{q}y$ and $\FD{q}\models\fd{y}{x}$.
\end{enumerate}
The term {\em Markov edge\/} is used for an edge in the Markov graph; likewise for {\em Markov path\/} and {\em Markov cycle\/}. 
\end{definition}

\begin{figure*}\centering
\begin{tabular}{cc}
\includegraphics[scale=1.0]{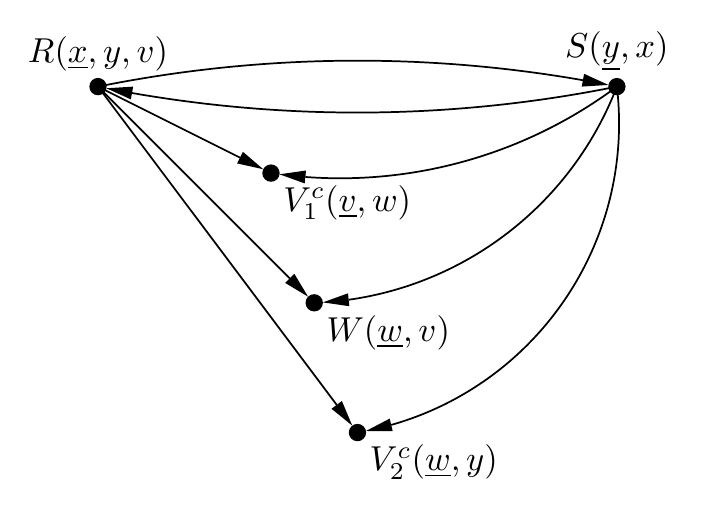}
&
\includegraphics[scale=1.0]{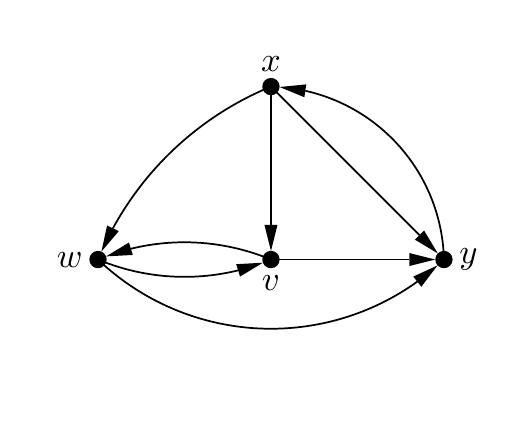}
\end{tabular}
\caption{Attack graph (left) and Markov graph (right) of the query
$\{R(\underline{x},y,v)$, $S(\underline{y},x)$, $V_{1}^{c}(\underline{v},w)$, $W(\underline{w},v)$ $V_{2}^{c}(\underline{w},y)\}$.}\label{fig:attmar}
\end{figure*}

\begin{example}
Let $q=\{R(\underline{x},y,v)$, $S(\underline{y},x)$, $V_{1}^{c}(\underline{v},w)$, $W(\underline{w},v)$ $V_{2}^{c}(\underline{w},y)\}$.
All atoms in $q$ are simple-key.
Then, $\consistent{q}=\{V_{1}^{c}(\underline{v},w)$, $V_{2}^{c}(\underline{w},y)\}$.

We have $\clutch{x}{q}=\{R(\underline{x},v,y)\}$.
Since $\FD{\clutch{x}{q}\cup\consistent{q}}\models\fd{x}{\{y,v,w\}}$,
the Markov graph of $q$ contains directed edges from $x$ to each of $y$, $v$, and $w$.

We have $\clutch{v}{q}=\emptyset$.
Since $\FD{\clutch{v}{q}\cup\consistent{q}}\models\fd{v}{\{y,w\}}$,
the Markov graph of $q$ contains directed edges from $v$ to both $y$ and $w$.
The complete Markov graph of $q$ is shown in Fig.~\ref{fig:attmar} (right). 

The attack graph of $q$ is shown in Fig.~\ref{fig:attmar} (left).
The atoms $R(\underline{x},y,v)$ and $S(\underline{y},x)$ together constitute an initial strong component of the attack graph. 
It is then straightforward that each cycle in the Markov graph of $q$ that contains $x$ or $y$, must be premier. 
Further, the cycle $v,w,v$ in the Markov graph of $q$ is also premier,
because there is a Markov path from $x$ to $v$, and $\FD{q}\models\fd{v}{x}$.
\end{example}

Let $q$ be like in Definition~\ref{def:markov} and assume that the Markov graph of $q$ contains an elementary directed cycle $\calC$.
Lemma~\ref{lem:soluble} states that $\cqa{q}$ can be reduced in polynomial time to $\cqa{q^{*}}$,
where $q^{*}$ is obtained from $q$ by ``dissolving" the Markov cycle $\calC$ as defined in Definition~\ref{def:resolve}. 
Moreover, we will show (Lemma~\ref{lem:nonewstrongcycles}) that if $\calC$ is premier and the attack graph of $q$ contains no strong cycle,
then the attack graph of $q^{*}$ will contain no strong cycle either.
The reduction that ``dissolves" Markov cycles will be the central idea in our polynomial-time algorithm for $\cqa{q}$ when the attack graph of $q$ contains no strong cycle. 

\begin{definition}\label{def:resolve}
Let $q$ be a self-join-free Boolean conjunctive query such that every atom with mode $i$ in $q$ is simple-key.
Let $\calC$ be an elementary directed cycle of length $k\geq 2$ in the Markov graph of $q$.
Then, $\resolve{\calC}{q}$ denotes the self-join-free Boolean conjunctive query defined next.
Let $x_{0},\dots,x_{k-1}$ be the variables in $\calC$, and let $q_{0}=\bigcup_{i=0}^{k-1}\clutch{x_{i}}{q}$.
Let $\vec{y}$ be a sequence of variables containing exactly once each variable of $\queryvars{q_{0}}\setminus\{x_{0},\dots,x_{k-1}\}$.
Let $q_{1}=\{T(\underline{u},x_{0},\dots,x_{k-1},\vec{y})\}\cup\{U_{i}^{c}(\underline{x_{i}},u)\}_{i=0}^{k-1}$,
where $u$ is a fresh variable, $T$ is a fresh relation name with mode $i$, 
and $U_{1},\dots,U_{k-1}$ are fresh relation names with mode $c$.
Then, we define
$$\resolve{\calC}{q}\defeq\formula{q\setminus q_{0}}\cup q_{1}.$$
Notice that $\resolve{\calC}{q}$ is unique up to a renaming of the variable $u$ and the relation names in $q_{1}$.
\end{definition}

\begin{example}
Let $q$ be the query of Fig.~\ref{fig:attmar}.
Let $\calC$ be the cycle $x,w,y,x$ in the Markov graph of $q$.
Using the notation of Definition~\ref{def:resolve},
we have
\begin{eqnarray*}
q_{0} & = & \{R(\underline{x},y,v), S(\underline{y},x), W(\underline{w},v)\}\\
q_{1} & = & \{T(\underline{u},x,w,y,v), U_{1}^{c}(\underline{x},u), U_{2}^{c}(\underline{w},u), U_{3}^{c}(\underline{y},u)\}
\end{eqnarray*}
Hence,
$\resolve{\calC}{q}=\{V_{1}^{c}(\underline{v},w)$, $V_{2}^{c}(\underline{w},y)$, $T(\underline{u},x,w,y,v)$, $U_{1}^{c}(\underline{x},u)$, $U_{2}^{c}(\underline{w},u)$, $U_{3}^{c}(\underline{y},u)\}$.
\end{example}

\begin{lemma}\label{lem:soluble}
Let $q$ be a self-join-free Boolean conjunctive query such that every atom with mode $i$ in $q$ is simple-key.
Let $\calC$ be an elementary directed cycle in the Markov graph of $q$, and let $q^{*}=\resolve{\calC}{q}$.
Then, there exists a polynomial-time many-one reduction from $\cqa{q}$ to $\cqa{q^{*}}$.
\end{lemma}

The reduction of Lemma~\ref{lem:soluble} will be explained in Section~\ref{subsec:soluble}.
To use the reduction in a proof of Theorem~\ref{the:ptime},
two more results are needed:
\begin{itemize}
\item
First, we need to show that the ``dissolution" of Markov cycles can be done while keeping the attack graph free of strong cycles (this is Lemma~\ref{lem:nonewstrongcycles}).
This turns out to be true only for Markov cycles that are premier (as defined in Definition~\ref{def:markov}).
\item
Second, we need to show the existence of premier Markov cycles that can be ``dissolved" (this is Lemma~\ref{lem:markov_cycle}).
\end{itemize}

\begin{lemma}\label{lem:nonewstrongcycles}
Let $q$ be a self-join-free Boolean conjunctive query such that every atom with mode $i$ in $q$ is simple-key.
Let $\calC$ be an elementary directed cycle in the Markov graph of $q$ such that $\calC$ is premier, and let $q^{*}=\resolve{\calC}{q}$.
If the attack graph of $q$ contains no strong cycle,
then the attack graph of $q^{*}$ contains no strong cycle either.
\end{lemma}

\begin{lemma}
\label{lem:markov_cycle}
Let $q$ be a self-join-free Boolean conjunctive query such that
\begin{itemize}
\item
for every atom $F\in q$, if $F$ has mode $i$,
then $F$ is simple-key and $\keyvars{F}\neq\emptyset$;
\item
$q$ is saturated; 
\item
the attack graph of $q$ contains no strong cycle; and
\item
the attack graph of $q$ contains an initial strong component with two or more atoms.
\end{itemize}
Then, the Markov graph of $q$ contains an elementary directed cycle that is premier and such that for every $y$ in $\calC$,
$\clutch{y}{q}\neq\emptyset$.
\end{lemma}

The condition $\clutch{y}{q}\neq\emptyset$ , for every $y$ in $\calC$, guarantees that $\resolve{\calC}{q}$ will contain strictly less atoms of mode $i$ than $q$.
This condition will be used in the proof of Theorem~\ref{the:ptime} which runs by induction on the number of atoms with mode $i$.
The following example shows that Lemma~\ref{lem:markov_cycle} is no longer true if $q$ is not saturated.

\begin{example} \label{ex:markov_path}
Continuing Example~\ref{ex:saturated}.
The query $q$ of Example~\ref{ex:saturated} is not saturated, but satisfies all other conditions in the statement of Lemma~\ref{lem:markov_cycle}.
In particular, the attack graph of $q$ contains a weak cycle $R\attacks{q}U\attacks{q}R$,
which is part of an initial strong component.
The Markov graph of $q$ consists of a single path $w\markovarg{q}x\markovarg{q}y\markovarg{q}z$, and hence is acyclic.

The query $q'$ of Example~\ref{ex:saturated} is saturated, and we have $x\markovarg{q'}w\markovarg{q'}x$, a Markov cycle which can be shown to be premier.
\end{example}

\subsection{The Proof of Theorem~\ref{the:ptime}}\label{subsec:final}

\begin{proof}[of Theorem~\ref{the:ptime}]
Assume that the attack graph of $q$ contains no strong cycle.
The proof runs by induction on increasing $\icard{q}$.
The desired result is obvious if $\icard{q}=0$.
Assume that $\icard{q}>0$ in the remainder of the proof.
Let $\db$ be an uncertain database that is input to $\cqa{q}$.

First, we reduce in polynomial time $\cqa{q}$ to $\cqa{q'}$ with $q'$ like in Lemma~\ref{lem:mostsimplified}.
We now distinguish two cases.

\paragraph{Case $q'$ contains an atom $F$ with mode $i$ that has zero indegree in the attack graph of $q$.}
We can assume either $F=R(\underline{x},\vec{y})$ or $F=R(\underline{a},\vec{y})$, where $\vec{y}$ is a sequence of distinct variables.
In the remainder, we treat the case  $F=R(\underline{x},\vec{y})$ (the case  $F=R(\underline{a},\vec{y})$ is even simpler).

Let $q''=q'\setminus\{R(\underline{x},\vec{y})\}$.
By Lemma~\ref{lem:inFO}, every repair of $\db$ satisfies $q'$ if and only if $\db$ includes an $R$-block $\block$ (there are only polynomially many such blocks) such for every $R(\underline{a},\vec{b})\in\block$,
every repair of $\db$ satisfies $\substitute{q''}{x,\vec{y}}{a,\vec{b}}$.
By Lemma~\ref{lem:nonewcycles},  the attack graph of $\substitute{q''}{x,\vec{y}}{a,\vec{b}}$ contains no strong cycle.
From $\icard{\substitute{q''}{x,\vec{y}}{a,\vec{b}}}=\icard{q'}-1<\icard{q}$, it follows that $\cqa{\substitute{q''}{x,\vec{y}}{a,\vec{b}}}$ is in $\P$ by the induction hypothesis.
It follows that $\cqa{q}$ is in $\P$ as well.

\paragraph{Case every atom $F$ with mode $i$ in $q'$ has an incoming attack in the attack graph of $q'$.}
It will be the case that no constant occurs in an atom of mode $i$ in $q'$.

Then, the attack graph of $q'$ must contain an initial strong component with two or more atoms.
By Lemma~\ref{lem:markov_cycle}, the Markov graph of $q'$ contains an elementary directed cycle $\calC$ that is premier and such that for every $y$ in $\calC$, $\clutch{y}{q'}\neq\emptyset$. 
By Lemma~\ref{lem:soluble}, we can reduce in polynomial time $\cqa{q'}$ to $\cqa{q^{*}}$ where $q^{*}=\resolve{\calC}{q'}$. 
Since the attack graph of $q'$ contains no strong cycle, it follows by Lemma~\ref{lem:nonewstrongcycles} that the attack graph of $q^{*}$ contains no strong cycle either.

Let $k\geq 2$ be the size of $\calC$.
It can be easily verified that $\icard{q^{*}}\leq\formula{\icard{q'}-k}+1<\icard{q'}$.
By the induction hypothesis,  $\cqa{q^{*}}$ is in $\P$.
Since there exists a polynomial-time reduction from $\cqa{q}$ to $\cqa{q^*}$, we conclude that $\cqa{q}$ is in $\P$ as well.
\end{proof}

\subsection{The Reduction of Lemma~\ref{lem:soluble}}\label{subsec:soluble}

This section first describes the reduction of Lemma~\ref{lem:soluble}, and then proves the lemma.

\paragraph{Relevance of subsets of repairs}
In Section~\ref{sec:preliminaries}, we distinguished database facts that are relevant for a query from those that are not. 
This notion is extended next.

\begin{definition}
Let $q$ be a self-join-free Boolean conjunctive query, and let $\db$ be an uncertain database.
A consistent subset $\sep$ of $\db$ is said to be {\em grelevant for $q$ in $\db$\/} (generalized relevant) if it can be extended into a repair $\rep$ of $\db$ such that some fact of $\sep$ is relevant for $q$ in $\rep$.
\end{definition}

It can be seen that $A\in\db$ is relevant for $q$ in $\db$ if and only if $\{A\}$ is grelevant for $q$ in $\db$.
Therefore, ``grelevant" is a notion that generalises ``relevant."

\begin{lemma}\label{lem:irrelevant}
Let $q$ be a self-join-free Boolean conjunctive query, and let $\db$ be an uncertain database.
Let $\sep$ be a consistent subset of $\db$ that is not grelevant for $q$ in $\db$.
Let $\db_{0}=\bigcup\{\theblock{A}{\db}\mid A\in\sep\}$.
Then, the following are equivalent:
\begin{enumerate}
\item\label{it:irrelevantone}
every repair of $\db$ satisfies $q$;
\item\label{it:irrelevanttwo}
every repair of $\db\setminus\db_{0}$ satisfies $q$.
\end{enumerate}
\end{lemma}
\begin{proof}
\framebox{\ref{it:irrelevantone}$\implies$\ref{it:irrelevanttwo}}
By contraposition.
Let $\rep$ be a repair of $\db\setminus\db_{0}$ that falsifies $q$.
Then, $\rep\cup\sep$ is a repair of $\db$.
If $\rep\cup\sep\models q$, then it must be the case that $\sep$ is grelevant for $q$ in $\db$, a contradiction. 
We conclude by contradiction that $\rep\cup\sep\not\models q$.
\framebox{\ref{it:irrelevanttwo}$\implies$\ref{it:irrelevantone}}
Trivial.
\end{proof}

\paragraph{Introductory example}
The following example illustrates the main ideas behind the reduction of Lemma~\ref{lem:soluble}.

\begin{example}\label{ex:frugal}
Let $q$ be a self-join-free Boolean conjunctive query.
Assume that $q$ includes $q_{0}=\{R(\underline{x},y)$, $S(\underline{y},z)$, $V(\underline{z},x)\}$.
Then, the Markov graph of $q$ contains a cycle $x \markov y \markov z \markov x$.
Let $\db$ be an uncertain database that is purified relative to $q$.
Let $\db_{0}$ be the subset of $\db$ containing all $R$-facts, $S$-facts, and $V$-facts of $\db$. 
Assume that the following three tables represent all facts of $\db_{0}$
(for convenience, we use variables as attribute names,
and we blur the distinction between a relation name $R$ and a table representing a set of $R$-facts).
$$
\begin{array}{llll}
\begin{array}{c|cc}
R & \underline{x} & y\\\cline{2-3}
  & 1 & a\\
	&   &  \\[1.5ex]  
	& 2 & b\\
	& 2 & c\\[1.5ex]  
	& 3 & d\\
	& 3 & e\\
	& 4 & e\\
	& 4 & f
\end{array}
&
\begin{array}{c|cc}
S & \underline{y} & z\\\cline{2-3}
  & a & \alpha\\
	& a & \kappa\\[1.5ex]   
	& b & \beta\\
	& c & \gamma\\[1.5ex]
	& d & \delta\\
	& e & \epsilon\\
	& e & \delta\\
	& f & \phi
\end{array}
&
\begin{array}{c|cc}
V & \underline{z} & x\\\cline{2-3}
  & \alpha   & 1\\
	& \kappa   & 1\\[1.5ex]  
	& \beta    & 2\\
	& \gamma   & 2\\[1.5ex]  
	& \delta   & 3\\
	& \epsilon & 3\\
	& \delta   & 4\\
	& \phi     & 4
\end{array}
&
\begin{array}{l}
\\
\MyLBrace{2.8ex}{\db_{01}}\\[2.5ex]
\MyLBrace{2.8ex}{\db_{02}}\\[3.0ex]
\MyLBrace{6ex}{\db_{03}}
\end{array}
\end{array}
$$

As indicated, we can partition $\db_{0}$ into three subsets $\db_{01}$, $\db_{02}$, and $\db_{03}$ whose active domains have, pairwise, no constants in common.
Consider each of these three subsets in turn.
\begin{enumerate}
\item
$\db_{01}$ has two repairs, each of which satisfies $q_{0}$.
For every repair $\rep$ of $\db$, 
either $\rep\models\substitute{q_{0}}{x,y,z}{1,a,\alpha}$ or  $\rep\models\substitute{q_{0}}{x,y,z}{1,a,\kappa}$.
\item
$\db_{02}$ has two repairs, each of which satisfies $q_{0}$.
For every repair $\rep$ of $\db$, 
either $\rep\models\substitute{q_{0}}{x,y,z}{2,b,\beta}$ or  $\rep\models\substitute{q_{0}}{x,y,z}{2,c,\gamma}$.
\item
$\db_{03}$ has $16$ repairs, and for 
$\sep\defeq\{R(\underline{3},d)$, $S(\underline{d},\delta)$, $V(\underline{\delta},4)$,
$R(\underline{4},e)$, $S(\underline{e},\epsilon)$, $V(\underline{\epsilon},3)$,
$S(\underline{f},\phi)$, $V(\underline{\phi},4)\}$,
we have that $\sep$ is a repair of $\db_{03}$ that falsifies $q_{0}$.
It can be seen that $\sep$ is not grelevant for $q$ in $\db$.
Then, by Lemma~\ref{lem:irrelevant}, every repair of $\db$ satisfies $q$
if and only if every repair of $\db\setminus\db_{03}$ satisfies $q$.
That is, $\db_{03}$ can henceforth be ignored.
\end{enumerate}  
The following table $T$ summarizes our findings.
In the first column (named with a fresh variable $u$), the values $01$ and $02$ refer to $\db_{01}$ and $\db_{02}$ respectively.
The table includes two blocks (separated by a dashed line for clarity).
The first block indicates that for every repair $\rep$ of $\db$,
either $\rep\models\substitute{q_{0}}{x,y,z}{1,a,\alpha}$ or $\rep\models\substitute{q_{0}}{x,y,z}{1,a,\kappa}$.
Likewise for the second block.

$$
\begin{array}{c|*{4}{c}}
T & \underline{u} & x & y & z\\\cline{2-5}
  & 01 & 1 & a & \alpha\bigstrut\\ 
	& 01 & 1 & a & \kappa\\\cdashline{2-5}
	& 02 & 2 & b & \beta\bigstrut\\
	& 02 & 2 & c & \gamma
\end{array}
$$
The table $U_{x}$ shown below is the projection of $T$ on attributes $x$ and $u$.
This table must be consistent, because by construction, the active domains of $\db_{01}$ and $\db_{02}$ are disjoint.
Likewise for $U_{y}$ and $U_{z}$.
$$
\begin{array}{llll}
\begin{array}{c|cc}
U_{x} & \underline{x} & u\\\cline{2-3}
  & 1 & 01\\
	& 2 & 02 
\end{array}
&
\begin{array}{c|cc}
U_{y} & \underline{y} & u\\\cline{2-3}
	& a & 01\\
	& b & 02\\
	& c & 02
\end{array}
&
\begin{array}{c|cc}
U_{z} & \underline{z} & u\\\cline{2-3}
  & \alpha  & 01\\
	& \kappa  & 01\\
	& \beta   & 02\\
	& \gamma  & 02
\end{array}
\end{array}
$$

Let $\db'$ be the database that extends $\db$ with all the facts shown in the tables $T$, $U_{x}$, $U_{y}$, and $U_{z}$.\footnote{Facts of $\db_{0}$ can be omitted from $\db'$, but that is not important.}
Let $q^*=\formula{q\setminus q_{0}}\cup\{T(\underline{u},x,y,z)$, $U_{x}^{c}(\underline{x},u)$, $U_{y}^{c}(\underline{y},u)$, $U_{z}^{c}(\underline{z},u)\}$.
From our construction, it follows that
every repair of $\db$ satisfies $q$ if and only if every repair of $\db'$ satisfies $q^*$.
\end{example}

\paragraph{Gblocks and gpurification}
The following definition strengthens the notion of purification introduced earlier in Section~\ref{sec:preliminaries}.

\begin{definition}\label{def:c_block}
Let $q$ be a self-join-free Boolean conjunctive query such that all atoms with mode $i$ in $q$ are simple-key.
Let $\db$ be an uncertain database that is purified and typed relative to $q$.
A {\em gblock\/} (generalized block) of $\db$ relative to $q$ is a maximal (with respect to $\subseteq$) subset $\gblock$ of $\db$ such that
all facts in $\gblock$ have mode $i$ and agree on their primary-key position (but may disagree on their relation name).
Notice that a gblock has at most polynomially many repairs (in the size of $\db$).\footnote{Indeed, since $\db$ is purified relative to $q$, every gblock of $\db$ contains at most $\card{q}$ distinct relation names,
and hence has at most $\card{\db}^{\card{q}}$ distinct repairs.}
We say that $\db$ is {\em gpurified relative to $q$} if for every gblock $\gblock$ of $\db$,
every repair of $\gblock$ is grelevant for $q$ in $\db$.
\end{definition}

Clearly, every gblock is the union of one or more blocks. 
Two facts of the same gblock have the same primary-key value, but can have distinct relation names.

\begin{example}
Let $q=\{R(\underline{x},y)$, $S(\underline{x},y)\}$.
Let $\db=\{R(\underline{a},1)$, $R(\underline{a},2)$, $S(\underline{a},1)$, $S(\underline{a},2)\}$.
Then, $\db$ is purified and typed relative to $q$.
All facts of $\db$ together constitute a gblock.
The uncertain database $\db$ is not gpurified, since 
$\sep=\{R(\underline{a},1)$, $S(\underline{a},2)\}$ is a repair of the gblock, and also a repair of $\db$.
However, neither $R(\underline{a},1)$ nor $S(\underline{a},2)\}$ is relevant for $q$ in $\sep$.
\end{example}

\begin{example}\label{ex:aggressive}
Let $q=\{R_{1}(\underline{x},y)$, $R_{2}(\underline{x},z)$, $S(\underline{y,z})\}$,
where the signature of $S$ is $\signature{2}{2}$.
Let $\db$ be the uncertain database containing the following facts.
$$
\begin{array}{ccc}
\begin{array}{c|cc}
R_{1} & \underline{x} & y\\\cline{2-3}
      & a & 1\\
	    & a & 2 
\end{array}
&
\begin{array}{c|cc}
R_{2} & \underline{x} & z\\\cline{2-3}
      & a & 3\\
	    & a & 4 
\end{array}
&
\begin{array}{c|cc}
S & \underline{y} & \underline{z}\\\cline{2-3}
  & 1 & 3\\
	& 2 & 4
\end{array}
\end{array}
$$
Then, $\db$ is purified and typed relative to $q$.
All $R_{1}$-facts and $R_{2}$-facts together constitute a gblock.
A repair of this gblock is $\sep=\{R_{1}(\underline{a},1)$, $R_{2}(\underline{a},4)\}$.
The uncertain database $\db$ is not gpurified.
Indeed, the only repair of $\db$ that extends $\sep$ is $\{R_{1}(\underline{a},1)$, $R_{2}(\underline{a},4)$, $S(\underline{1,3})$, $S(\underline{2,4})\}$ (call it $\rep$).
Neither $R_{1}(\underline{a},1)$ nor $R_{2}(\underline{a},4)$ is relevant for $q$ in $\rep$.
\end{example} 

The following lemma is similar to Lemma~\ref{lem:purified} and has an easy proof.

\begin{lemma}\label{lem:gpurified}
Let $q$ be a self-join-free Boolean conjunctive query such that all atoms with mode $i$ in $q$ are simple-key.
Let $\db$ be an uncertain database that is purified and typed relative to $q$.
It is possible to compute in polynomial time an uncertain database $\db'$ that is gpurified relative to $q$ such that
every repair of $\db$ satisfies $q$ if and only if every repair of $\db'$ satisfies $q$.
\end{lemma}

\paragraph{Specification of the reduction of Lemma~\ref{lem:soluble}}
Let $q$ and $\calC$ be as in the statement of Lemma~\ref{lem:soluble}.
Assume that the elementary directed cycle $\calC$ in the Markov graph of $q$ is $x_0 \markov x_1 \dotsm \markov x_{k-1} \markov x_0$.
In what follows, let $\resolve{\calC}{q}$ be as in Definition~\ref{def:resolve},
with $q_{0}$, $q_{1}$, $\vec{y}$, $u$, $T$, and $U_{0},\dots,U_{k-1}$ as defined there.
Moreover,  we write $\oplus$ for addition modulo $k$, and $\ominus$ for subtraction modulo $k$.
For every $i\in\{0,\dots,k-1\}$, we define $X_{i}$ as follows:
$$
X_{i}\defeq\atomvars{\clutch{x_{i}}{q}}.
$$
The reduction of Lemma~\ref{lem:soluble} will be described under the following simplifying assumptions which can be made without loss of generality:
\begin{itemize}
\item
every uncertain database $\db$ that is input to $\cqa{q}$ is typed, purified, and gpurified relative to $q$.
This assumption is without loss of generality as argued in Section~\ref{sec:preliminaries}, and by Lemmas~\ref{lem:purified} and~\ref{lem:gpurified}; and
\item
for every $i\in\{0,\dots,k-1\}$, no atom of $\clutch{x_{i}}{q}$ contains constants or double occurrences of the same variable.
This assumption is without loss of generality by Lemma~\ref{lem:mostsimplified}.
\end{itemize}

Under these notations and assumptions, we describe the reduction of Lemma~\ref{lem:soluble}.
Let $\db$ be an uncertain database that is input to $\cqa{q}$.
Define a directed $k$-partite graph, denoted $\dbgraph{\db}$, as follows:
\begin{enumerate}
\item
the vertex set of $\dbgraph{\db}$ is $\bigcup_{i=0}^{k-1}\type{x_{i}}$; and
\item
there is a directed edge from $a\in\type{x_{i}}$ to $b\in\type{x_{i\oplus 1}}$ if for some valuation $\theta$ over $\queryvars{q}$, we have that $\theta(q)\subseteq\db$ and $\theta(x_{i})=a$ and $\theta(x_{i\oplus 1})=b$.
In this case, we say that $\theta[X_{i}]$ {\em realizes\/} the edge $(a,b)$,
where $\theta[X_{i}]$ denotes the restriction of $\theta$ on $X_{i}$.
\end{enumerate}
Notice that distinct valuations can realize the same edge of $\dbgraph{\db}$ (but if $\db$ is consistent, then every edge in $\dbgraph{\db}$ is realized at most once).

\begin{example}
Let
$q =\{R_{1}(\underline{x_{0}},y_{1})$, $R_{2}(\underline{x_{0}},y_{2})$, $S^{c}(\underline{y_{1},y_{2}},x_{1})$, $R_{3}(\underline{x_{0}},y_{3})$, $V(\underline{x_{1}},x_{0})\}$.
Then, $x_{0}\markov x_{1}$ and $X_{0}=\{x_{0},y_{1},y_{2},y_{3}\}$.
%
Assume an uncertain database $\db$ containing, among others, the following facts.
$$
\begin{array}{cccc}
\begin{array}{c|cc}
R_{1} & \underline{x_{0}} & y_{1}\\\cline{2-3}
      & a & c_{1}
\end{array}
&
\begin{array}{c|cc}
R_{2} & \underline{x_{0}} & y_{2}\\\cline{2-3}
      & a & c_{2}\\
			& a & c_{3}
\end{array}
&
\begin{array}{c|ccc}
S & \underline{y_{1}} & \underline{y_{2}} & x_{1}\\\cline{2-4}
  & c_{1} & c_{2} & 1\\
	& c_{1} & c_{3} & 1
\end{array}
&
\begin{array}{c|cc}
R_{3} & \underline{x_{0}} & y_{3}\\\cline{2-3}
      & a & \beta\\
	    & a & \gamma 
\end{array}
\end{array}
$$
The graph $\dbgraph{\db}$ contains a directed edge $(a,1)$, which is realized by $\{x_{0}\mapsto a$, $y_{1}\mapsto c_{1}$, $y_{2}\mapsto c_{2}$, $y_{3}\mapsto\beta\}$.
The edge $(a,1)$ is also realized by $\{x_{0}\mapsto a$, $y_{1}\mapsto c_{1}$, $y_{2}\mapsto c_{3}$, $y_{3}\mapsto\gamma\}$.
\end{example}

Let $\consistent{\db}$ be the subset of $\db$ that contains all facts with mode $c$.
Significantly, the edges in $\dbgraph{\db}$ outgoing from some constant $a\in\type{x_{j}}$ (for some $j\in\{0,\dots,k-1\}$)
are fully determined by $\consistent{\db}$ and the gblock of $\db$ containing all facts whose relation name is in $\clutch{x_{j}}{q}$
and whose primary-key position contains the constant $a$ (call this gblock $\gblock_{a}$).
Since $\db$ is gpurified, for every repair $\sep$ of $\gblock_{a}$, there exists a unique constant $b\in\type{x_{j\oplus 1}}$ such that 
$$
\sep\cup\consistent{\db}\models\substitute{\formula{\clutch{x_{j}}{q}\cup\consistent{q}}}{x_{j},x_{j\oplus 1}}{a,b},
$$
in which case $\dbgraph{\db}$ will contain a directed edge from $a$ to $b$.
Uniqueness of $b$ follows from $\FD{\clutch{x_{j}}{q}\cup\consistent{q}}\models\fd{x_{j}}{x_{j\oplus 1}}$ and~\cite[Lemma~4.3]{DBLP:journals/tods/Wijsen12}.

Since $\db$ is gpurified, $\dbgraph{\db}$ is a vertex-disjoint union of strong components such that no edge leads from one strong component to another strong component 
(i.e., all strong components are initial).\footnote{Strong components are defined by Definition~\ref{def:tarjan}.}
In what follows, let $D$ be a strong component of $\dbgraph{\db}$.
Since $\dbgraph{\db}$ is $k$-partite, the length of any cycle in $\dbgraph{\db}$ must be a multiple of $k$, i.e., must be in $\{k,2k,3k,\dots\}$.
Let $\db_{D}$ be the subset of $\db$ that contains $R(\underline{a},\vec{b})$ whenever $R$ is of mode $i$ and the constant $a$ is a vertex in $D$ (and $\vec{b}$ is any sequence of constants).
Obviously, every block of $\db$ is either included in $\db_{D}$ or disjoint with $\db_{D}$.

Clearly, $D$ must contain a cycle.
Among the cycles in $D$ of length exactly $k$, we now distinguish the cycles that {\em support\/} $q$ from those that do not, as defined next.
Let such cycle in $D$ be 
\begin{equation}\label{eq:datacycle}
a_{0},a_{1},\dots,a_{k-1},a_{0}
\end{equation}
where for $i\in\{0,\dots,k-1\}$, $a_{i}\in\type{x_{i}}$.
For $i\in\{0,\dots,k-1\}$, let $\Delta_{i}$ be the set of all valuations over $X_{i}$ that realize $(a_{i},a_{i\oplus 1})$.
We say that the cycle~(\ref{eq:datacycle}) {\em supports $q$} if for for all $i,j\in\{0,\dots,k-1\}$,
for all $\mu_{i}\in\Delta_{i}$ and $\mu_{j}\in\Delta_{j}$, it is the case that $\mu_{i}$ and $\mu_{j}$ agree on all variables in $X_{i}\cap X_{j}$.
Notice that $X_{i}\cap X_{j}$ can be empty.
The cycle~(\ref{eq:datacycle}) may not support $q$, because $\mu_{i}$ and $\mu_{j}$ can  disagree on variables in $X_{i}\cap X_{j}\cap\sequencevars{\vec{y}}$, as illustrated next.

\begin{example}\label{ex:disagree}
Let $q=\{R(\underline{x_{0}},x_{1},y)$, $S(\underline{x_{1}},x_{0},y)\}$.
We have $x_{0}\markov x_{1}\markov x_{0}$.
Let $\db$ be the uncertain database containing the following facts.
$$
\begin{array}{cc}
\begin{array}{c|ccc}
R & \underline{x_{0}} & x_{1} & y\\\cline{2-4}
  & a & 1 & \alpha\\
	& a & 1 & \beta
\end{array}
&
\begin{array}{c|ccc}
S & \underline{x_{1}} & x_{0} & y\\\cline{2-4}
  & 1 & a & \alpha\\
	& 1 & a & \beta
\end{array}
\end{array}
$$
The edge set of $\dbgraph{\db}$ is $\{(a,1)$, $(1,a)\}$.
Both $(a,1)$ and $(1,a)$ are realized by the valuations $\{x_{0}\mapsto a$, $x_{1}\mapsto 1$, $y\mapsto\alpha\}$ and $\{x_{0}\mapsto a$, $x_{1}\mapsto 1$, $y\mapsto\beta\}$, which disagree on $y$. Hence, the cycle $a,1,a$ does not support $q$.
\end{example} 

%

On the other hand, we can assume without loss of generality that $\mu_{i}$ and $\mu_{j}$ agree on all variables in $X_{i}\cap X_{j}\cap\{x_{0},\dots,x_{k-1}\}$.
In particular, if $x_{i}\in X_{j}$, then $\mu_{j}(x_{i})=\mu_{i}(x_{i})=a_{i}$. 
To see why this is the case, assume that $x_{i}\in X_{j}$, where $i,j\in\{0,\dots,k-1\}$ and $i\neq j$.
Then, it must be that $x_{j}\markov x_{i}$.
Two cases can occur:
\begin{itemize}
\item
if $j=i\ominus 1$, then $\mu_{j}$ realizes the edge $(a_{i\ominus 1},a_{i})$ and $\mu_{j}(x_{i})=a_{i}$; and
\item
if $j\neq i\ominus 1$, then $x_{j}\markov x_{i}\markov x_{i\oplus 1}\dotsm\markov x_{j\ominus 1}\markov x_{j}$ is a shorter Markov cycle.
\end{itemize}
The second case can be avoided by picking $\calC$ to be the shorter cycle, as illustrated by Example~\ref{ex:shorter}. 
It can be seen that such choice of $\calC$ is without loss of generality.
In particular, in Lemma~\ref{lem:markov_cycle}, if $\calC$ was premier, then the shorter cycle will also be premier.

\begin{example}\label{ex:shorter}
Let
$q =\{R(\underline{x_{0}},x_{1})$, $S(\underline{x_{1}},x_{2},x_{0})$, $V(\underline{x_{2}},x_{0})\}$.
Then, $x_{0}\markov x_{1}\markov x_{2}\markov x_{0}$.
We have $X_{0}=\{x_{0},x_{1}\}$, $X_{1}=\{x_{1},x_{2},x_{0}\}$, and $X_{2}=\{x_{2},x_{0}\}$.
Assume an uncertain database $\db$ with the following facts.
$$
\begin{array}{ccc}
\begin{array}{c|cc}
R & \underline{x_{0}} & x_{1}\\\cline{2-3}
  & a & 1\\
	& b & 1		
\end{array}
&
\begin{array}{c|ccc}
S & \underline{x_{1}} & x_{2} & x_{0}\\\cline{2-4}
  & 1 & \beta & a\\
  & 1 & \beta & b
\end{array}
&
\begin{array}{c|cc}
V & \underline{x_{2}} & x_{0}\\\cline{2-3}
  & \beta & a\\
	& \beta & b 
\end{array}
\end{array}
$$
The graph $\dbgraph{\db}$ contains an elementary directed cycle $a,1,\beta,a$.
The edge $(a,1)$ is realized by $\mu_{0}=\{x_{0}\mapsto a$, $x_{1}\mapsto 1\}$.
The edge $(1,\beta)$ is realized, among others, by $\mu_{1}=\{x_{1}\mapsto 1$, $x_{2}\mapsto\beta$, $x_{0}\mapsto b\}$.
Notice that $\mu_{0}$ and $\mu_{1}$ disagree on $x_{0}$.
Although it is easy to deal with this situation where two valuations disagree on a variable in the Markov cycle,
it is even easier to avoid this situation by working with the shorter Markov cycle  $x_{0}\markov x_{1}\markov x_{0}$.
\end{example}

We now distinguish two cases.

\paragraph{
Case $D$ contains either an elementary directed cycle of size $k$ that does not support $q$,
or an elementary directed cycle of size strictly greater than $k$.}
We show in the next paragraph how to construct a repair $\sep$ of $\db_{D}$ such that $\sep$ is not grelevant for $q$ in $\db$. 
Then, by Lemma~\ref{lem:irrelevant}, 
every repair of $\db$ satisfies $q$ if and only if every repair of $\db\setminus\db_{D}$ satisfies $q$. 
In this case, the reduction deletes from $\db$ all facts of $\db_{D}$. 

The construction of $\sep$ proceeds as follows.
Pick an elementary cycle in $D$ that has size strictly greater than $k$,
or that has size $k$ but does not support $q$.
The cycle picked will henceforth be denoted by $\calE$.
Construct a maximal sequence
$$(V_{0},E_{0}),b_{1},(V_{1},E_{1}),b_{2},(V_{2},E_{2}),\dots,b_{n},(V_{n},E_{n})$$
where 
\begin{enumerate}
\item 
$V_{0}$ is the set of vertices in $\calE$, and $E_{0}$ is the set of directed edges in $\calE$; and
\item
for every $i\in\{1,\dots,n\}$,
\begin{enumerate}
\item
$b_{i}\not\in V_{i-1}$ and for some $c\in V_{i-1}$, $(b_{i},c)$ is a directed edge in $\dbgraph{\db}$; and
\item
$V_{i}=V_{i-1}\cup\{b_{i}\}$ and $E_{i}=E_{i-1}\cup\{(b_{i},c)\}$.
\end{enumerate}
\end{enumerate}
The resulting graph $(V_{n},E_{n})$ is such that $V_{n}$ is equal to the vertex set of $D$,
and $E_{n}$ contains exactly one outgoing edge for each vertex in $V_{n}$.
The graph $(V_{n},E_{n})$ contains no directed cycle other than $\calE$.
To construct $\sep$, for each $j\in\{0,\dots,k-1\}$,
for each vertex $a\in V_{n}\cap\type{x_{j}}$, select some valuation $\mu$ that realizes the edge in $E_{n}$ outgoing from $a$,
and add $\mu(\clutch{x_{j}}{q})$ to $\sep$.
If $\calE$ has size $k$, then the valuations $\mu$ should be selected such that for some vertices $a,b$ in $\calE$,
the valuations chosen for $a$ and $b$ disagree on some variable of $\sequencevars{\vec{y}}$.
It is not hard to see that the set $\sep$ so obtained is a repair of $\db_{D}$ that is not grelevant for $q$ in $\db$.

We illustrate the above construction by two examples.

\begin{example}
In Example~\ref{ex:disagree}, one can choose $\sep=\{R(\underline{a},1,\alpha)$, $S(\underline{1},a,\beta)\}$.
The treatment of a directed cycle of size strictly greater than $k$ is illustrated by $\db_{03}$ in Example~\ref{ex:frugal}.
\end{example}

\begin{example}\label{ex:involved}
Let $q=\{R(\underline{x_{0}},y_{1},y_{2})$, 
$V(\underline{x_{1}},y_{2})$,
$S_{1}^{c}(\underline{y_{1},y_{2}},x_{1})$,
$S_{2}^{c}(\underline{y_{2}},x_{0})\}$.
We have $x_{0}\markov x_{1}\markov x_{0}$, $X_{0}=\{x_{0},y_{1},y_{2}\}$, and $X_{1}=\{x_{1},y_{2}\}$.
Let $\db$ be an uncertain database with the following facts.

$$
\begin{array}{cccc}
\begin{array}{c|ccc}
R & \underline{x_{0}} & y_{1} & y_{2}\\\cline{2-4}
  & a & 1 & 2\\
	& a & 3 & 4\\
	& a & 1 & 6\\
\end{array}
&
\begin{array}{c|cc}
V & \underline{x_{1}} & y_{2}\\\cline{2-3}
  & \gamma & 2\\
	& \gamma & 4\\
	& \beta  & 6
\end{array}
&
\begin{array}{c|ccc}
S_{1}^{c} & \underline{y_{1}} & \underline{y_{2}} & x_{1}\\\cline{2-4}
  & 1 & 2 & \gamma\\
	& 3 & 4 & \gamma\\
	& 1 & 6 & \beta\\
\end{array}
&
\begin{array}{c|cc}
S_{2}^{c} & \underline{y_{2}} & x_{0}\\\cline{2-3}
  & 2 & a\\
	& 4 & a\\
	& 6 & a
\end{array}
\end{array}
$$

The following table lists the edges in $\dbgraph{\db}$, by type,
along with the valuations that realize each edge.
$$
\begin{array}{ll}
\begin{array}{l|l@{=}c}
\multicolumn{2}{l}{\textrm{Edges in $\type{x_{0}}\times\type{x_{1}}$}}\\
\multicolumn{1}{c|}{\textrm{Edge}} & \multicolumn{1}{c}{\textrm{Realized by}}\\\hline
(a,\gamma)  & \{x_{0}\mapsto a, y_{1}\mapsto 1, y_{2}\mapsto 2\} & \mu_{1}\\
            & \{x_{0}\mapsto a, y_{1}\mapsto 3, y_{2}\mapsto 4\} & \mu_{2}\\\hline
(a,\beta)   & \{x_{0}\mapsto a, y_{1}\mapsto 1, y_{2}\mapsto 6\} & \mu_{3}\\							

\end{array}
&
\begin{array}{l|l@{=}c}
\multicolumn{2}{l}{\textrm{Edges in $\type{x_{1}}\times\type{x_{0}}$}}\\
\multicolumn{1}{c|}{\textrm{Edge}} & \multicolumn{1}{c}{\textrm{Realized by}}\\\hline
(\gamma,a)  & \{x_{1}\mapsto\gamma, y_{2}\mapsto 2\} & \mu_{4}\\
            & \{x_{1}\mapsto\gamma, y_{2}\mapsto 4\} & \mu_{5}\\\hline
(\beta,a)   & \{x_{1}\mapsto\beta, y_{2}\mapsto 6\} & \mu_{6}\\							

\end{array}
\end{array}
$$
Then, $\dbgraph{\db}$ contains two elementary cycles, $a,\gamma,a$ and $a,\beta,a$, both of length~$2$.
The cycle $a,\beta,a$ supports $q$.
The cycle $a,\gamma,a$ does not support $q$, because $\mu_{1}$ and $\mu_{5}$ disagree on $y_{2}$.
Therefore, the edges $(a,\gamma)$ and $(\gamma,a)$, along with $\mu_{1}$ and $\mu_{5}$, will be used in the construction of a consistent set $\sep$ that is not grelevant for $q$ in $\db$.
For the remaining vertex $\beta$, we add the edge $(\beta,a)$, which is only realized  by $\mu_{6}$.
Then, $\sep$ contains the $R$-fact $R(\underline{a},1,2)$ (because of $\mu_{1}$), and the $V$-facts  $V(\underline{\gamma},4)$ and $V(\underline{\beta},6)$ (because of $\mu_{5}$ and $\mu_{6}$ respectively).
In this example, there is only one repair that contains $\sep$, and this repair falsifies $q$.
\end{example}

\paragraph{Case every elementary directed cycle in $D$ has length $k$ and supports $q$.}
In this case, we will encode each cycle of $D$ as a set of $T$-facts, as follows.
Consider any cycle of the form~(\ref{eq:datacycle}) in $D$, and take the cross product 
\begin{equation}\label{eq:crossproduct}
\Delta_{0}\times\Delta_{2}\times\dots\times\Delta_{k-1},
\end{equation}
which is of polynomial size (in the size of $\db$).
Since we are in the case where any cycle of the form~(\ref{eq:datacycle}) supports $q$,
for every tuple $(\mu_{0},\mu_{1},\dots,\mu_{k-1})$ in the cross product~(\ref{eq:crossproduct}), the set $\mu\defeq\bigcup_{i=0}^{k-1}\mu_{i}$ is a well defined valuation over $\{x_{0},\dots,x_{k-1}\}\cup\sequencevars{\vec{y}}$.
In this case, for each such tuple, the reduction adds the following $k+1$ facts:
$$
\begin{array}{l}
T(\underline{D}, a_{0}, \dots, a_{k-1}, \mu(\vec{y}))\\
U_{0}^{c}(\underline{a_{0}}, D)\\ 
\phantom{U_{k-1}^{c}}\vdots\\
U_{k-1}^{c}(\underline{a_{k-1}}, D)
\end{array}
$$
in which $D$ is used as a constant.
Recall that $a_{i}=\mu(x_{i})$ for $i\in\{0,\dots,k-1\}$.
Notice that if the sequence $\vec{y}$ is empty, then the reduction will add exactly one $T$-fact for every cycle of the form~(\ref{eq:datacycle}). 
Otherwise, the reduction may add multiple $T$-facts for the same cycle, as illustrated next. 

\begin{example}
Let $q=\{R(\underline{x_{0}},x_{1},y)$, $S(\underline{x_{1}},x_{0})\}$.
We have $x_{0}\markov x_{1}\markov x_{0}$, $X_{0}=\{x_{0},x_{1},y\}$ and  $X_{1}=\{x_{0},x_{1}\}$.
Let $\db$ be the uncertain database containing the following facts.
$$
\begin{array}{cc}
\begin{array}{c|ccc}
R & \underline{x_{0}} & x_{1} & y\\\cline{2-4}
  & a & 1 & \alpha\\
	& a & 1 & \beta
\end{array}
&
\begin{array}{c|cc}
S & \underline{x_{1}} & x_{0}\\\cline{2-3}
  & 1 & a
\end{array}
\end{array}
$$
The edge set of $\dbgraph{\db}$ is $\{(a,1)$, $(1,a)\}$.
The edge $(a,1)$ is realized by both $\{x_{0}\mapsto a$, $x_{1}\mapsto 1$, $y\mapsto\alpha\}$ and $\{x_{0}\mapsto a$, $x_{1}\mapsto 1$, $y\mapsto\beta\}$.
The edge $(1,a)$ is realized only by $\{x_{0}\mapsto a$, $x_{1}\mapsto 1\}$.
The cycle $a,1,a$ in $\dbgraph{\db}$ supports $q$.
The reduction will add the following $T$-facts (for some identifier $D$):
$$
\begin{array}{c|cccc}
T & \underline{u} & x_{0} & x_{1} & y\\\cline{2-5}
  & D & a & 1 & \alpha\\
	& D & a & 1 & \beta
\end{array}
$$
\end{example} 

\begin{example}
Take the query $q$ of Example~\ref{ex:involved}, with the following uncertain database $\db$. 

$$
\begin{array}{cccc}
\begin{array}{c|ccc}
R & \underline{x_{0}} & y_{1} & y_{2}\\\cline{2-4}
  & a & 1 & 2\\
	& a & 1 & 6\\
  & a & 3 & 6
\end{array}
&
\begin{array}{c|cc}
V & \underline{x_{1}} & y_{2}\\\cline{2-3}
  & \gamma & 2\\
	& \beta  & 6
\end{array}
&
\begin{array}{c|ccc}
S_{1}^{c} & \underline{y_{1}} & \underline{y_{2}} & x_{1}\\\cline{2-4}
  & 1 & 2 & \gamma\\
	& 1 & 6 & \beta\\
	& 3 & 6 & \beta
\end{array}
&
\begin{array}{c|cc}
S_{2}^{c} & \underline{y_{2}} & x_{0}\\\cline{2-3}
  & 2 & a\\
	& 6 & a
\end{array}
\end{array}
$$
Then, $\dbgraph{\db}$ contains two elementary cycles, $a,\gamma,a$ and $a,\beta,a$, both of length~$2$ and both supporting $q$.
The reduction will add the following $T$-facts (for some identifier $D$):
$$
\begin{array}{c|*5{c}}
T & \underline{u} & x_{0} & x_{1} & y_{1} & y_{2}\\\cline{2-6}
  & D & a & \gamma & 1 & 2\\
	& D & a & \beta  & 1 & 6\\
	& D & a & \beta  & 3 & 6
\end{array}
$$
\end{example}

Each relation $U_i^c$ encodes that each constant in $\type{x_{i}}\cap\adom{\db}$ occurs in a unique strong component of $\dbgraph{\db}$.
The meaning of the $T$-facts is as follows.
Let  $V=\{x_{0},\dots,x_{k-1}\}\cup\sequencevars{\vec{y}}$.
Let $\Theta_{D}$ be the set of all valuations over $V$ such that 
$$T(\underline{D},\mu(x_{1}),\dots,\mu(x_{k-1}),\mu(\vec{y}))$$ 
has been added by the reduction.
Then the following hold (recall $q_{0}=\bigcup_{i=0}^{k-1}\clutch{x_{i}}{q}$):
\begin{itemize}
\item
for every repair $\rep$ of $\db$,
there exists $\mu\in\Theta_{D}$ such that $\rep\models\mu(q_{0})$; and
\item
for every $\mu\in\Theta_{D}$,
there exists a repair $\rep$ of $\db$ such that 
\begin{enumerate}
\item
$\rep\models\mu(q_{0})$; and
\item
for each $\mu'\in\Theta_{D}$, if $\mu'\neq\mu$, then $\rep\not\models\mu'(q_{0})$.
\end{enumerate}
\end{itemize}

The cycles in $D$ can be found in polynomial time by solving reachability problems, as explained in~\cite[Theorem~4]{DBLP:conf/pods/Wijsen13} and~\cite{DBLP:conf/icdt/KoutrisS14}.
The crux is that the number of cycles in $\dbgraph{\db}$ of length exactly $k$ is polynomially bounded.
Any longer cycle consists of an elementary path $a_{0},a_{1},\dots,a_{k-1},a_{0}'$ of length $k$ ($a_{0}\neq a_{0}'$),
concatenated with an elementary path from $a_{0}'$ to $a_{0}$ that contains no vertex in $\{a_{1},\dots,a_{k-1}\}$.
Notice incidentally that the reduction needs to know the existence (or not) of cycles of size strictly greater than $k$ in any strong component $D$, 
but the vertices on such cycle need not be remembered. 

It can now be seen that, in general, the above reduction results in a database $\db'$ that is as in the following lemma.

\begin{lemma}\label{lem:frugalencoding}
Let $q$ and $\calC$ be as in the statement of Lemma~\ref{lem:soluble}.
Let $q^{*}=\resolve{q}{\calC}$, and let the variable $u$ be as in Definition~\ref{def:resolve}. 
Let $\db$ be an uncertain database that is input to $\cqa{q}$.
We can compute in polynomial time an uncertain database $\db'$ that is a legal input to $\cqa{q^{*}}$ such that the following hold:
\begin{enumerate}
\item\label{it:frugone}
for every repair $\rep$ of $\db$,
there exists a repair $\rep'$ of $\db'$ such that  for every valuation $\theta$ over $\queryvars{q^{*}}$,
if $\theta(q^{*})\subseteq\rep'$,
then $\theta(q)\subseteq\rep$; and
\item\label{it:frugtwo}
for every repair $\rep'$ of $\db'$,
there exists a repair $\rep$ of $\db$ such that  for every valuation $\theta$ over $\queryvars{q}$,
if $\theta(q)\subseteq\rep$,
then there exists a constant $D$ such that $\substitute{\theta}{u}{D}(q^{*})\subseteq\rep'$.
\end{enumerate}
\end{lemma}

We can now prove Lemma~\ref{lem:soluble}.

\noindent
\begin{proof}[of Lemma~\ref{lem:soluble}]
Let $\db$ be an uncertain database that is input to $\cqa{q}$.
By Lemma~\ref{lem:frugalencoding}, we can compute in polynomial time an uncertain database $\db'$ that is a legal input to $\cqa{q^{*}}$ such 
that $\db'$ satisfies conditions~\ref{it:frugone} and~\ref{it:frugtwo} in the statement of Lemma~\ref{lem:frugalencoding}.
It suffices to show that the following are equivalent.
\begin{enumerate}
\item\label{it:prosolone}
Every repair of $\db$ satisfies $q$.
\item\label{it:prosoltwo}
Every repair of $\db'$ satisfies $q^{*}$.
\end{enumerate}
\framebox{\ref{it:prosolone}$\implies$\ref{it:prosoltwo}}
Proof by contraposition.
Assume a repair $\rep'$ of $\db'$ such that $\rep'\not\models q^{*}$.
By item~\ref{it:frugtwo} in the statement of Lemma~\ref{lem:frugalencoding},
we can assume a repair $\rep$ of $\db$ such that  for every valuation $\theta$ over $\queryvars{q}$,
if $\theta(q)\subseteq\rep$,
then there exists a constant $D$ such that $\substitute{\theta}{u}{D}(q^{*})\subseteq\rep'$.
Obviously, if $\rep\models q$, then $\rep'\models q^{*}$, a contradiction.
We conclude by contradiction that $\rep\not\models q$.
\framebox{\ref{it:prosoltwo}$\implies$\ref{it:prosolone}}
Proof by contraposition.
Assume a repair $\rep$ of $\db$ such that $\rep\not\models q$.
By item~\ref{it:frugone} in the statement of Lemma~\ref{lem:frugalencoding},
we can assume a repair $\rep'$ of $\db'$ such that  for every valuation $\theta$ over $\queryvars{q^{*}}$,
if $\theta(q^{*})\subseteq\rep'$,
then $\theta(q)\subseteq\rep$.
Obviously, $\rep'\not\models q^{*}$.
\end{proof}

%% file: conclusion.tex
\section{Conclusion}

This paper settles a long-standing open question in certain query answering, by establishing an effective complexity trichotomy in the set containing $\cqa{q}$ for each self-join-free Boolean conjunctive query $q$. 
In particular, we show that, given $q$, there exists a procedure that looks at the structure of the attack graph of $q$ and decides whether $\cqa{q}$ is in $\FO$, in $\P\setminus\FO$, or $\coNP$-complete. 

The exciting question that still remains open is whether the above trichotomy can be extended beyond self-join-free conjunctive queries, to conjunctive queries with self-joins and unions of conjunctive queries.

%% file: proofs.tex
\section{Proofs for Section~\ref{sec:attackgraph}}

\subsection{Proof of Lemma~\ref{lem:trans}}

We use the following helping lemma.

\begin{lemma}
\label{lem:goback}
Let $q$ be a self-join-free Boolean conjunctive query.
Let $F,G\in q$ such that $F\attacks{q}G$.
Then, for every $x\in\keycl{F}{q}\setminus\keycl{G}{q}$,
there exists a sequence
$F_{0}, F_{1}, \dots, F_{n}$ of atoms of $q$ such that 
\begin{itemize}
\item
$F_{0}=F$; 
\item
for all $i\in\{0,\dots,n-1\}$, $\atomvars{F_{i}}\cap\atomvars{F_{i+1}}\nsubseteq\keycl{G}{q}$; and
\item
$x\in\atomvars{F_{n}}$.
\end{itemize}
\end{lemma}
\begin{proof}
Consider a maximal sequence
$$
\begin{array}{llll}
\keyvars{F} & = & S_{0}   & H_{1}\\
            &   & S_{1}   & H_{2}\\
            &   & \multicolumn{1}{c}{\vdots} & \multicolumn{1}{c}{\vdots}\\
            &   & S_{k-1} & H_{k}\\
            &   & S_{k}
\end{array}
$$
where 
\begin{enumerate}
\item $S_{0}\subsetneq S_{1}\subsetneq\dotsm\subsetneq S_{k-1}\subsetneq S_{k}$; and
\item for every $i\in\{1,2,\dots,k\}$, 
\begin{enumerate}
\item
$H_{i}\in q\setminus\{F\}$. Thus, $\FD{q\setminus\{F\}}$ contains the functional dependency $\fd{\keyvars{H_{i}}}{\atomvars{H_{i}}}$.
\item
$\keyvars{H_{i}}\subseteq S_{i-1}$ and $S_{i}=S_{i-1}\cup\atomvars{H_{i}}$.
\end{enumerate}
\end{enumerate}
Then, $S_{k}=\keycl{F}{q}$.
From $F\attacks{q}G$, it follows $G\not\in\{H_{1},\dots,H_{k}\}$.
For every $v\in S_{k}$, define $\depth{v}$ as the smallest integer $i$ such that $v\in S_{i}$.
Let $x\in\keycl{F}{q}\setminus\keycl{G}{q}$.
We define the desired result by induction on $\depth{x}$.
\paragraph{Basis: $\depth{x}=0$.}
Then the desired sequence is $F$.

\paragraph{Step: $\depth{x}=i$.}
Hence, $x\in S_{i}$ and $x\not\in S_{i-1}$.
Then, $x\not\in\keyvars{H_{i}}\subseteq S_{i-1}$ and $x\in\atomvars{H_{i}}$.
Since $H_{i}\neq G$, we have $\keyvars{H_{i}}\nsubseteq\keycl{G}{q}$,
or else $x\in\keycl{G}{q}$, a contradiction. 
Therefore, we can assume some variable $y\in\keyvars{H_{i}}\setminus\keycl{G}{q}$.
Since $y\in S_{i-1}$, we have $\depth{y}<\depth{x}$.
By the induction hypothesis,
there exists a sequence 
$F_{0}, F_{1}, \dots, F_{n}$ of atoms of $q$ such that 
\begin{itemize}
\item
$F_{0}=F$; 
\item
for all $i\in\{0,\dots,n-1\}$, $\atomvars{F_{i}}\cap\atomvars{F_{i+1}}\nsubseteq\keycl{G}{q}$; and
\item
$y\in F_{n}$.
\end{itemize}
The desired sequence is
$F_{0}, F_{1}, \dots, F_{n},H_{i}$.
\end{proof}

The proof of Lemma~\ref{lem:trans} is given next.

\noindent
\begin{proof}[of Lemma~\ref{lem:trans}]
Assume $F\attacks{q}G$, $G\attacks{q}H$, and $F\nattacks{q}H$.

Since $F\attacks{q}G$,
there exists a sequence 
$F_{0}, F_{1}, \dots, F_{n}$ of atoms of $q$ such that 
\begin{itemize}
\item
$F_{0}=F$ and $F_{n}=G$; and
\item
for all $i\in\{0,\dots,n-1\}$, $\atomvars{F_{i}}\cap\atomvars{F_{i+1}}\nsubseteq\keycl{F}{q}$.
\end{itemize}

Since $G\attacks{q}H$,
there exists a sequence 
$G_{0}, G_{1}, \dots, G_{m}$ of atoms of $q$ such that 
\begin{itemize}
\item
$G_{0}=G$ and $G_{m}=H$; and
\item
for all $i\in\{0,\dots,m-1\}$, $\atomvars{G_{i}}\cap\atomvars{G_{i+1}}\nsubseteq\keycl{G}{q}$.
\end{itemize}
Consider the path
$$
F_{0}, F_{1}, \dots, F_{n}, G_{1}, G_{2}, \dots, G_{m}
$$
where $F_{0}=F$, $F_{n}=G=G_{0}$, and $G_{m}=H$.
Since $F\nattacks{q}H$, we can assume  $j\in\{0,\dots,m-1\}$ such that $\atomvars{G_{j}}\cap\atomvars{G_{j+1}}\subseteq\keycl{F}{q}$.
Since $\atomvars{G_{j}}\cap\atomvars{G_{j+1}}\nsubseteq\keycl{G}{q}$,
we can assume $x\in\atomvars{G_{j}}\cap\atomvars{G_{j+1}}$ such that $x\in\keycl{F}{q}\setminus\keycl{G}{q}$.

By Lemma~\ref{lem:goback}, there exists a sequence 
$H_{0}, H_{1}, \dots, H_{k}$ of atoms of $q$ such that 
\begin{itemize}
\item
$H_{0}=F$; 
\item
for all $i\in\{0,\dots,k-1\}$, $\atomvars{H_{i}}\cap\atomvars{H_{i+1}}\nsubseteq\keycl{G}{q}$; and
\item
$x\in H_{k}$.
\end{itemize}
Consider the sequence
$$G_{0}, G_{1}, \dots, G_{j}, H_{k}, H_{k-1}, \dots, H_{0},$$
where $G_{0}=G$ and $H_{0}=F$.
Every two consecutive atoms in this sequence share a variable not in $\keycl{G}{q}$.
In particular, $G_{j}$ and $H_{k}$ share the variable $x$.
It follows $G\attacks{q}F$.
\end{proof}

\subsection{Proof of Lemma~\ref{lem:cycletwo}}

\noindent
\begin{proof}[of Lemma~\ref{lem:cycletwo}]
The first item is an immediate consequence of Lemma~\ref{lem:trans}.
In what follows, we show the second item.

We show that if the attack graph of $q$ contains a strong cycle of length $n$ with $n\geq 3$,
then it contains a strong cycle of some length $m$ with $m<n$.

Let $H_{0}\attacks{q}H_{1}\attacks{q}H_{2}\attacks{q}\dotsm\attacks{q}H_{n-1}\attacks{q}H_{0}$ be a strong cycle of length $n$ ($n\geq 3$) in the attack graph of $q$, where $i\neq j$ implies $H_{i}\neq H_{j}$.
Assume without loss of generality that the attack $H_{0}\attacks{q}H_{1}$ is strong.
Thus, $\FD{q}\not\models\fd{\keyvars{H_{0}}}{\keyvars{H_{1}}}$.

We write $i\oplus j$ as shorthand for for $(i+j)\mod n$.
If $H_{1}\attacks{q}H_{1\oplus 2}$,
then $H_{0}\attacks{q}H_{1}\attacks{q}H_{1\oplus 2}\attacks{q}\dotsm\attacks{q}H_{n-1}\attacks{q}H_{0}$ is a strong cycle of length $n-1$, and the desired result holds.
Assume next $H_{1}\nattacks{q}H_{1\oplus 2}$.
By Lemma~\ref{lem:trans}, $H_{2}\attacks{q}H_{1}$.
We distinguish two cases.

\paragraph{Case $H_{2}\attacks{q}H_{1}$ is a strong attack.}
Then $H_{1}\attacks{q}H_{2}\attacks{q}H_{1}$ is a strong cycle of length $2<n$.

\paragraph{Case $H_{2}\attacks{q}H_{1}$ is a weak attack.}
If $H_{1}\attacks{q}H_{0}$, then $H_{0}\attacks{q}H_{1}\attacks{q}H_{0}$ is a strong cycle of length $2<n$.
Assume next $H_{1}\nattacks{q}H_{0}$.
Then, from $H_{0}\attacks{q}H_{1}\attacks{q}H_{2}$ and Lemma~\ref{lem:trans},
it follows $H_{0}\attacks{q}H_{2}$.
The cycle $H_{0}\attacks{q}H_{2}\attacks{q}H_{2\oplus 1}\attacks{q}\dotsm\attacks{q}H_{n-1}\attacks{q}H_{0}$ has length $n-1$.
It suffices to show that the attack $H_{0}\attacks{q}H_{2}$ is strong.
Assume towards a contradiction that the attack $H_{0}\attacks{q}H_{2}$ is weak.
Then, $\FD{q}\models\fd{\keyvars{H_{0}}}{\keyvars{H_{2}}}$.
Since $H_{2}\attacks{q}H_{1}$ is a weak attack, $\FD{q}\models\fd{\keyvars{H_{2}}}{\keyvars{H_{1}}}$.
By transitivity, $\FD{q}\models\fd{\keyvars{H_{0}}}{\keyvars{H_{1}}}$, a contradiction.
This concludes the proof.
\end{proof}

\subsection{Proof of Lemma~\ref{lem:nonewcycles}}

\begin{proof}[of Lemma~\ref{lem:nonewcycles}]
Let $q'=\substitute{q}{x}{a}$.
For every $F\in q'$, there exists a (unique) atom $\widehat{F}\in q$ such that $F=\substitute{\widehat{F}}{x}{a}$.
It can be easily shown that for every $F\in q'$,
we have $\keycl{\widehat{F}}{q}\setminus\{x\}\subseteq\keycl{F}{q'}$.

Assume $F\attacks{q'}G$.
Then, there exists a witness 
$F_{0}\step{z_{1}}F_{1}\step{z_{2}}F_{2}\dots\step{z_{n}}F_{n}$
for $F\attacks{q'}G$ where $F_{0}=F$ and $F_{n}=G$.
It can now be easily seen that
$\widehat{F_{0}}\step{z_{1}}\widehat{F_{1}}\step{z_{2}}\widehat{F_{2}}\dots\step{z_{n}}\widehat{F_{n}}$
is a witness for $\widehat{F}\attacks{q}\widehat{G}$.
Therefore, if the attack graph of $q'$ is cyclic, then the attack graph of $q$ is cyclic.

The second item in the statement of Lemma~\ref{lem:nonewcycles} follows from the observation that for all $F,G\in q'$,
if $\FD{q}\models\fd{\keyvars{\widehat{F}}}{\keyvars{\widehat{G}}}$,
then $\FD{q'}\models\fd{\keyvars{F}}{\keyvars{G}}$.
\end{proof}

\section{Proofs for Section~\ref{sec:firstorder}}

\subsection{Proof of Lemma~\ref{lem:Lhard}}

\begin{proof}[of Lemma~\ref{lem:Lhard}]
We show a first-order reduction from the problem UFA (Undirected Forest Accessibility)~\cite{DBLP:journals/jal/CookM87} to $\cqa{q_{\fuxman}}$.
In UFA, we are given an acyclic undirected graph, and nodes $u,v$. 
The problem is to determine whether there is a path between $u$ and $v$. The problem is $\L$-complete, and remains $\L$-complete when the given graph has exactly two connected components. 
Moreover, we can assume in the reduction that the two connected components each contain at least one edge.

Given an acyclic undirected graph $G=(V,E)$ with exactly two connected components, and two nodes $u,v$, we construct an uncertain database $\db$ as follows:
\begin{enumerate}
\item 
for every edge $\{a,b\}$ in $E$,
the uncertain database $\db$ contains the facts $R_{\fuxman}(\underline{a}, \{a,b\})$, $R_{\fuxman}(\underline{b}, \{a,b\})$, $S_{\fuxman}(\underline{\{a,b\}}, a)$, and $S_{\fuxman}(\underline{\{a,b\}}, b)$, in which $\{a,b\}$ is treated as a constant; and
\item
$\db$ contains $R_{\fuxman}(\underline{u}, t)$ and $R_{\fuxman}(\underline{v}, t)$, where $t$ is a new value not occurring elsewhere.
\end{enumerate}
Clearly, the computation of $\db$ from $G$ is in $\FO$.

We next show that there exists a path between $u$ and $v$ in $G$ if and only if every repair of $\db$ satisfies $q_{\fuxman}$. 

Assume first that $u,v$ belong to the same connected component.
Let $\db'$ be the uncertain database that is constructed from the connected component not containing $u,v$.
Let $a_{0},b_{0},a_{1},b_{1},\dots,a_{n-1},b_{n-1},a_{n}$ be a sequence of distinct constants such that
\begin{enumerate}
\item
$a_{0}=a_{n}$ and for $0\leq i<j\leq n-1$, $a_{i}\neq a_{j}$ and $b_{i}\neq b_{j}$; and
\item
for $i\in\{0,\dots,n-1\}$, $\db'$ contains $R_{\fuxman}(\underline{a_{i}},b_{i})$ and $S_{\fuxman}(\underline{b_{i}},a_{i+1})$.
\end{enumerate}
Since $G$ is acyclic, any such sequence satisfies $n=1$.
An existing algorithm for $\cqa{q_{\fuxman}}$~\cite{DBLP:conf/pods/Wijsen13, DBLP:conf/icdt/KoutrisS14} will return that every repair of $\db'$ satisfies $q_{\fuxman}$.
Consequently, every repair of $\db$ satisfies $q_{\fuxman}$.

For the opposite implication, assume that one connected component contains $u$, and the other contains $v$.
By Lemma~\ref{lem:purified}, there exists an uncertain database $\db'$ that is purified relative to $q_{\fuxman}$ such that $q_{\fuxman}$ is true in every repair of $\db'$ if and only if $q_{\fuxman}$ is true in every repair of $\db$.
It is easy to see that if $u$ and $v$ belong to distinct connected components, then this purified uncertain database $\db'$ will be the empty database, whose only repair is the empty repair which falsifies $q_{0}$.
It follows that $q_{\fuxman}$ is not true in every repair of $\db$.
%
\end{proof}

\subsection{Proof of Lemma~\ref{lem:inFO}}

We  first show two helping lemmas.

\begin{lemma}\label{lem:dontcarebis}
Let $q$ be a self-join-free Boolean conjunctive query.
Let $X\subseteq\queryvars{q}$ and let $G\in q$ be an $R$-atom such for every $x\in X$, $G\nattacks{q}x$. 
Let $\rep$ be a repair of some database such that $\rep\models q$.
Let $A\in\rep$ be an $R$-fact that is relevant for $q$ in $\rep$.
Let $B$ be key-equal to $A$ and $\rep_{B}=\formula{\rep\setminus\{A\}}\cup\{B\}$.
Then, for every valuation $\zeta$ over $X$, if $\rep_{B}\models\zeta(q)$, then $\rep\models\zeta(q)$.
\end{lemma}
\begin{proof}
Let $\zeta$ be a valuation over $X$ such that $\rep_{B}\models\zeta(q)$.
We can assume a valuation $\zeta^{+}$ over $\queryvars{q}$ such that $\zeta^{+}[X]=\zeta[X]$ and $\zeta^{+}(q)\subseteq\rep_{B}$.
Thus, $\zeta^{+}$ extends $\zeta$ to $\queryvars{q}$.
We need to show $\rep\models\zeta(q)$, which is obvious if $B\not\in\zeta^{+}(q)$.
Assume next $B\in\zeta^{+}(q)$.
Since $A$ is relevant for $q$ in $\rep$, we can assume a valuation $\mu$ over $\queryvars{q}$ such that $A\in\mu(q)\subseteq\rep$.
Let $q'=q\setminus\{G\}$.
Let $\rep'=\rep_{B}\setminus\{B\}=\rep\setminus\{A\}$.
Since $q'$ contains no $R$-atom (no self-join),
$\zeta^{+}(q')\subseteq\rep'$ and $\mu(q')\subseteq\rep'$.
Moreover, $\zeta^{+}[\keyvars{G}]=\mu[\keyvars{G}]$, because $A$ and $B$ are key-equal.

From $\FD{q'}\models\fd{\keyvars{G}}{\keycl{G}{q}}$ and~\cite[Lemma~4.3]{DBLP:journals/tods/Wijsen12},
it follows $\zeta^{+}[\keycl{G}{q}]=\mu[\keycl{G}{q}]$.

Let $\tau$ be the complete edge-labeled undirected graph whose vertices are the atoms of $q$;
an edge between $H$ and $H'$ is labeled by $\atomvars{H}\cap\atomvars{H'}$.

Let $\tau'$ be the graph obtained from $\tau$ by cutting every edge whose label is included in $\keycl{G}{q}$.
Let $q_{G}$ be the subset of $q$ containing all atoms that are in $\tau'$'s strong component that contains $G$.
Let $q_{X}=q\setminus q_{G}$.

Let $\kappa$ be the valuation over $\queryvars{q}$ such that for every $x\in \queryvars{q}$,
$$
\kappa(x)=
\left\{
\begin{array}{ll}
\mu(x) & \mbox{if $x\in\queryvars{q_{G}}$}\\
\zeta^{+}(x) & \mbox{if $x\in\queryvars{q_{X}}$}
\end{array}
\right.
$$
We show that $\kappa$ is well defined.
Assume $x\in\queryvars{q_{X}}\cap\queryvars{q_{G}}$.
Then, there exist atoms $F'\in q_{X}$ and  $G'\in q_{G}$ such that $x\in\atomvars{F'}\cap\atomvars{G'}$.
Since $F'$ and $G'$ belong to distinct strong components of $\tau'$, it follows $\atomvars{F'}\cap\atomvars{G'}\subseteq\keycl{G}{q}$.
Consequently, $x\in\keycl{G}{q}$.
Since $\zeta^{+}[\keycl{G}{q}]=\mu[\keycl{G}{q}]$, it follows that $\mu(x)=\zeta^{+}(x)$.

Obviously, $\kappa(q)\subseteq\rep$.
Finally, we show that for every $u\in X$, $\kappa(u)=\zeta(u)$.
This is obvious if $u\in X\cap\keycl{G}{q}$.
Assume next that $u\in X\setminus\keycl{G}{q}$.
Since $G\nattacks{q}u$ by the assumption in the statement of Lemma~\ref{lem:dontcarebis},
it must be the case $u\in\queryvars{q_{X}}$, hence $\kappa(u)=\zeta^{+}(u)=\zeta(u)$.
It follows $\rep\models\zeta(q)$.
This concludes the proof.
\end{proof}

The following helping lemma extends~\cite[Lemma~B.1]{DBLP:journals/tods/Wijsen12}.

\begin{lemma}\label{lem:dontcare}
Let $q$ be a self-join-free Boolean conjunctive query.
Let $F\in q$ such that $F$ has zero indegree in the attack graph of $q$.
Let $\rep$ be a repair of some database.
Let $A\in\rep$ such that $A$ is 
relevant for $q$ in $\rep$.\footnote{Recall from Section~\ref{sec:preliminaries} that $A\in\rep$ is {\em relevant\/} for $q$ in $\rep$ if $A\in\theta(q)\subseteq\rep$ for some valuation $\theta$ over $\queryvars{q}$.}
Let $B$ be key-equal to $A$ and $\rep_{B}=\formula{\rep\setminus\{A\}}\cup\{B\}$.
Then, for every valuation $\zeta$ over $\keyvars{F}$,
if $\rep_{B}\models\zeta(q)$, then $\rep\models\zeta(q)$.
\end{lemma}
\begin{proof}
The proof is obvious if $A$ has the same relation name as $F$.
%
Assume next that relation names in $A$ and $F$ are distinct.
We can assume some atom $G\in q\setminus\{F\}$ such that $A$ has the same relation name as $G$.
Since $G\nattacks{q}F$, we have that for each $x\in\keyvars{F}$, $G\nattacks{q}x$.
The desired result then follows by Lemma~\ref{lem:dontcarebis}.
\end{proof}

Assume that a query $q$ contains an $R$-atom that has no incoming attack in the attack graph of $q$.
Paraphrasing Lemma~\ref{lem:dontcare}, if one replaces, in a repair $\rep$, some relevant fact $A$ with another fact $B$ that belongs to the same block as $A$,
then every $R$-fact of $\rep$ that was not relevant in $\rep$, will remain non-relevant in $\formula{\rep\setminus\{A\}}\cup\{B\}$. 
Notice, however, that the fact $B$ may be non-relevant in the new repair $\formula{\rep\setminus\{A\}}\cup\{B\}$.

The proof of Lemma~\ref{lem:inFO} can now be given.

\noindent
\begin{proof}[of Lemma~\ref{lem:inFO}]
Let $X=\keyvars{F}$.
Let $\db$ be an uncertain database.
Let $\rep$ be a repair of $\db$ that is $\fpro{q}{X}$-frugal.
Let $\sep$ be any repair of $\db$.
Construct a maximal sequence
\begin{equation}\label{eq:seqsr}
(\rep_{0},\sep_{0}),
(\rep_{1},\sep_{1}),
\dots,
(\rep_{n},\sep_{n})
\end{equation}
where
\begin{enumerate}
\item $\rep_{0}=\rep$ and $\sep_{0}=\sep$;
\item for every $i\in\{1,\dots,n\}$, one of the following holds:
\begin{enumerate}
\item
$\rep_{i}=\rep_{i-1}$ and $\sep_{i}=\formula{\sep_{i-1}\setminus\{A\}}\cup\{B\}$ for distinct, key-equal facts $A,B$ such that
$A\in\sep_{i-1}$, $B\in\rep_{i-1}$, and $A$ is relevant for $q$ in $\sep_{i-1}$; or
\item
$\sep_{i}=\sep_{i-1}$ and $\rep_{i}=\formula{\rep_{i-1}\setminus\{A\}}\cup\{B\}$ for distinct, key-equal facts $A,B$ such that
$A\in\rep_{i-1}$, $B\in\sep_{i-1}$, and $A$ is relevant for $q$ in $\rep_{i-1}$.
\end{enumerate}
\end{enumerate}
That is, the construction repeatedly replaces a fact that is relevant in one repair with its distinct, key-equal fact in the other repair.
The sequence~(\ref{eq:seqsr}) is finite, since the total number of distinct relevant facts distinguishes at each step.
For the last element $(\rep_{n},\sep_{n})$, it holds that the set of facts that are relevant for $q$ in $\rep_{n}$ is equal the set of facts that are relevant for $q$ in $\sep_{n}$.
It follows that for every valuation $\theta$ over $X$, 
\begin{equation}\label{eq:repnsepn}
\rep_{n}\models\theta(q)\iff\sep_{n}\models\theta(q).
\end{equation} 
By Lemma~\ref{lem:dontcare}, for every valuation $\theta$ over $X$,
\begin{eqnarray}
\rep_{n}\models\theta(q) & \implies & \rep\models\theta(q)\label{eq:repsub}\\
\sep_{n}\models\theta(q) & \implies & \sep\models\theta(q)\label{eq:sepsub}
\end{eqnarray}
From~(\ref{eq:repsub}) and since $\rep$ is $\fpro{q}{X}$-frugal, it follows that for every valuation $\theta$ over $X$,
\begin{equation}
\rep_{n}\models\theta(q)\iff\rep\models\theta(q)\label{eq:repeq}
\end{equation}
From~(\ref{eq:repeq}), (\ref{eq:sepsub}), and~(\ref{eq:repnsepn}),
it follows that for every valuation $\theta$ over $X$,
\begin{equation*}
\rep\models\theta(q)\implies\sep\models\theta(q)\label{eq:repsepconclude}
\end{equation*}
Since $\sep$ is an arbitrary repair, the desired result follows.
\end{proof}

%% file: tractability_extended.tex
\section{Proofs for Section~\ref{sec:road}}

This section contains helping lemmas and proofs that are used in the proof of Theorem~\ref{the:ptime}.

\subsection{Helping Lemmas}

\begin{lemma}\label{lem:witness}
Let $q$ be a self-join-free Boolean conjunctive query. 
Let $G\in q$ and $x,y\in\queryvars{q}$ such that $\FD{q\setminus\{G\}}\models\fd{x}{y}$ and $y\not\in\keycl{G}{q}$.
Then, there exists a sequence $G_{1},\dots,G_{n}$ of distinct atoms in $q$
such that $x\in\atomvars{G_{1}}$, $y\in\atomvars{G_{n}}$, and for every $i\in\{1,\dots,n-1\}$,
$\atomvars{G_{i}}\cap\atomvars{G_{i+1}}\nsubseteq\keycl{G}{q}$.
\end{lemma}
\begin{proof}
If $x=y$, then the desired sequence that proves the lemma is any atom that contains $x$.
In the remainder, we treat the case $x\neq y$.

Since $\FD{q\setminus\{G\}}\models\fd{x}{y}$,
we can assume a shortest sequence $F_{1},F_{2},\dots,F_{m}$ (call it $\pi$) 
that is a sequential proof of $\FD{q\setminus\{G\}}\models\fd{x}{y}$, as defined by Definition~\ref{def:sqp}.
Note that $G\not\in\{F_{1},\dots,F_{m}\}$.
It will be the case that $y$ occurs at a non-primary-key position in $F_{m}$.

The proof runs by induction on the length $m$ of the proof.

\paragraph{Basis}
If $m=1$, then the sequential proof $\pi$ is $F_{1}$ with $\keyvars{F_{1}}=\{x\}$.
Notice that $\keyvars{F_{1}}\neq\emptyset$, or else $y\in\keycl{G}{q}$, a contradiction.
The desired sequence that proves the lemma is $F_{1}$.

\paragraph{Induction}
Assume $m>1$.
Consider the last atom $F_{m}$ in $\pi$.
We have $\keyvars{F_{m}}\nsubseteq\keycl{G}{q}$, or else $y\in\keycl{G}{q}$, a contradiction.
If $x\in\queryvars{F_{m}}$, then the desired sequence is $F_{m}$.
In the remainder, we treat the case  $x\not\in\queryvars{F_{m}}$.
We can assume a variable $u\in\keyvars{F_{m}}$ such that $u\not\in\keycl{G}{q}$.
There exists an integer $k<m$ such that $u$ occurs at a non-primary-key position in $F_{k}$.
Then, 
$F_{1},F_{2},\dots,F_{k}$
contains a shortest subsequence that is a sequential proof of $\FD{q\setminus\{G\}}\models\fd{x}{u}$, where $u\not\in\keycl{G}{q}$.
By the induction hypothesis,
there exists a sequence $G_{1},\dots,G_{\ell}$ of distinct atoms in $q$
such that $x\in\atomvars{G_{1}}$, $u\in\atomvars{G_{\ell}}$, and for every $i\in\{1,\dots,\ell-1\}$,
$\atomvars{G_{i}}\cap\atomvars{G_{i+1}}\nsubseteq\keycl{G}{q}$.
The desired sequence that proves the lemma is $G_{1},\dots,G_{\ell},F_{m}$.
Notice that $u\in\atomvars{G_{\ell}}\cap\atomvars{F_{m}}$ and $u\not\in\keycl{G}{q}$.
\end{proof}

The following two lemmas are important tools for inferring attacks.

\begin{lemma}\label{lem:backpropagation}
Let $q$ be a self-join-free Boolean conjunctive query.
Let $G\in q$ and $y\in\queryvars{q}$ such that $G\attacks{q}y$. 
Let $x\in\queryvars{q}$ such that $\FD{q\setminus\{G\}}\models\fd{x}{y}$.
Then, $G\attacks{q}x$.
\end{lemma}
\begin{proof}
From $G\attacks{q}y$, it follows $y\notin\keycl{G}{q}$. 
A witness for $G\attacks{q}x$ can be obtained by concatenating
the sequence $G_{1},\dots,G_{n}$ like in the statement of Lemma~\ref{lem:witness}, where $y\in\atomvars{G_{n}}$, with a witness of $G\attacks{q}y$.
\end{proof}
%

\begin{lemma} \label{lem:attack-transitivity}
Let $q$ be a self-join-free Boolean conjunctive query.
Let $G\in q$ and $y\in\queryvars{q}$ such that $G\attacks{q}y$ and $\FD{q} \not \models \fd{\keyvars{G}}{y}$. If $\FD{q} \models \fd{x}{y}$, then $G\attacks{q}x$.
\end{lemma}
\begin{proof}
The desired result is obvious in case $x=y$.
In the remainder of the proof, we treat the case $x\neq y$.
Assume $\FD{q}\models\fd{x}{y}$.
Then, we can assume a shortest sequence 
$ F_{1},F_{2},\dots,F_{n} $
that is a sequential proof of $\FD{q}\models\fd{x}{y}$ as defined by Definition~\ref{def:sqp}.

Let $V=\formula{\bigcup_{j=1}^{n}\atomvars{F_{j}}}\cup\{x\}$.
For every $u\in V\setminus\{x\}$,
we define the {\em depth\/} of $u$, denoted $\depth{u}$, as the smallest integer $j$ such that $u\in\atomvars{F_{j}}$.
Furthermore, we define $\depth{x}=0$.
Clearly, $\depth{y}=n$.

We show next that if $G$ attacks some variable $u\in V$ with $\depth{u}>0$ and $\FD{q} \not \models \fd{\keyvars{G}}{u}$, 
then  also $G$  attacks some variable $u'\in V$ with $\depth{u'}<\depth{u}$ and $\FD{q} \not \models \fd{\keyvars{G}}{u'}$.

Assume $G\attacks{q}u$ with $\depth{u}=k>0$ and $\FD{q} \not \models \fd{\keyvars{G}}{u}$.
It must be the case that $u\in\atomvars{F_{k}}\setminus\keyvars{F_{k}}$.
Also,  $\FD{q} \not\models \fd{\keyvars{G}}{\keyvars{F_k}}$ (otherwise, $\FD{q}\models\fd{\keyvars{G}}{u}$, a contradiction). 
Then, there must be some $w\in\keyvars{F_{k}}$ such that $\FD{q}\not\models\fd{\keyvars{G}}{w}$, which implies $w\not\in\keycl{G}{q}$. 
Clearly, $\depth{w}<k$ and $G\attacks{q}w$.

It follows $G\attacks{q}x$.
\end{proof}

\subsection{Proof of Lemma~\ref{lem:reduction}}
\label{sec:saturated}

\noindent
\begin{proof}[of Lemma~\ref{lem:reduction}]
\framebox{Item~\ref{it:reductionone}}
Let $\pi=H_{1},H_{2},\dots,H_{n}$ be a shortest sequence that is a sequential proof of $\FD{q} \models \fd{x}{z}$.
Clearly, for $i\in\{1,\dots,n\}$, we have $\FD{q}\models\fd{x}{\keyvars{H_{i}}}$, hence $H_{i}\nattacks{q}x$ and $H_{i}\nattacks{q}z$,
by the assumption in the statement of Lemma~\ref{lem:reduction}.
%

Let $\db$ be an uncertain database that is the input to $\cqa{q}$.

\begin{sublemma}\label{sub:chain}
Let $a,b$ be constants.
If some $\fpro{q}{\{x,z\}}$-frugal repair of $\db$ satisfies $\substitute{q}{x,z}{a,b}$,
then for every repair $\rep_{B}$ of $\db$,
for every valuation $\theta$ over $\queryvars{q}$ such that $\theta(q)\subseteq\rep_{B}$,
if $\theta(x)=a$, then $\theta(z)=b$.
\end{sublemma}
\begin{subproof}
Let $\rep_{A}$ be a $\fpro{q}{\{x,z\}}$-frugal repair of $\db$.
Let $\theta_{A}$ be a valuation over $\queryvars{q}$ such that $\theta_{A}(q)\subseteq\rep_{A}$, and $\theta_{A}(x)=a$ and $\theta_{A}(z) = b$.
That is, $\rep_{A}\models\substitute{q}{x,z}{a,b}$.
Let  $\rep_{B}$ be a repair of $\db$ such that for some valuation $\theta_{B}$ over $\queryvars{q}$,
we have $\theta_{B}(q)\subseteq\rep_{B}$ and $\theta_{B}(x)=a$.
We need to show $\theta_{B}(z) =b$.

We show how to inductively construct a maximal sequence 
$$(p_{0},\rep_{0},\zeta_{0}), (p_{1},\rep_{1},\zeta_{1}), \dots, (p_{m},\rep_{m},\zeta_{m})$$
where for every $j\geq 0$,
\begin{enumerate}
\item\label{it:zetaone}
$\rep_{j}$ is a $\fpro{q}{\{x,z\}}$-frugal repair of $\db$;
\item
$\zeta_{j}$ is a valuation over $\queryvars{q}$ such that
$\zeta_{j}(q)\subseteq\rep_{j}$;
\item
$\zeta_{j}(x)=a$ and $\zeta_{j}(z)=b$, i.e., $\rep_{j}\models\substitute{q}{x,z}{a,b}$;
\item\label{it:prefix}
$p_{j}\in\{0,1,\dots,n\}$ and for all $i\in\{1,\dots,p_{j}\}$, $\zeta_{j}(H_{i})=\theta_{B}(H_{i})$;
\item\label{it:zetalast}
$p_{0}<p_{1}<\dots<p_{j}$.
\end{enumerate}
Intuitively, one can think of $p_{j}$ as an index in $\pi$ indicating that $\zeta_{j}$ and $\theta_{B}$ agree on all variables in $H_{1},H_{2},\dots,H_{p_{j}}$.

For the basis of the induction,  we choose $(p_{0},\rep_{0},\zeta_{0})=(0,\rep_{A},\theta_{A})$.
In this way, the above conditions are obviously satisfied for $j=0$.

For the induction step $j\rightarrow j+1$,
let $p_{j+1}$ be be the smallest integer $k$ such that $\zeta_{j}(H_{k})\neq\theta_{B}(H_{k})$.
It can be seen that $\zeta_{j}(H_{k})$ and $\theta_{B}(H_{k})$ must be key-equal.
Let $\rep_{j+1}=\formula{\rep_{j}\setminus\{\zeta_{j}(H_{k})\}}\cup\{\theta_{B}(H_{k})\}$.
By Lemma~\ref{lem:dontcarebis} and since $\rep_{j}$ is $\fpro{q}{\{x,z\}}$-frugal, it follows $\rep_{j+1}\models\substitute{q}{x,z}{a,b}$.
So there exists a valuation $\mu$ over $\queryvars{q}$ such that $\mu(q)\subseteq\rep_{j+1}$, and $\mu(x)=a$ and $\mu(z)=b$.
From $\rep_{j}\setminus\{\zeta_{j}(H_{k})\}=\rep_{j+1}\setminus\{\theta_{B}(H_{k})\}$ and $\mu(x)=\zeta_{j}(x)$,
it will be that case that $\mu(H_{i})=\zeta_{j}(H_{i})$ for all $i\in\{1,\dots,p_{j}\}$.
By the condition~\ref{it:prefix},  $\mu(H_{i})=\theta_{B}(H_{i})$ for all $i\in\{1,\dots,p_{j}\}$.
Then by our choice of $p_{j+1}$ and our construction of $\rep_{j+1}$, we have $\mu(H_{i})=\theta_{B}(H_{i})$ for all $i\in\{1,\dots,p_{j+1}\}$.
We choose $\zeta_{j+1}=\mu$.
With these choices, the above conditions \ref{it:zetaone}--\ref{it:zetalast} are satisfied for $j+1$. 

For $j=m$, we will have that $\zeta_{m}$ and $\theta_{B}$ agree on all variables in $\bigcup_{i=1}^{n}\atomvars{H_{i}}$.
Since $\zeta_{m}(z)=b$, it follows $\theta_{B}(z)=b$.
This concludes the proof of Sublemma~\ref{sub:chain}.
\end{subproof}

\begin{sublemma}\label{sub:chaincons}
Let $a,b_{1},b_{2}$ be constants such that $b_{1}\neq b_{2}$.
If $\db\models\substitute{q}{x,z}{a,b_{1}}$ and $\db\models\substitute{q}{x,z}{a,b_{2}}$,
then for every $\fpro{q}{\{x,z\}}$-frugal repair $\rep_{f}$ of $\db$, $\rep_{f}\not\models\substitute{q}{x}{a}$.
\end{sublemma}
\begin{subproof}
Assume the existence of two valuations $\theta_{1},\theta_{2}$ over $\queryvars{q}$
such that $\theta_{1}(q)\subseteq\db$,  $\theta_{2}(q)\subseteq\db$, $\theta_{1}(x)=\theta_{2}(x)=a$, and $b_{1}=\theta_{1}(z)\neq\theta_{2}(z)=b_{2}$. 
Then, there exist two repairs $\rep_{1}, \rep_{2}$ such that $\theta_{1}(q)\subseteq\rep_{1}$ and $\theta_{2}(q)\subseteq\rep_{2}$.


Assume towards a contradiction the existence of a $\fpro{q}{\{x,z\}}$-frugal repair $\rep_{f}$ of $\db$ such that $\rep_{f}\models\substitute{q}{x}{a}$.
Then, we can assume a valuation $\mu$ over $\queryvars{q}$ such that $\mu(q)\subseteq\rep_{f}$ and $\mu(x)=a$.
By Sublemma~\ref{sub:chain}, $\theta_{1}(z)=\mu(z)$ and $\theta_{2}(z)=\mu(z)$, hence $\theta_{1}(z)=\theta_{2}(z)$, a contradiction.
This concludes the proof of Sublemma~\ref{sub:chaincons}.
\end{subproof}

Construct a maximal sequence
\begin{equation}\label{eq:saturate}
\db_{0},a_{1},\db_{1},a_{2},\db_{2},\dots,a_{\ell},\db_{\ell}
\end{equation}
where $\db_{0}=\db$ and for $i\in\{1,\dots,\ell\}$,
\begin{enumerate}
\item
there exist two constants $b_{i},c_{i}$ such that $b_{i}\neq c_{i}$,
$\db_{i-1}\models\substitute{q}{x,z}{a_{i},b_{i}}$, and $\db_{i-1}\models\substitute{q}{x,z}{a_{i},c_{i}}$; and
\item
$\db_{i}=\db_{i-1}\setminus\widehat{\db}_{i-1}$,
where $\widehat{\db}_{i-1}$ is the smallest subset of $\db_{i-1}$ that includes every block $\block$ of $\db_{i-1}$ such that $a_{i}$ occurs in some fact of $\block$.
Recall from Section~\ref{sec:preliminaries} that we assume uncertain databases to be typed.
\end{enumerate}
Then, the following are equivalent:
\begin{enumerate}
\item\label{it:redufruone}
every repair of $\db$ satisfies $q$;
\item\label{it:redufrutwo}
every $\fpro{q}{\{x,z\}}$-frugal repair of $\db$ satisfies $q$; and
\item\label{it:redufruthree}
every $\fpro{q}{\{x,z\}}$-frugal repair of $\db_{\ell}$ satisfies $q$.
\end{enumerate}
Equivalence of items~\ref{it:redufruone} and~\ref{it:redufrutwo} follows from Lemma~\ref{lem:frugal}.
Equivalence of items~\ref{it:redufrutwo} and~\ref{it:redufruthree} follows from Sublemma~\ref{sub:chaincons}, using  induction on increasing $i\in\{0,\dots,\ell\}$.

Since the sequence~(\ref{eq:saturate}) is maximal, it must be that $\db_{\ell}\vmodels{q}\fd{x}{z}$.
Let $\db'$ be the database that includes $\db_{\ell}$ and such that for every valuation $\theta$,
if $\theta(q)\subseteq\db_{\ell}$, then $\db'$ contains $T^{c}(\underline{\theta(x)},\theta(z))$.
Clearly, the set of $T$-facts of $\db'$ is consistent, and the following are equivalent:
\begin{enumerate}
\item
every $\fpro{q}{\{x,z\}}$-frugal repair of $\db_{\ell}$ satisfies $q$; 
\item
every $\fpro{q}{\{x,z\}}$-frugal repair of $\db'$ satisfies $q\cup\{T^{c}(\underline{x},z)\}$; and
\item
every repair of $\db'$ satisfies $q\cup\{T^{c}(\underline{x},z)\}$.
\end{enumerate}
Finally, it can be easily seen that $\db'$ can be computed from $\db$ in polynomial time.
This concludes the proof of the first item.

\framebox{Item~\ref{it:reductiontwo}}
Define $q' = q \cup \{T^c(\underline{x},z)\}$. 
We show that for all $F,G\in q$, if $F\attacks{q'}G$, then  $F\attacks{q}G$.
For every attack $F\attacks{q'}G$, we distinguish two cases depending on $F$.

\paragraph{Case $\FD{q\setminus\{F\}}\models\fd{x}{z}$.}
Then clearly, $\keycl{F}{q}=\keycl{F}{q'}$. 
The only hard case is where a witness for the attack $F \attacks{q'}G$ contains the atom $T^c(\underline{x},z)$.
Then, $z\notin\keycl{F}{q'}$, hence $z\notin\keycl{F}{q}$.
From Lemma~\ref{lem:witness},
it follows that there exists a witness for $F\attacks{q}G$.

\paragraph{Case $\FD{q\setminus\{F\}}\not\models\fd{x}{z}$.}
Since $\FD{q}\models\fd{x}{z}$, it must be the case that every sequential proof of $\FD{q}\models\fd{x}{z}$ contains $F$.
Then $\FD{q} \models \fd{x}{\keyvars{F}}$. 
By the assumption in the statement of Lemma~\ref{lem:reduction}, $F\nattacks{q}x$ and  $F\nattacks{q}z$.
Assume towards a contradiction that a witness of $F\attacks{q'}G$ contains $T^{c}(\underline{x},z)$.
Then, since $\keycl{F}{q}\subseteq\keycl{F}{q'}$, it must be the case that $F\attacks{q}x$ or  $F\attacks{q}z$, a contradiction.
We conclude by contradiction that no witness of $F\attacks{q'}G$ contains $T^{c}(\underline{x},z)$.
Since $\keycl{F}{q}\subseteq\keycl{F}{q'}$, it follows $F\attacks{q}G$.

Assume that the attack graph of $q'$ contains a strong cycle $C$.
Since the atom $T^c(\underline{x},z)$ cannot be in $C$ (since it has no outgoing attacks),
the attack graph of $q$ contains the same cycle $C$.
It can be easily seen that $C$ is strong in the attack graph of $q$.
\end{proof}

\subsection{Proof of Lemma~\ref{lem:mostsimplified}}
\label{sec:binarized}

We first show two helping lemmas.

\begin{lemma}\label{lem:simplify}
Let $q$ be a self-join-free Boolean conjunctive query.
Let $F$ be an atom of $q$.
Let $G$ be an atom with a fresh relation name such that $\keyvars{G}=\keyvars{F}$ and $\atomvars{G}=\atomvars{F}$.
Let $q'=\formula{q\setminus\{F\}}\cup\{G\}$.
Then, 
\begin{enumerate}
\item
there exists a polynomial-time many-one reduction from $\cqa{q}$ to $\cqa{q'}$; and\footnote{We know that there exists such  a first-order reduction.
However, polynomial-time is sufficient here and allows for an easier proof.}
\item
if the attack graph of $q$ contains no strong cycle, then the attack graph of $q'$ contains no strong cycle either.
\end{enumerate}
\end{lemma}
\begin{proof}
The proof of the second item is straightforward.

For the first item, let $\db$ be an uncertain database that is input to $\cqa{q}$.
By Lemma~\ref{lem:purified}, we can compute in polynomial time a database $\db_{p}$ such that $\db_{p}$ is purified relative to $q$
and such that every repair of $\db$ satisfies $q$ if and only if every repair of $\db_{p}$ satisfies $q$.

Let $\db'$ be the uncertain database that includes $\db_{p}$ and such that whenever $\db_{p}$ contains $\theta(F)$ for some valuation $\theta$ over $\atomvars{F}$,
then  $\db'$ contains $\theta(G)$.
Notice here that $\atomvars{F}=\atomvars{G}$ and, since $\db_{p}$ is purified, whenever $A\in\db_{p}$ has the same relation name as $F$,
then there exists a valuation $\theta$ over $\atomvars{F}$ such that $A=\theta(F)$. 
It can now be easily verified that every repair of $\db_{p}$ satisfies $q$
if and only if every repair of $\db'$ satisfies $q'$. 
\end{proof}

Notice that the roles of $F$ and $G$ can be switched in the statement of Lemma~\ref{lem:simplify},
showing that $\cqa{q}$ and $\cqa{q'}$ are polynomially equivalent.

\begin{example}
If $F=R(\underline{a,x,x,y},y,z,z,b,u)$ and $G=S(\underline{x,y},z,u)$,
then $\keyvars{F}=\keyvars{G}$ and $\atomvars{F}=\atomvars{G}$.
So Lemma~\ref{lem:simplify} implies that we can replace $F$ with $G$ in the study of $\cqa{q}$.
\end{example}

\begin{lemma}\label{lem:makesimple}
Let $q$ be a self-join-free Boolean conjunctive query.
Let $R(\underline{\vec{x}},\vec{y})$ be an atom of $q$ with mode $i$.
Let $q_{0}=\{R_{1}^{c}(\underline{\vec{x}},w)$ $R_{2}^{c}(\underline{w},\vec{x})$, $S(\underline{w},\vec{y})\}$,
where $R_{1}$, $R_{2}$ are fresh relation names of mode $c$, $S$ is a fresh relation name of mode $i$, and $w$ is a variable such that $w\not\in\queryvars{q}$.
Let $q'=\formula{q\setminus\{R(\underline{\vec{x}},\vec{y})\}}\cup q_{0}$.
Then, 
\begin{enumerate}
\item\label{it:binone}
there exists a polynomial-time many-one reduction from $\cqa{q}$ to $\cqa{q'}$; and
\item\label{it:bintwo}
if the attack graph of $q$ contains no strong cycle, then the attack graph of $q'$ contains no strong cycle either.
\end{enumerate}
\end{lemma}
\begin{proof}
\framebox{Item~\ref{it:binone}}
Assume that the signature of $R$ is $\signature{n}{k}$.
Let $\db$ be an uncertain database that is input to $\cqa{q}$.
Define an injective function $h$ that maps every element in $\formula{\adom{\db}}^{k}$ to a fresh constant not occurring elsewhere.
Let $\db'$ be the database obtained from $\db$ by replacing each fact $R(\underline{\vec{a}},\vec{b})$ with the following three facts:
$$
R_{1}^{c}(\underline{\vec{a}},h(\vec{a})),
R_{2}^{c}(\underline{h(\vec{a})},\vec{a}), \mbox{and\ }
S(\underline{h(\vec{a})},\vec{b}).
$$
Since the function $h$ is injective,
the set of $R_{1}$-facts and $R_{2}$-facts of $\db'$ is consistent.
Hence, $\db'$ is a legal input to $\cqa{q'}$.
Intuitively, $R_{1}$-facts encode the function $h$,
and $R_{2}$-facts affirm that $h$ is injective.
It remains to be shown that every repair of $\db$ satisfies $q$ if and only if every repair of $\db'$ satisfies $q'$.

Define $f:\repairs{\db}\rightarrow\repairs{\db'}$ such that for every $\rep\in\repairs{\db}$,
\begin{itemize}
\item
if $\rep$ contains $R(\underline{\vec{a}},\vec{b})$,
then $f(\rep)$ contains $S(\underline{h(\vec{a})},\vec{b})$; 
\item
$f(\rep)$ contains all $R_{1}$-facts and all $R_{2}$-facts of $\db'$; and
\item
if $T$ is a relation name in $q$ such that $T\neq R$, then $f(\rep)$ contains exactly the same $T$-facts as $\rep$.
\end{itemize}
The following can be easily verified for every $\rep\in\repairs{\db}$:
\begin{itemize}
\item
$f(\rep)$ is indeed a repair of $\db'$; and
\item
$q$ is true in $\rep$ if and only if $q'$ is true in $f(\rep)$.
\end{itemize}
The desired result follows from the easy observation that $f$ is bijective.

\framebox{Item~\ref{it:bintwo}}
By a little abuse of notation, we will denote atoms by their relation name.
First, observe that $\FD{\consistent{q'}}\models\fd{w}{\sequencevars{\vec{x}}} $ and $\FD{\consistent{q'}}\models\fd{\sequencevars{\vec{x}}}{w}$.
This implies that for any atom $F\in q\setminus\{R\}$, we have $\keycl{F}{q}=\keycl{F}{q'}\setminus \{w\}$. 
Furthermore, $\keycl{R}{q}=\keycl{S}{q'}\setminus\{w\}$.

Notice that atoms $R_{1}$ and  $R_{2}$ have mode $c$, and hence have no outgoing attacks in the attack graph of $q'$.
We will now show that for all $F,G\in q\setminus\{R\}$,
\begin{itemize}
\item if $S\attacks{q'}G$, then $R\attacks{q}G$; 
\item if $F\attacks{q'}S$, then $F\attacks{q}R$; and
\item if $F\attacks{q'}G$, then $F\attacks{q}G$. 
\end{itemize} 
To this extent, assume an attack $F\attacks{q'}G$ where $F,G\in\formula{q\setminus\{R\}}\cup\{S\}$. 
We can assume a witness 
\begin{equation}\label{eq:witnessFG}
F_{0}\step{z_{1}}F_{1}\step{z_{2}}F_{2}\dots\step{z_{n}}F_{n}
\end{equation}
for $F\attacks{q'}G$ where $F_{0}=F$ and $F_{n}=G$.
We can assume without loss of generality that $1\leq i<j\leq n$ implies $z_{i}\neq z_{j}$,
and that $0\leq i<j\leq n$ implies $F_{i}\neq F_{j}$.
Moreover, since $\atomvars{R_{1}}=\atomvars{R_{2}}$,
we can assume that $R_{2}$ does not occur in~(\ref{eq:witnessFG}). 
We distinguish two cases.
\begin{description}
\item[Case $F_{0}=S$.]
Since $\{w\}\cup\sequencevars{\vec{x}}\subseteq\keycl{S}{q'}$,
we have that $R_{1}$ and $R_{2}$ do not occur in the sequence~(\ref{eq:witnessFG}),
and that $w\not\in\{z_{1},\dots,z_{n}\}$.
Then,
$R\step{z_{1}}F_{1}\step{z_{2}}F_{2}\dots\step{z_{n}}F_{n}$
is a witness for $R\attacks{q}F_{n}$.
\item[Case $F_{n}=S$.] 
It may be the case that $w\in\{z_{1},\dots,z_{n}\}$.
Then, by the form of $q_{0}$, we can assume a smallest integer $i$ such that $z_{i}\in\sequencevars{\vec{x}}\cup\sequencevars{\vec{y}}$.
Then,
$F_{0}\step{z_{1}}F_{1}\step{z_{2}}F_{2}\dots\step{z_{i}}R$
is a witness for $F_{0}\attacks{q}R$.
\item[Case $F_{0}\neq S\neq F_{n}$.] 
The only hard case is when the sequence~(\ref{eq:witnessFG}) is of one of the following forms:
$$
\begin{array}{ll} 
F_{0}\dots\step{x}R_{1}^{c}\step{w}S\step{y}\dots F_{n}, & \mbox{or}\\   
F_{0}\dots\step{y}S\step{w}R_{1}^{c}\step{x}\dots F_{n},\\   
\end{array}  
$$
where $x\in\sequencevars{\vec{x}}$ and $y\in\sequencevars{\vec{y}}$.
Then, $y\notin\keycl{F_{0}}{q'}$ and $x\notin\keycl{F_{0}}{q'}$. 
It follows $y\notin\keycl{F_{0}}{q}$ and $x\notin\keycl{F_{0}}{q}$, which implies that we can replace the subsequence $R_{1}^{c}\step{w}S$ (or $S\step{w}R_{1}^{c}$) with $R$ to obtain a witness for $F_{0}\attacks{q}F_{n}$.
\end{description}
It follows that every cycle in the attack graph of $q'$ is present in the attack graph of $q$ modulo a replacement of $S$ with $R$.

Assume that the attack graph of $q$ contains no strong cycle.
Let $C'$ be an elementary directed cycle in the attack graph of $q'$.
Let $C$ be the directed cycle in the attack graph of $q$ obtained from $C'$ by replacing $S$ with $R$. 
The attack cycle $C$ must be weak.
Then, the attack cycle $C'$ will be weak, because for every $F,G\in q\setminus\{R\}$,
\begin{itemize}
\item
if $\FD{q}\models\fd{\keyvars{F}}{\keyvars{G}}$, then $\FD{q'}\models\fd{\keyvars{F}}{\keyvars{G}}$;
\item
if $\FD{q}\models\fd{\keyvars{F}}{\keyvars{R}}$, then $\FD{q'}\models\fd{\keyvars{F}}{\keyvars{S}}$; and
\item
if $\FD{q}\models\fd{\keyvars{R}}{\keyvars{G}}$, then $\FD{q'}\models\fd{\keyvars{S}}{\keyvars{G}}$.
\end{itemize}
This concludes the proof.
\end{proof}

The proof of Lemma~\ref{lem:mostsimplified} is now straightforward.

\noindent
\begin{proof}[of Lemma~\ref{lem:mostsimplified}]
Apply the reductions of Lemmas~\ref{lem:simplify} and~\ref{lem:makesimple}.
Then repeatedly apply the reduction of Lemma~\ref{lem:reduction} until it can no longer be applied.
Notice that the reduction of Lemma~\ref{lem:reduction} consists in adding atoms of the form $T^{c}(\underline{x},z)$.
\end{proof}

\subsection{Proof of Lemma~\ref{lem:nonewstrongcycles}}
\label{subsec:preservation}


\noindent
\begin{proof}[of Lemma~\ref{lem:nonewstrongcycles}]
Assume that $k$, $x_{0},\dots,x_{k-1}$, $\vec{y}$, $q_{0}$, $q_{1}$ are as in Definition~\ref{def:resolve}.
Let $K=T(\underline{u},x_{0},\dots,x_{k-1},\vec{y})$.

Since the Markov cycle $\calC$ is premier,
we can assume an atom $F_{0}\in q$ with mode $i$ and $x\in\queryvars{q}$ such that $\keyvars{F_{0}}=\{x\}$ and $x\markovpatharg{q}x_{0}$ and $\FD{q}\models\fd{x_{0}}{x}$.

Assume that the attack graph of $q$ contains no strong cycle. 

\begin{sublemma}\label{sublem:zeroone}
$\FD{q_{0}\cup\consistent{q}}\cup\{\fd{u}{x_{0}},\fd{x_{0}}{u}\}\models\FD{q_{1}}$.
\end{sublemma}
\begin{subproof}
$\FD{q_{1}}$ is logically equivalent to
$\{\fd{u}{z}\mid z\in\queryvars{q_{0}}\}\cup\{\fd{x_{i}}{u}\mid 0\leq i\leq k-1\}$.

Let $z\in\queryvars{q_{0}}$.
Clearly, for all $i,j\in\{0,\dots,k-1\}$, $\FD{q_{0}\cup\consistent{q}}\models\fd{x_{i}}{x_{j}}$.
It is then obvious that $\FD{q_{0}\cup\consistent{q}}\models\fd{x_{0}}{z}$.
Hence, $\FD{q_{0}\cup\consistent{q}}\cup\{\fd{u}{x_{0}},\fd{x_{0}}{u}\}\models\fd{u}{z}$.

Let $i\in\{0,\dots,k-1\}$.
As argued before, $\FD{q_{0}\cup\consistent{q}}\models\fd{x_{i}}{x_{0}}$.
Hence, $\FD{q_{0}\cup\consistent{q}}\cup\{\fd{u}{x_{0}},\fd{x_{0}}{u}\}\models\fd{x_{i}}{u}$.

It follows that every functional dependency of $\FD{q_{1}}$ is logically implied by $\FD{q_{0}\cup\consistent{q}}\cup\{\fd{u}{x_{0}},\fd{x_{0}}{u}\}$.
\end{subproof}

\begin{sublemma}\label{sublem:onezero}
$\FD{q_{1}}\models\FD{q_{0}}\cup\{\fd{u}{x_{0}},\fd{x_{0}}{u}\}$.
\end{sublemma}
\begin{subproof}
Obviously, $\FD{q_{1}}\models\fd{u}{x_{0}}$,  $\FD{q_{1}}\models\fd{x_{0}}{u}$, and for every $i\in\{0,\dots,k-1\}$, $\FD{q_{1}}\models\fd{x_{i}}{\queryvars{q_{0}}}$.
Every atom of $q_{0}$ is of the form $R(\underline{x_{i}},\vec{z})$ where $i\in\{0,\dots,k-1\}$ and $\sequencevars{\vec{z}}\subseteq\queryvars{q_{0}}$.
Since $\FD{q_{1}}\models\fd{x_{i}}{\queryvars{q_{0}}}$, we have $\FD{q_{1}}\models\fd{x_{i}}{\sequencevars{\vec{z}}}$.
\end{subproof}

Sublemmas~\ref{sublem:zeroone} and~\ref{sublem:onezero} immediately lead to the following results.

\begin{sublemma}\label{sublem:fdeqbis}
$\FD{q^{*}}\equiv\FD{q}\cup\{\fd{u}{x_{0}},\fd{x_{0}}{u}\}$.
\end{sublemma}

\begin{sublemma}\label{sublem:fdeq}
For every $F\in q\setminus q_{0}$ such that $F$ has mode~$i$,
we have $\FD{q^{*}\setminus\{F\}}\equiv\FD{q\setminus\{F\}}\cup\{\fd{x_{0}}{u}$, $\fd{u}{x_{0}}\}$.
\end{sublemma}


\begin{sublemma}\label{sublem:uw}
For every $F\in q\setminus q_{0}$ such that $F$ has mode~$i$,  we have $\keycl{F}{q}=\keycl{F}{q^{*}}\setminus\{u\}$.
\end{sublemma}
\begin{subproof}
Let $F\in q\setminus q_{0}$ such that the mode of $F$ is $i$.
From Sublemma~\ref{sublem:fdeq}, it follows that $\keycl{F}{q}\subseteq\keycl{F}{q^{*}}$.
Since $u\not\in\atomvars{q}$, it follows $\keycl{F}{q}\subseteq\keycl{F}{q^{*}}\setminus\{u\}$.

The inclusion $\keycl{F}{q^{*}}\setminus \{u\}\subseteq\keycl{F}{q}$ follows from Sublemma~\ref{sublem:fdeq} and
the observation that in the computation of $\keycl{F}{q^{*}}$, the functional dependencies $\fd{x_{0}}{u}$ and $\fd{u}{x_{0}}$ are useless, 
except for inferring $u\in\keycl{F}{q^{*}}$ from $x_{0}\in\keycl{F}{q^{*}}$.
\end{subproof}

All $U_{i}$-atoms have mode $c$ and hence have no outgoing attacks in the attack graph of $q^{*}$.
The following lemma states that all attacks among atoms of $q\setminus q_{0}$ in the attack graph of $q^{*}$ are also present in the attack graph of $q$.

\begin{sublemma}\label{sublem:lunatic}
For all $F,G\in q\setminus q_{0}$, if $F\attacks{q^{*}}G$, then  $F\attacks{q}G$.
\end{sublemma}
\begin{subproof}
Let  $F,G\in q\setminus q_{0}$ such that $F\attacks{q^{*}}G$.
Then, we can assume a witness for $F\attacks{q^{*}}G$ of the following form:
\begin{equation}\label{eq:witness}
H_{0}\step{z_{1}}H_{1}
     \step{z_{2}}H_{2}\dots
		\step{z_{n}}H_{n},
\end{equation}
where $H_{0}=F$ and $H_{n}=G$.
We can assume without loss of generality that $1\leq i<j\leq n$ implies $z_{i}\neq z_{j}$,
and that $0\leq i<j\leq n$ implies $H_{i}\neq H_{j}$. 
Since $\keycl{H_{0}}{q}\subseteq\keycl{H_{0}}{q^{*}}$ by Sublemma~\ref{sublem:uw},
it follows that $\{z_{1},\dots,z_{n}\}\cap\keycl{H_{0}}{q}=\emptyset$.

If the sequence~(\ref{eq:witness}) contains no atom of $q_{1}$, then it is also a witness for $F\attacks{q}G$, and the desired result holds.
In the remainder, assume that the sequence~(\ref{eq:witness}) contains an atom of $q_{1}$.
Because of the structure of $q_{1}$, we can assume without loss of generality that $K$ is the only atom of $q_{1}$ that occurs in the sequence~(\ref{eq:witness}).
So we can assume $\ell\in\{1,\dots,n-1\}$ such that $H_{\ell}=K$.
Clearly, $z_{\ell},z_{\ell+1}\in\queryvars{q_{0}}$ and by Sublemma~\ref{sublem:uw}, $z_{\ell}, z_{\ell+1}\not\in\keycl{F}{q}$.  

For the variable $z_{\ell+1}$, there exists some $i\in\{0,\dots,k-1\}$ such that either $z_{\ell+1}=x_{i}$ or the atom $R(\underline{x_{i}}, z_{\ell+1})$ belongs to $q_{0}$. 
Since $\FD{q_{0}\cup\consistent{q}}\models\fd{x_{i}}{x_{j}}$ for all $i,j\in\{0,\dots,k-1\}$, it follows $\FD{q\setminus{\{F\}}}\models\fd{x_{i}}{z_{\ell}}$.
From $F\attacks{q}z_{\ell}$, it follows $F\attacks{q}x_{i}$ by Lemma~\ref{lem:backpropagation},
and hence $F\attacks{q}z_{\ell+1}$.
It can then be easily seen that there exists a witness for $F\attacks{q}G$.
\end{subproof}

We finally focus on attacks in the attack graph of $q^{*}$ that involve the atom $K$.

\begin{sublemma}\label{sublem:atuw}
For every $H\in q^{*}$,
if $H\attacks{q^{*}}K$,
then $H\in q\setminus q_{0}$, and both $\FD{q}\models\fd{\keyvars{F_{0}}}{\keyvars{H}}$ and $\FD{q}\models\fd{\keyvars{H}}{\keyvars{F_{0}}}$.
\end{sublemma}
\begin{subproof}
Let $H\in q^{*}$ such that $H\attacks{q^{*}}K$.
Since $U_{i}$-atoms have no outgoing attacks in the attack graph of $q^{*}$,
it must be the case that $H\in q\setminus q_{0}$.
The Markov graph of $q$ contains a directed path from $x$ to  $x_{0}$ (recall $\{x\}=\keyvars{F_{0}}$);
let $M$ be the set of variables on this path. 
We now distinguish two cases. 
\begin{itemize}
\item
If  $\keyvars{H}\subseteq M$, then clearly $\FD{q}\models\fd{\keyvars{F_{0}}}{\keyvars{H}}$.
Since $\FD{q}\models\fd{\keyvars{H}}{x_{0}}$ and $\FD{q}\models\fd{x_{0}}{\keyvars{F_{0}}}$, we obtain $\FD{q}\models\fd{\keyvars{H}}{\keyvars{F_{0}}}$. 
\item
Otherwise, $\FD{q\setminus\{H\}}\models\fd{\keyvars{F_{0}}}{z}$ for every $z\in\queryvars{q_{0}}$.
Since $H\attacks{q^{*}}K$, it must be that $H \attacks{q} z$ for some $z\in\queryvars{q_{0}}$. 
Then, $H\attacks{q}x$ by Lemma~\ref{lem:backpropagation}, and consequently $H \attacks{q}F_{0}$.
Then, it must be the case that $H$ belongs to the initial strong component of the attack graph of $q$ that also contains $F_{0}$.
Since the attack graph of $q$ contains no strong cycle, we have $\FD{q}\models\fd{\keyvars{F_{0}}}{\keyvars{H}}$ and $\FD{q}\models\fd{\keyvars{H}}{\keyvars{F_{0}}}$.
\end{itemize}
This concludes the proof of Sublemma~\ref{sublem:atuw}.
\end{subproof}

We can now complete the proof of Lemma~\ref{lem:nonewstrongcycles}.
Assume towards a contradiction that the attack graph of $q^{*}$ contains a strong cycle. 
By Lemma~\ref{lem:cycletwo}, the attack graph of $q^{*}$ contains a strong cycle of size~2.
So we can assume atoms $H_{0}, H_{1}\in q^{*}$ such that $H_{0}\attacks{q^{*}}H_{1}\attacks{q^{*}}H_{0}$, and at least one of the attacks is strong.
\begin{description}
\item[Case $H_{0}, H_{1} \in q \setminus q_{0}$.] 
By Sublemma~\ref{sublem:lunatic}, $H_{0}\attacks{q}H_{1}\attacks{q}H_{0}$.
Since the attack graph of $q$ contains no strong attack cycles,
we have $\FD{q}\models\fd{\keyvars{H_{0}}}{\keyvars{H_{1}}}$ and  $\FD{q}\models\fd{\keyvars{H_{0}}}{\keyvars{H_{1}}}$.
From Sublemma~\ref{sublem:fdeqbis}, it follows $\FD{q^{*}}\models\fd{\keyvars{H_{0}}}{\keyvars{H_{1}}}$ and  $\FD{q^{*}}\models\fd{\keyvars{H_{0}}}{\keyvars{H_{1}}}$,
contradicting that $H_{0}\attacks{q^{*}}H_{1}\attacks{q^{*}}H_{0}$ is a strong attack cycle.
\item[Case $H_{0} = K$ (the case $H_{1} = K$ is symmetrical).]
Then, $\keyvars{H_{0}}=\{u\}$. 
By Sublemma~\ref{sublem:atuw}, $H_{1} \in q \setminus q_{0}$, and both $\FD{q}\models\fd{\keyvars{F_{0}}}{\keyvars{H_{1}}}$ and $\FD{q}\models\fd{\keyvars{H_{1}}}{\keyvars{F_{0}}}$.
From Sublemma~\ref{sublem:fdeqbis} and $\FD{q}\models\fd{x_{0}}{\keyvars{F_{0}}}$, it follows $\FD{q^{*}}\models\fd{u}{\keyvars{H_{1}}}$.
From Sublemma~\ref{sublem:fdeqbis} and $\FD{q}\models\fd{\keyvars{F_{0}}}{x_{0}}$ (because there is a Markov path from $x$ to $x_{0}$), it follows $\FD{q^{*}}\models\fd{\keyvars{H_{1}}}{u}$.
But then $H_{0}\attacks{q^{*}}H_{1}\attacks{q^{*}}H_{0}$ is a weak attack cycle, a contradiction.
\end{description}
In both cases, we conclude by contradiction that the attack graph of $q^{*}$ contains no strong attack cycle.
\end{proof}

\subsection{Proof of Lemma~\ref{lem:markov_cycle}}
\label{sec:markov_exists}

We use the following helping lemma.

\begin{lemma}\label{lem:basic-step}
Let $q$ be a self-join-free Boolean conjunctive query such that
\begin{itemize}
\item
for every atom $F\in q$, if $F$ has mode $i$,
then $F$ is simple-key and $\keyvars{F}\neq\emptyset$; 
\item
$q$ is saturated; and
\item
the attack graph of $q$ contains no strong cycle.
\end{itemize}
Let $F_{0}$ be an atom of $q$ that belongs to an initial strong component of the attack graph of $q$, and let $\keyvars{F_{0}}=\{y\}$. 
Let $x\in\atomvars{q}$ such that $\FD{q}\models\fd{x}{y}$ and $\FD{q}\models\fd{y}{x}$.
Then, there exists $z\in\queryvars{q}$ with $\clutch{z}{q}\neq\emptyset$ such that  $x\markov z$ and $\FD{q}\models\fd{z}{y}$.
\end{lemma}
%
\begin{proof}
If $x\markov y$, then the desired result holds for $z=y$. 
In the remainder of the proof, we treat the case $x\not\markov y$.

Let $q_{0}$ be a minimal (with respect to $\subseteq$) subset of $q$ 
such that $\FD{\clutch{x}{q}\cup\consistent{q}\cup q_{0}}\models\fd{x}{y}$.
Obviously, $q_{0}\cap\clutch{x}{q}=\emptyset$ and $q_{0}\cap\consistent{q}=\emptyset$.
Let $p$ be a minimal (with respect to $\subseteq$) subset of 
$\clutch{x}{q}\cup\consistent{q}\cup q_{0}$
such that the atoms of $p$ can be sequentially ordered into a sequential proof 
(call it $\pi$) of $\FD{q}\models\fd{x}{y}$.
Clearly, $\pi$ must contain all atoms of $q_{0}$.

From $x\not\markov y$, it follows  $\FD{\clutch{x}{q}\cup\consistent{q}}\not\models\fd{x}{y}$.
Hence, $q_{0}\neq\emptyset$.
Let $G$ be the leftmost atom in $\pi$ such that $G\in q_{0}$.
Notice that $\keyvars{G}\neq\emptyset$ by the premise in the statement of Lemma~\ref{lem:basic-step}.
We can assume $z\in\queryvars{q}$ such that $G\in \clutch{z}{q}$.
Since $G$ is chosen leftmost, $\FD{\clutch{x}{q}\cup\consistent{q}}\models\fd{x}{z}$, hence $x\markov z$ and $\clutch{z}{q}\neq\emptyset$. 
It remains to be shown that $\FD{q}\models\fd{z}{y}$. 

Assume towards a contradiction that $\FD{q}\not\models\fd{z}{y}$.
In the next paragraph, we show that  $\pi$ contains an atom $H$ such 
that for some $w_{1},w_{2}\in\keyvars{H}$,
\begin{enumerate}
\item
$\FD{q}\models\fd{z}{w_{1}}$ but $\FD{\consistent{q}}\not\models\fd{z}{w_{1}}$; and
\item
$\FD{q}\not\models\fd{z}{w_{2}}$.
\end{enumerate}

\paragraph{Existence of $H$, $w_{1}$, and $w_{2}$.}
Let $V=\queryvars{p}\cup\{x\}$ and let the sequential proof $\pi$ be $H_{1},H_{2},\dots,H_{\ell}$.
For every $u\in V\setminus\{x\}$,
we define the {\em depth\/} of $u$, denoted $\depth{u}$, as the smallest integer $j$ such that $u\in\atomvars{H_{j}}$.
Furthermore, we define $\depth{x}=0$.
Clearly, $\depth{y}=\ell$.

For $u\in V$ and $i,j\in\{0,\dots,\ell\}$,
we write $i\stackrel{u}{\rightarrowtail}j$ if $\depth{u}=i$ and $j\in\{i+1,\dots,\ell\}$ such that $u\in\keyvars{H_{j}}$.
Intuitively, if $i>0$, then $i\stackrel{u}{\rightarrowtail}j$ says that the variable $u$ is introduced in the sequential proof by $H_{i}$, and ``used" later on by $H_{j}$. 
We can assume $k\in\{1,\dots,\ell\}$ such that $G=H_{k}$.
Clearly, $\depth{z}<k$.
It can be easily seen that the following can be assumed without loss of generality.
\begin{quote}
{\em Simple-Things-First Condition:\/}
for every $u\in V$,
if $\FD{\consistent{q}}\models\fd{z}{u}$, then $\depth{u}<k$.
\end{quote}
Since no atom of $\pi$ is redundant, there exists a sequence
$$k_{0}\stackrel{u_{1}}{\rightarrowtail}k_{1}\stackrel{u_{2}}{\rightarrowtail}k_{2}\dotsm\stackrel{u_{m}}{\rightarrowtail}k_{m}$$
where $k_{0}=k$ and $k_{m}=\ell$.
Thus, $y$ occurs at a non-primary-key position in $H_{k_{m}}$.
For all $i\in\{1,\dots,m\}$, $\depth{u_{i}}\geq k$, hence $\FD{\consistent{q}}\not\models\fd{z}{u_{i}}$ by the {\em Simple-Things-First Condition\/}.

Since $\FD{q}\not\models\fd{z}{y}$, we have $\FD{q}\not\models\fd{z}{\keyvars{H_{k_{m}}}}$.
Hence, we can assume a smallest integer $j\in\{1,2,\dots,m\}$ such that $\FD{q}\not\models\fd{z}{\keyvars{H_{k_{j}}}}$.
Then obviously, $\FD{q}\models\fd{z}{\keyvars{H_{k_{j-1}}}}$, hence $\FD{q}\models\fd{z}{u_{j}}$.
We can choose $w_{1}=u_{j}$ and $H=H_{k_{j}}$.
Further, since $\FD{q}\not\models\fd{z}{\keyvars{H_{k_{j}}}}$, we can choose $w_{2}\in\keyvars{H_{k_{j}}}$ such that $\FD{q}\not\models\fd{z}{w_{2}}$.
We conclude that $H$, $w_{1}$, and $w_{2}$ indeed exist.

Since $q$ is saturated,
from $\FD{q}\models\fd{z}{w_{1}}$ and $\FD{\consistent{q}}\not\models\fd{z}{w_{1}}$,
it follows that there exists an atom $G'\in q$ such that
$\FD{q}\models\fd{z}{\keyvars{G'}}$ and such that either $G'\attacks{q}z$ or $G'\attacks{q}w_{1}$.
Clearly, $G'$ is an atom with mode~$i$.

We show  $\FD{q\setminus\{G'\}}\models\fd{x}{z}$.
Assume towards a contradiction that $\FD{q\setminus\{G'\}}\not\models\fd{x}{z}$.
Since $\FD{\clutch{x}{q}\cup\consistent{q}}\models\fd{x}{z}$, it must be the case that $G'\in \clutch{x}{q}$, hence $\keyvars{G'}=\{x\}$.
Then, from $\FD{q}\models\fd{z}{\keyvars{G'}}$ and $\FD{q}\models\fd{x}{y}$, it follows $\FD{q}\models\fd{z}{y}$, a contradiction.
We conclude by contradiction that  $\FD{q\setminus\{G'\}}\models\fd{x}{z}$.

Two cases can occur.
\begin{description}
\item[Case $G'\attacks{q}w_{1}$.] 
Since $\FD{q}\not\models\fd{\keyvars{G'}}{w_{2}}$ (or otherwise $\FD{q}\models\fd{z}{w_{2}}$, a contradiction), we have $w_{2}\not\in\keycl{G'}{q}$, hence $G'\attacks{q}w_{2}$.  
Since $\FD{q}\models\fd{x}{w_{2}}$, it follows by Lemma~\ref{lem:attack-transitivity} that $G'\attacks{q}x$. 
\item[Case $G'\attacks{q} z$.] 
Since $\FD{q\setminus\{G'\}}\models\fd{x}{z}$, we have that $G'\attacks{q}x$ by Lemma~\ref{lem:backpropagation}.
\end{description}
Thus, at this part of the proof, we have $G'\attacks{q}x$.
We now distinguish two cases.
\begin{description}
\item[Case $\FD{q}\models\fd{\keyvars{G'}}{x}$.]
From $\FD{q}\models\fd{z}{\keyvars{G'}}$ and $\FD{q}\models\fd{x}{y}$, we have $\FD{q}\models\fd{z}{y}$, a contradiction.
\item[Case $\FD{q}\not\models\fd{\keyvars{G'}}{x}$.]
From $\FD{q}\models\fd{y}{x}$ and  $G'\attacks{q}x$, it follows from Lemma~\ref{lem:attack-transitivity} that $G'\attacks{q}y$, which implies $G'\attacks{q}F_{0}$.
Since $F_{0}$ belongs to an initial strong component of $q$'s attack graph and since the attack graph of $q$ contains no strong cycle, the attack $G'\attacks{q}F_{0}$ must be weak, so $\FD{q}\models\fd{\keyvars{G'}}{y}$.
Since $\FD{q}\models\fd{z}{\keyvars{G'}}$, we obtain $\FD{q}\models\fd{z}{y}$, a contradiction.
\end{description}
We conclude by contradiction that $\FD{q}\models\fd{z}{y}$.
\end{proof}

The proof of Lemma~\ref{lem:markov_cycle} is given next.

\noindent
\begin{proof}[of Lemma~\ref{lem:markov_cycle}]
By repeated application of Lemma~\ref{lem:trans}, the initial strong component with two or more atoms will contain two atoms $F_{0},G$ such that $F_{0}\attacks{q}G\attacks{q}F_{0}$. 

Let $\{w_{0}\} = \keyvars{F_{0}}$ (and thus $\clutch{w_{0}}{q} \neq \emptyset$) and $\{y\} = \keyvars{G}$. 
Since the attack graph of $q$ contains no strong cycle, we have $\FD{q} \models \fd{w_{0}}{y}$ and $\FD{q} \models \fd{y}{w_{0}}$.
By~Lemma~\ref{lem:basic-step}, there exists $w_{1}\in\atomvars{q}$ such that $w_{0}\markov w_{1}$, $\clutch{w_{1}}{q}\neq\emptyset$, and $\FD{q}\models\fd{w_{1}}{y}$. 
The latter implies that $\FD{q} \models \fd{w_{1}}{w_{0}}$ as well.

By repeated application of Lemma~\ref{lem:basic-step}, for every $k>0$, there exists a Markov path  $w_{0}\markov w_{1}\dotsm\markov w_{k}$, where $\clutch{w_{i}}{q}\neq\emptyset$ for every $i\in\{0,\dots,k\}$, and $\FD{q}\models\fd{w_k}{w_{0}}$.
Since $\atomvars{q}$ is a finite set, at some point we will have $w_{k} = w_{i}$ for some $i$ with $i<k$, at which point we have found the desired Markov cycle.
\end{proof}

\begin{figure}\centering
\includegraphics[scale=1.0]{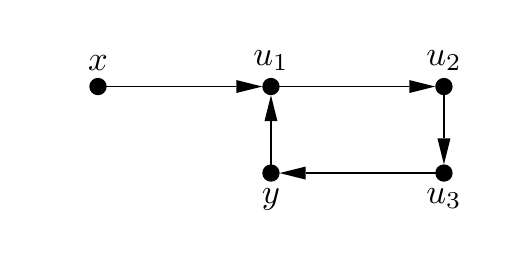}
\caption{Markov graph of the query in Example \ref{ex:mg}.}\label{fig:mg}
\end{figure}

The proof of Lemma~\ref{lem:markov_cycle} actually shows a slightly stronger result than the statement of Lemma~\ref{lem:markov_cycle}.
The proof shows that whenever $R(\underline{x},\vec{z})$ belongs to an attack cycle of size~2 that is part of an initial strong component of the attack graph, 
then the Markov graph contains a directed path from $x$ to a Markov cycle with the desired properties.
This is illustrated by the following example.
\begin{example}\label{ex:mg}
Let $q=\{R_{1}(\underline{x},u_{1})$,
$R_{2}(\underline{u_{1}},u_{2})$,
$R_{3}(\underline{u_{2}},u_{3})$,
$R_{4}(\underline{u_{3}},y)$,
$R_{5}(\underline{y},u_{1})$,
$S^{c}(\underline{u_{2},y},x)\}$.
In the attack graph of $q$,
every $R_{i}$-atom attacks every other atom of $q$, and all these attacks are weak.

The Markov graph of $q$ is shown in Figure~\ref{fig:mg}.
As predicted by the proof of Lemma~\ref{lem:markov_cycle},
for every variable among $x,y,u_{1},u_{2},u_{3}$, there is a path that starts from the variable and ends in a Markov cycle.
Notice, however, that $x$ itself is not part of a Markov cycle.
\end{example}

\subsection{Proof of Lemma~\ref{lem:gpurified}}

\noindent
\begin{proof}[of Lemma~\ref{lem:gpurified}]
Construct a maximal sequence
\begin{equation}\label{eq:gpurify}
\db_{0},\gblock_{1},\db_{1},\gblock_{2},\db_{2},\dots,\gblock_{n},\db_{n}
\end{equation}
such that $\db_{0}=\db$ and for every $i\in\{1,\dots,n\}$,
\begin{enumerate}
\item\label{it:gpurified} 
$\gblock_{i}$ is a gblock of $\db_{i-1}$ such that some repair of $\gblock_{i}$ is not grelevant for $q$ in $\db_{i-1}$;
\item
$\db_{i}=\db_{i-1}\setminus\gblock_{i}$.
\end{enumerate}
Clearly, $\db_{n}$ is gpurified relative to $q$, and by repeated application of Lemma~\ref{lem:irrelevant},
every repair of $\db$ satisfies $q$ if and only if every repair of $\db_{n}$ satisfies $q$.

It remains to be shown that $\db_{n}$ can be computed in polynomial time.
Clearly, the above sequence~(\ref{eq:gpurify}) satisfies $n\leq\card{\db}$.
The condition~\ref{it:gpurified} can be tested in polynomial time, as argued in the sequel of this proof.

First, every uncertain database that is purified relative to $q$ has at most polynomially many gblocks,
and every gblock has at most polynomially many repairs. 
Further, for any repair $\sep$ of some gblock $\gblock_{i}$, the following are equivalent:
\begin{enumerate}
\item
$\sep$ is grelevant for $q$ in $\db_{i-1}$;
\item
there exists a repair $\rep$ of $\db$ such that $\sep\subseteq\rep$ and for some valuation $\theta$ over $\queryvars{q}$ and some fact $A\in\sep$,
$A\in\theta(q)\subseteq\rep$; and
\item
$\formula{\db_{i-1}\setminus\db_{\sep}}\cup\sep\models q$,
where $\db_{\sep}$ is the subset of $\db$ that contains all facts whose relation name occurs in $\sep$.
\end{enumerate}
The first two items are equivalent by definition.
Equivalence of the last two items follows from the observation that if some atom $A\in\sep$ is relevant for $q$ in $\rep$,
then every atom of $\sep$ must be relevant for $q$ in $\rep$.
The latter test is obviously in polynomial time.
\end{proof}